\documentclass[final,12pt,notitlepage]{article}
\usepackage{amsfonts}
\usepackage{amsmath}
\usepackage{mathtools}
\usepackage{amsthm}
\usepackage{amssymb}
\usepackage[svgnames]{xcolor}
\usepackage[sc,labelsep=period,textfont=normalfont]{caption}
\usepackage[margin=1.2in]{geometry}
\usepackage[bottom,multiple,splitrule]{footmisc}
\usepackage{graphicx}
\usepackage[authoryear,round,longnamesfirst]{natbib}
\usepackage[inline,shortlabels]{enumitem}
\setlist[enumerate]{topsep=4pt,itemsep=4pt,partopsep=1ex,parsep=0pt}
\setlist[enumerate,1]{label=\emph{(\alph*)}}
\setlist[enumerate,2]{label=\emph{(\roman*)}}
\setlist[enumerate,3]{label=\emph{(\arabic*)}}
\setlist[itemize]{topsep=4pt,itemsep=4pt,partopsep=1ex,parsep=0pt}
\usepackage{hyperref}
\usepackage[onehalfspacing]{setspace}
\usepackage{titlesec}
\usepackage{sectsty}
\usepackage{dsfont}
\usepackage[capitalize,nameinlink]{cleveref}
\usepackage{xargs}
\usepackage{comment}
\usepackage{booktabs}
\usepackage[font={footnotesize},justification=justified]{caption}
\usepackage[labelformat=simple]{subcaption}
\usepackage{multirow,array}
\usepackage{tkz-euclide}
\usetikzlibrary{angles} 
\usepackage{tikz}
\usetikzlibrary{patterns, decorations.pathreplacing, arrows.meta,automata,positioning}
\usetikzlibrary{arrows}
\usetikzlibrary{calc}

\hypersetup{
pdftitle={Coalitions in Repeated Games},
pdfauthor={S. Nageeb Ali and Ce Liu},
pdfkeywords={self-enforcing, repeated game, cooperative games}
}
\hypersetup{
colorlinks,breaklinks,
citecolor=DarkBlue,urlcolor=Indigo,linkcolor=DarkBlue
}

\bibpunct[, ]{(}{)}{;}{a}{}{,}
\theoremstyle{plain}
\newtheorem{theorem}{Theorem}

\newtheorem{proposition}{Proposition}
\newtheorem{definition}{Definition}
\newtheorem{example}{Example}
\newtheorem{lemma}{Lemma}

\theoremstyle{definition}
\newtheorem{assumption}{Assumption}

\newtheorem{exampleu}[example]{Example}

\renewcommand{\hat}{\widehat}
\renewcommand{\tilde}{\widetilde}

  {\list{}{\leftmargin=0.2in\rightmargin=0.2in}\item[]}%
  {\endlist}
  {\list{}{\leftmargin=0.3in\rightmargin=0.3in}\item[]}%
  {\endlist}

\renewcommand{\paragraph}[1]{{\flushleft \textbf{\color{DarkRed}#1 }}}
\renewcommand{\succeq}{\succcurlyeq}
\renewcommand{\preceq}{\preccurlyeq}

\newcommand{\players}{N}
\newcommand{\coalitions}{\mathcal{C}}
\newcommand{\alternatives}{A}
\newcommand{\feasiblet}{\mathcal{U}}

\newcommand{\outcomesNTU}{\mathcal{O}}
\newcommand{\outcomesTU}{\overline{\mathcal{O}}}

\renewcommand{\Re}{\mathbb{R}}
\newcommand{\hchi}{\raise0.4ex\hbox{$\chi$}}

\newcommand{\abs}[1]{ \left| #1 \right| }
\newcommand{\norm}[1]{ \left| \left| #1 \right| \right| }

\newcommand{\conv}{\text{co}}

\newcommand{\feasible}{\mathcal{V}}
\newcommand{\fcr}{\feasible_{CR}}
\newcommand{\firt}{\mathcal{U}_{IR} }
\newcommand{\fcrt}{\mathcal{U}_{CR}}

\newcommand{\cf}{\underline{u}}
\newcommand{\cm}{\underline{v}^{\circ}}
\newcommand{\im}{\underline{v}}

\newcommand{\cma}{\underline{a}^{\circ}}
\newcommand{\ima}{\underline{a}}

\newcommand{\transfers}{T}

\DeclareMathOperator*{\argmin}{arg\,min}
\DeclareMathOperator*{\argmax}{arg\,max}

\newcommand{\histories}{\mathcal{H}}
\newcommand{\historiesTU}{\overline{\mathcal{H}}}
\newcommand{\seccoal}{\mathcal{S}}
\newcommand{\sincoal}{\mathcal{N}}
\newcommand{\esscoal}{\mathcal{E}}
\renewcommand{\equiv}{:=}

\newcommand{\firms}{\mathcal{F}}
\newcommand{\workers}{\mathcal{W}}
\newcommand{\payoffs}{V}

\sectionfont{\color{DarkRed}}
\subsectionfont{\color{DarkRed}}
\subsubsectionfont{\color{DarkRed}}

\newtheorem{manualtheoreminner}{Theorem}
\newenvironment{rtheorem}[1]{%
  \renewcommand\themanualtheoreminner{#1}%
  \manualtheoreminner
}{\endmanualtheoreminner}

\newtheorem{manualexampleinner}{Example}
\newenvironment{rexample}[1]{%
  \renewcommand\themanualexampleinner{#1}%
  \manualexampleinner
}{\endmanualexampleinner}

\begin{document}


\title{Coalitions in Repeated Games\thanks{We thank anonymous referees, Dan Barron, Federico Echenique, Jon Eguia, Matt Elliott, Drew Fudenberg, Ben Golub, Yingni Guo, Navin Kartik, Scott Kominers, Maciej Kotowski, Elliot Lipnowski, George Mailath, Francesco Nava, \href{refine.ink}{Refine.ink}, Ziwei Wang, Alex Wolitzky, Nathan Yoder, and various audiences for helpful comments.}}

\author{S. Nageeb Ali\thanks{Department of Economics, Pennsylvania State University. Email: \href{mailto:nageeb@psu.edu}{nageeb@psu.edu}.}\and Ce Liu\thanks{Department of Economics, Michigan State University. Email: \href{mailto:celiu@msu.edu}{celiu@msu.edu}.} }

\date{April 6, 2026}

\maketitle

\vspace{1.4\baselineskip}
\begin{abstract}

This paper proposes a framework and solution concept for repeated coalitional behavior. We model history-dependent schemes that deter coalitions from blocking using continuation promises and punishments. We evaluate the effectiveness of these schemes across a range of settings, and apply our results to repeated matching and negotiations.

\end{abstract}

\vfill

\thispagestyle{empty} 
\clearpage

 \setcounter{page}{1}
\setstretch{1.25}

\section{Introduction}

The study of repeated games models  history-dependent schemes that enable players to cooperate even if each myopically favors defection. 
This canonical approach focuses on non-cooperative play in which actions are chosen only by individuals. 
But, in many contexts, analysts have found it more useful to allow groups of players to act jointly.  For instance, matching and network theory model ``pairwise stable arrangements'' from which no pair of players can profitably deviate. Political economy models emphasize the ``Condorcet Winner,'' a policy preferred by a majority of voters to all others. More broadly, the study of cooperative games studies the ``core,'' an arrangement that no group of players would find it profitable to block. These notions are all modeled for static interactions, without harnessing the power of promises and punishments. 

A natural question is how to marry these two approaches to group behavior. In this regard, we make two contributions. First, we develop a tractable and portable framework for studying repeated interaction in matching, voting, and other coalitional games. Second, we identify settings in which dynamic incentives deter group defection and those in which they fail to do so.

We illustrate our framework using the \emph{Roommates Problem}. Ann, Bella, and Carol decide who will room together. The hitch is that only two people can share a room, leaving at least one person out. Each person prefers to have a roommate, and each has a favorite; \Cref{Table-Roommates} depicts their payoffs.
\begin{table}[h]\centering
	\begin{tabular}{cccc} 
	\toprule   
    & {Ann}  & {Bella}  & {Carol}  \\ 
Ann & $1$ & $3$ & $2$ \\ 
Bella & $2$ & $1$ & $3$ \\ 
Carol & $3$ & $2$ & $1$ \\ \bottomrule
 \hline
\end{tabular}
\caption{Payoffs of Row Player from matching with Column Player (or remaining unmatched).}\label{Table-Roommates}\vspace{-.1in}
\end{table}

As is well known, no arrangement is pairwise stable. For instance, were Ann and Bella paired as roommates, Bella and Carol would each gain if they defected and roomed together instead. Our point of entry is to see how punishment and rewards can ``solve'' this problem.

Suppose that instead of a one-time decision, the trio made choices monthly. Each accrues the flow payoffs described above, and weights them by the per period  discount factor $\delta$. As in repeated games, no player can commit to her future behavior on- or off-path. What long-run stable matches can be supported through continuation play? 

\Cref{Figure-Automaton-Roommate} depicts a stable scheme. On the path of play, Ann and Bella room together each month, leaving Carol out. While Bella and Carol might like to defect and share a room this month, the scheme assures that they refrain from doing so if $\delta$ exceeds $1/2$. For Bella anticipates that after the deviation, starting from next month, Ann and Carol would room together and she would then be left out. Her short-term gain from rooming with Carol would not offset her long-term loss, since $(1-\delta)3+\delta(1)\leq 2$. Moreover, the punishment is itself credible because the prescription following every history, including those off-path, is self-enforcing. 

\begin{figure}[t]
\begin{center}
\includegraphics[width=3.5in]{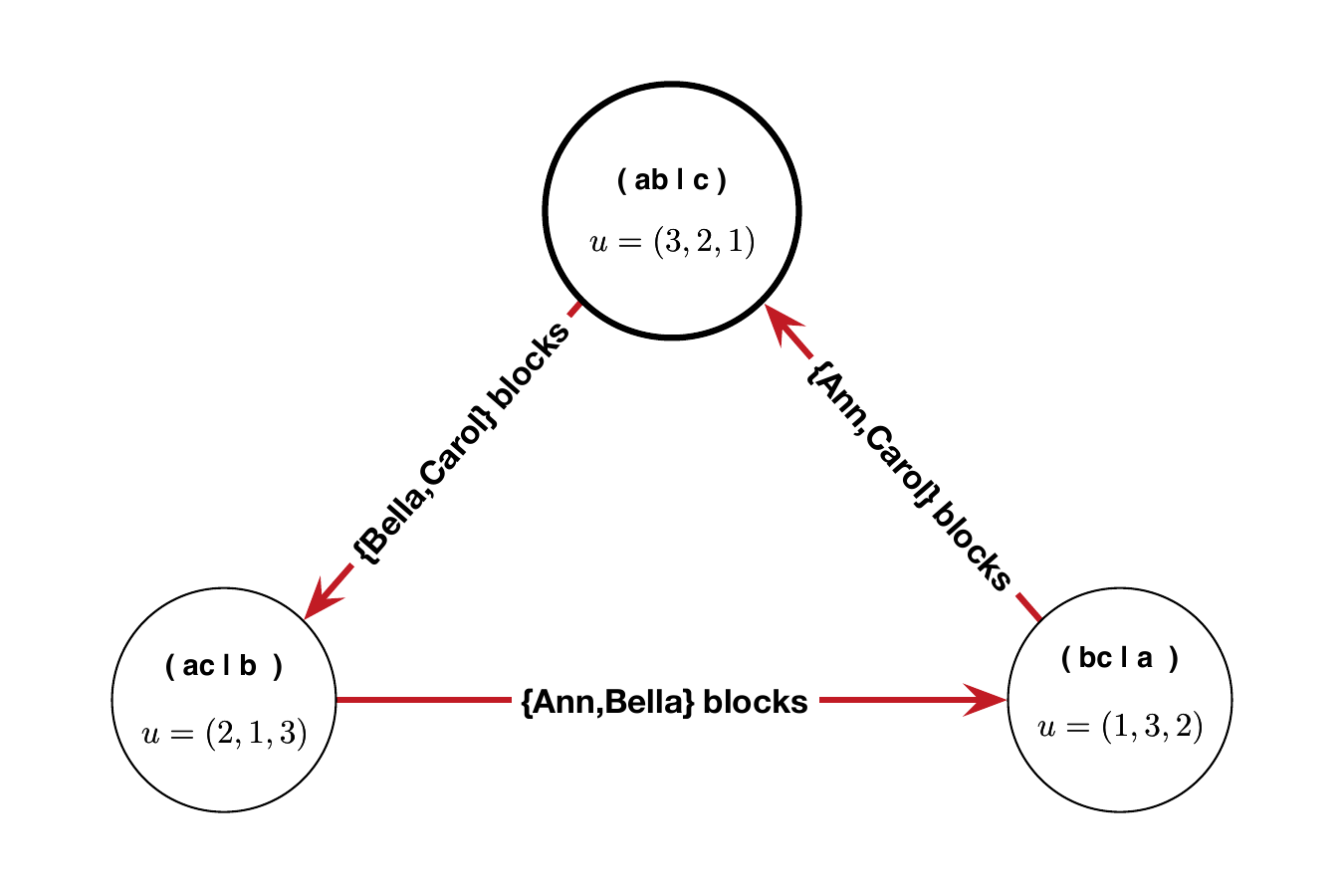}	
\vspace{-0.2in}
\caption{A perfect coalitional equilibrium for the Roommates Problem if $\delta\geq 1/2$.}	
\label{Figure-Automaton-Roommate}
\end{center}\vspace{-.3in}
\end{figure}

We model schemes of this form in general games. We consider the repeated play of an abstract stage game that accommodates many settings: (i) a strategic form game in which players choose action profiles, (ii) a characteristic function game in which players obtain payoffs based on a partitional structure, (iii) matching games with and without transferable utility, and (iv) voting games in which different coalitions have varying power to push through policies. In this repeated game, we propose history-dependent schemes such that no coalition profits from blocking at any history given how it affects continuation play. We call such schemes \emph{perfect coalitional equilibria} (PCE). 

PCEs are inherently recursive. In strategic form games, PCE refines subgame perfect equilibria. More generally, PCE offers a tractable way to model how continuation promises shape coalitional behavior; PCE's recursive nature implies that its payoff set can be obtained via self-generation approaches developed by \cite{abreu1990toward}. 
In some settings, it can be even simpler: all PCE-supportable payoffs can be supported by stationary PCE if the stage game exhibits \emph{default-independent power}. This property holds in the characteristic function games studied in cooperative game theory as well as in matching and voting models.

Our main results characterize when continuation promises and punishments deter coalitional blocking. We offer conditions under which history dependence thwarts coalitional deviations so that the set of PCE-supportable payoffs is large. Conversely, we also identify settings in which coalitional deviations choke the possibility for cooperation, resulting in ``anti-folk'' theorems. Underlying our results is a simple principle: a coalition can withstand punishments if and only if its members have highly aligned interests. Then, all members of that coalition enjoy a high minmax, considerably above their individual minmax. However, if there is any wedge in members' incentives, player-specific punishments can splinter coalitions. Then, the set of PCE-supportable payoffs virtually coincides with those that are feasible and \textit{individually} rational. 

Building on this principle, we explore features that align coalitions' interests. One such feature, we find, are \emph{strongly symmetric} schemes in which players behave symmetrically after every history. Given their tractability, these schemes feature often in the study of repeated games, such as grim-trigger punishments used to support cooperation in a repeated prisoner's dilemma or to sustain collusion among oligopolists. Although these schemes constitute subgame perfect equilibria, we show that they typically are not PCE. The reason is that once players' incentives are aligned, punishments are no longer credible. By contrast, schemes that feature asymmetric play off-path can credibly deter blocking coalitions and support a larger payoff set.

We also study whether transferable utility aligns incentives within a coalition. One might have expected the answer to be yes: if a coalition achieves a net gain by blocking, transferable utility allows it to distribute those gains among its members. However, we show that if all transfers are publicly observed, a PCE can break coalitions apart by conditioning its continuation play on who pays whom. Only if a coalition can make transfers ``secretly''---that is, without the transfers being publicly observed---can it entirely align its incentives. Such a coalition is then difficult to punish, which then limits what a PCE can enforce.

We consider a detailed application of our results to repeated labor-market matching, building on \cite{kelsocrawford1982}'s seminal work. Here, we study the kinds of matches and wages that can be supported through repeated play. It turns out that the set of supportable outcomes hinges crucially on the transparency of past wages. Under wage transparency, a vast range of outcomes can be supported, enabling workers or firms to capture much of the surplus. By contrast, if wage terms are observed only by the employer and employee, there is a complete collapse of intertemporal incentives: the supportable payoff set reduces to the core of the stage game. As to who then benefits from wage transparency depends on economic primitives. Workers benefit if they are plentiful or their marginal returns fall quickly. Absent transparency, competitive forces bid down their wages; by contrast, wage transparency enables them to enforce higher wages through collective-bargaining schemes. If workers are scarce or their marginal returns fall slowly, then it is firms who profit from wage transparency because that enables them to collusively suppress wages. Thus, our application highlights a new dimension to the debate on wage transparency, connecting it to the side of the market that it empowers to collude. 

In the Supplementary Appendix, we consider a second application to repeated negotiations when some players have veto power. The core of this stage game entails that veto players capture the entire surplus, making them de facto dictators. Against that backdrop, we show that history dependence can promote egalitarianism but that these schemes destabilize once players can make secret side-payments. 

This paper proceeds as follows. \Cref{Section-Model} describes the basic framework. \Cref{Section-NTU} identifies structural properties of PCE and characterizes its payoff set. \Cref{Section-Transfers} studies the game augmented with transfers. \Cref{Section-Applications} applies our results to labor market matching. \Cref{Section-Discussion} concludes. All proofs are in appendices. The remainder of this introduction briefly discusses the related literature. 

\paragraph{Related Literature:} We build on results in repeated games, particularly \cite{fudenberg1986folk}, \cite{abreu1994folk}, and \cite{wen1994folk}; see \cite{mailathrepeated} for a textbook treatment. Our result on how secret side-payments can undermine intertemporal incentives also aligns with the findings of  \cite{barron2021use}, who model how it disrupts collective punishment in relational contracting.

Several precursors model coalitional deviations in repeated strategic form games. \cite{aumann1959acceptable} and \cite{rubinstein1980strong} study Strong Nash and Strong Perfect Equilibria of infinitely repeated games. These concepts presume that deviating players can commit to long-term plans, even if those plans are not credible. \cite{demarzo1992coalitions} adds interim incentive constraints and identifies benefits of scapegoat policies in finite horizon games. More closely related, \cite{greenberg1990theory} develops a general theory of social situations, encompassing strategic form, extensive form, and characteristic function games. In Chapter 9 of this book and \cite{greenberg1989application}, he applies this theory to repeated strategic form games and defines a set-valued concept,  \textit{Conservative Stable Standard of Behavior} (CSSB). 
Intuitively, a CSSB models worst-case decisionmaking: a player contemplates deviating only if she is guaranteed to gain for every continuation play allowed by the standard, and a CSSB describes behavior that is immune to such deviations.
\citeauthor{greenberg1989application} focuses on individual deviations and finds that the largest ``nondiscriminating'' CSSB coincides with the set of paths supported by subgame perfect Nash equilibria. He also proposes an extension that permits coalitional deviations but does not obtain general results for this setting. Instead, he illustrates using a specific common-interest game that every nondiscriminating ``coalitional'' CSSB selects the efficient action profile. This prediction matches PCE's for this game, as seen in our \Cref{Theorem-NTU}. While we formulate our solution concept differently---our approach is rooted in rational expectations rather than conservatism---we establish a connection between CSSB and PCE in \cite{aliliu2026note}. Therein, we show that, for every discount factor, the largest nondiscriminating CSSB for \citeauthor{greenberg1989application}'s coalitional repeated game coincides with the set of paths supportable by a PCE.

Outside of strategic form games, \cite{BS2009} develop the notion of \textit{Dynamic Condorcet Winners} (DCW) for repeated elections; our solution concepts coincide in that setting. While they characterize some properties of DCW, they do not identify its limits. Hence, apart from our more general focus, many issues central to our study do not feature in their work.

Like the work above, our baseline approach takes as given the exogenously specified blocking technology in a cooperative game and studies its implications for repeated play. A different route would instead demand that blocking be coalition-proof in the sense of \cite{bernheim1987coalition}. They define inductively a solution concept for one-shot and finite-horizon games that is immune to ``credible'' coalitional deviations, where credibility requires that such deviations themselves be immune to credible sub-coalitional deviations. In the Supplementary Appendix, we develop the notion of \textit{perfect coalition-proof equilibrium} (PCPE) that combines ideas of PCE and coalition proofness. Because coalition proofness imposes a more demanding criterion for which deviations are admissible, every PCE is necessarily a PCPE. However, we show that the two concepts result in identical coalitional minmaxes and that the sets of PCE and PCPE payoffs coincide as $\delta\rightarrow 1$.

Also related is the work on renegotiation-proofness initiated by \cite{bernheim1989collective} and \cite{farrell1989renegotiation}; also see  \cite{miller2013theory} and \cite{safronov2018contestable}. Renegotiation embodies the idea that players can collectively replace an agreed upon continuation with another that they all prefer. A PCE, by contrast, studies the complementary question of when coalitions refrain from profitably blocking \emph{given} the continuation plan. Like \cite{greenberg1989application}, our approach presupposes that players have a shared understanding of future play, can form temporary agreements to block, but cannot coordinate a shift in their beliefs about continuation play. This shared understanding must be internally consistent in that players can see that no coalition can profitably block at any future stage, given the continuation. An implication of \Cref{Theorem-NTU} is that if players' interests are fully aligned---i.e., a common-interest game---then the unique PCE selects the efficient action profile, a prediction  shared by \citeauthor{farrell1989renegotiation}'s strong renegotiation-proofness. 
However, the two notions diverge when interests are misaligned. A further difference is that the renegotiation literature largely emphasizes renegotiation by the \textit{grand} coalition alone, whereas our results stem from allowing other coalitions to deviate. Two exceptions in this regard are \cite{genicotray2003} and \cite*{mailath2024trust}, who evaluate risk-sharing arrangements that are immune to subcoalitional renegotiation.

Our work connects to that on farsighted coalition formation. This literature has two interrelated strands. One evaluates farsighted coalitional negotiations over behavior within a static game; see, for example, \cite{chwe1994farsighted}, and \cite{ray2015farsighted}. A more closely related strand, exemplified by \cite{KR2003} and \cite{gomes2005dynamic}, studies dynamic games in which payoffs accrue in real time and coalitions evaluate moves based on discounted continuation values. Our work makes two departures from this strand. One is that we study a repeated game, in which the past behavior has no direct payoff relevance for the future; by contrast, this literature studies stochastic games in which the alternative chosen in period $t$ is the designated default for period $(t+1)$. The second is that we model the implications of history dependence whereas the prior work largely concerns Markov perfect behavior. An exception is \cite{vartiainen2011dynamic} who models history-dependence with limit-of-means discounting, and his focus is on existence of deterministic absorbing processes.

Dynamic considerations feature in studies of matching in which players account for future play when deciding with whom to match, e.g., \cite{corbae2003directed},  \cite{damiano2005stability}, \cite{kadam2018multiperiod,kadam2018time}, \cite{doval2022dynamically}, and \cite{kotowski2024}. \cite{rostek2024matching} propose a stability notion for static matching in which similar considerations emerge from players' strategic thinking. 

Our application to labor market matching and wage transparency contribute to a growing literature, surveyed in \cite{cullen2024jep}. In a bargaining model with incomplete information, \cite{cullen2023equilibrium} show that wage transparency disadvantages workers. We offer a complementary perspective, studying the implications of wage transparency for the repeated game. Beyond this application, we view incorporating long-run incentives in \cite{kelsocrawford1982}'s workhorse model to be of independent interest. In the static setting, this framework has been enriched in various directions \citep[e.g.,][]{hatfield2005matching} but relatively little is known about how carrots and sticks affect these matching markets. 

Since our initial draft, \cite{bardhi2024early} use our approach to model stability in labor markets in which firms learn about workers' types, evaluating how early-career discrimination can result in persistent wage gaps. \cite{liu2023stability} and \cite{liu2024self} also apply the framework here to study repeated matching between long-run firms and short-run workers. The former studies matching without transfers and shows that repetition overturns the Rural Hospital Theorem; the latter shows that repetition in a setting with transferable utility can lead to stable solutions despite peer effects and complementarities. Neither of these papers considers matching between long-run players nor the effect of wage transparency.

\vspace{-.15in}
\section{Model}\label{Section-Model}
Players $\players \equiv \{1,2, \ldots, n \}$ interact in a repeated game at $t=0,1,2,\ldots$. A coalition is a nonempty subset of $\players$, and we denote the set of coalitions by $\coalitions \equiv 2^{\players}\backslash\{\emptyset\}$. We also denote the set of singleton coalitions by $\sincoal$.

\paragraph{The Stage Game.}
In each period, the players collectively choose an \emph{alternative} $a$ from  $\alternatives$, a compact metrizable space. The alternative $a$ generates payoff vector $v(a)\equiv (v_1(a),\ldots,v_n(a))\in \Re^{n}$ for players, where the mapping $v:\alternatives\rightarrow \Re^n$ is continuous.

Given an alternative, a coalition of players can choose to block or participate in it. Our specification allows for blocking by either a single coalition or multiple disjoint coalitions, but the former is what matters when evaluating stability. If a generic coalition $C$ alone blocks alternative $a$, then it can choose any alternative in $E_C(a)$. The correspondence $E_C:\alternatives\rightrightarrows  \alternatives$ is $C$'s \emph{effectivity correspondence} and offers a standard approach to model coalitional power.\footnote{The use of effectivity functions to model the cooperative game dates back to \cite{rosenthal1972cooperative}, \cite{moulin1982cores}, \cite{greenberg1989application,greenberg1990theory}, and \cite{chwe1994farsighted}.}
For every coalition $C$, $E_C(\cdot)$ is continuous, compact-valued, and reflexive  (i.e., $a\in E_C(a)$). In our analysis, we also assume that larger coalitions can do more: for each alternative $a$, $E_{C'}(a)\subseteq E_C(a)$ for $C'\subseteq C$. This assumption is for notational convenience; we detail in footnotes, when necessary, how to adapt notation if this assumption fails.

We interpret alternatives as a description of behavior, for instance, an action profile in a strategic form game, the matches that form in a matching game, or the partition in a characteristic function game. Given an alternative $a$ that players anticipate will occur, if a coalition $C$ blocks, the set $E_C(a)$ reflects the set of alternatives it can freely choose \textit{given the behavior of others}.

To see what this formulation captures, we revisit the Roommates Problem described in the introduction. An alternative is a rooming arrangement, and the set of alternatives, $A=\{ab|c,bc|a,ac|b,a|b|c\}$ is that of all arrangements, where $ij|k$ denotes $i$ and $j$ rooming together leaving $k$ out. For the alternative that puts Ann and Bella together, Bella could block as an individual and choose an arrangement in  $E_{\{\text{Bella}\}}(ab|c)= \{ab|c,a|b|c\}$. This specification is tantamount to her only choice \emph{as an individual} being whether to accept or reject Ann as a roommate. Carol has even less power---$E_{\{\text{Carol}\}}(ab|c)= \{ab|c\}$---because she cannot room with someone else without that player's consent. But by teaming up and blocking as a pair, Bella and Carol could choose any alternative in $E_{\{\text{Bella},\text{Carol}\}}(ab|c)= \{bc|a,ab|c,a|b|c\}$ where the first element denotes the pair rooming together. 

The Roommates Problem is a specific illustration but the abstract form is considerably more general. We formalize below how to embed strategic form, voting, and characteristic function games in this setup.

\begin{example}\label{Example-NormalFormGame}
Consider a strategic form game in which player $i$'s action set, $\alternatives_i$, is compact: $A_i$ can be either the set of pure actions or the set of mixtures over finite actions. The set of alternatives comprises all action profiles $\alternatives\equiv \times_{i=1}^n \alternatives_i$. The effectivity correspondence is $E_C(a)\equiv\left\{a'\in \alternatives: a'_j=a_j\text{ for all }j\notin C\right\}$, modeling the possibility for a blocking coalition to choose action profiles in which players outside the coalition do not change their actions. This formulation extends the standard definition for individual deviations that are used to define Nash equilibria.\footnote{We note two conceptual considerations. First, this example assumes that when coalition $C$ blocks action profile $a$, the actions of players outside the coalition, $a_{-C}$, are held fixed. Our general formulation is agnostic about what players outside coalition $C$ do; alternative specifications of how these players respond correspond to alternative specifications of $E_C(\cdot)$. Second, we  do not require that the actions chosen by coalition $C$ be individual best responses for each of its members. As we describe later, we develop a coalition-proof variant of our solution concept in the Supplementary Appendix that imposes stronger requirements on this margin.} 
\end{example}

\begin{example}\label{Example-Voting}
Consider majority voting, as in \cite{BS2009}. Let $\mathcal W$ be the set of coalitions that have at least $\left \lceil{\frac{n}{2}}\right \rceil $ players. The effectivity correspondence specifies that for every $a$, $E_C(a)= \alternatives$ if $C\in \mathcal W$, and $E_C(a)=\{a\}$ otherwise. 
\end{example}

\begin{example}\label{Example-CoalitionalGame}
Consider a characteristic function game $(N,U)$, where for each coalition
$C \subseteq N$, the set $U(C) \subseteq \mathbb{R}^{|C|}$ specifies the payoff
vectors that coalition $C$ can feasibly secure if it forms. An alternative is a
pair $a=(\pi,u)$, where $\pi$ is a partition of $N$ and $u \in \mathbb{R}^n$
is a payoff vector satisfying $u_C \in U(C)$ for every coalition $C \in \pi$.
Payoffs are projections of the alternative onto its second coordinate space, that is, $v((\pi,u))=u$. For every coalition $C$ and alternative $a$, the effectivity correspondence $E_C(a)$ specifies the set of alternatives to which coalition $C$ may move.\footnote{Special cases of this example include the Roommates Problem and two-sided matching.}
\end{example}

\paragraph{Outcomes, Histories, and Plans.}
We develop our notation recursively. A plan specifies a default alternative, say $a$, at the beginning of period $t=0$. This default is chosen if no coalition blocks it, in which case we record the stage-game outcome as $(a,\emptyset)$. If coalitions $\{C_1, \ldots, C_k\}$ block the default, and their moves result in the alternative $a'$, we record the stage-game outcome as $(a', \{C_1, \ldots, C_k\})$. Based on the outcome at $t=0$, a plan specifies a default at $t=1$, and the game continues recursively.\footnote{We assume that coalitional blocking is directly observable so as to hew closely to repeated games with perfect monitoring, which we view to be the natural starting point. In some settings, the identity of a blocking coalition is implied by the chosen alternative. But, in other settings (e.g., matching), the chosen alternative alone might not be enough to encode who initiated the block. For instance, in the Roommates Problem, if the alternative $ab|c$ is blocked and we only record that $a|b|c$ is chosen instead, it does not distinguish which of Ann and Bella blocked.}

Proceeding abstractly, let $\mathcal{P}$ be the set of all partitions over players, and define $\mathcal{B}\equiv \{B\subseteq 2^{\players}:  B \subseteq \pi \text{ for some }\pi\in \mathcal{P}\}$, so that each $B$ in $\mathcal B$ is a collection of disjoint coalitions. We denote the set of stage-game outcomes by $\outcomesNTU\equiv \alternatives\times \mathcal{B}$; an outcome specifies an alternative $a$ and a collection of disjoint blocking coalitions. At the beginning of period $t$, the history $h \equiv (a^\tau, B^\tau)_{\tau=0}^{t-1}$ records the stage-game outcomes up to the start of period  $t$. We denote the set of all $t$-period histories by $\histories^t$ for $t\geq 1$. The set of all histories is $\histories\equiv \bigcup_{t=0}^\infty \histories^t$, where $\histories^0=\{\emptyset\}$. A \emph{plan} $\sigma:\histories\rightarrow \alternatives$ specifies a default alternative following each history.

\paragraph{Payoffs.}
A {path} $(a^t)_{t=0,1,2,\ldots}$ is an infinite sequence  of alternatives; from that path, player $i$ accrues a normalized discounted payoff of $(1-\delta) \sum_{t=0}^{\infty} \delta^t v_i(a^t)$, in which $\delta$ in $[0,1)$ is a common discount factor.\footnote{In this repeated game, the history affects future payoffs only through the continuation plan; the alternative chosen in period $t$ has no direct bearing on future payoffs or the set of feasible alternatives. In this way, our setup departs from the stochastic games studied by \cite{KR2003} and \cite{gomes2005dynamic} in which the alternative chosen at the end of period $t$ becomes the default in period $(t+1)$ and thereby directly restricts the set of feasible alternatives in that period.} After a history $h$, a plan $\sigma$ recursively specifies a path of default alternatives and we denote by $U_i(h|\sigma)$ player $i$'s payoff from that path.\footnote{To be explicit, after a period-$t$ history $h$, the plan $\sigma$ specifies a default alternative $\sigma(h)$ for period $t$. If no coalition blocks this default, the resulting period-$(t+1)$ history is $(h,\sigma(h),\emptyset)$, where $\emptyset$ indicates that no blocking occurred in period $t$. Given this extended history, the plan then specifies the next default alternative $\sigma(h,\sigma(h),\emptyset)$ for period $(t+1)$. Iterating this construction yields the path of default alternatives induced by plan $\sigma$ after history $h$: $(\sigma(h),\sigma(h,\sigma(h),\emptyset),\ldots)$.}

\paragraph{Solution Concept.}
Before describing our solution concept, we restate the ``static {core}'' in the language of this model. In the stage game, coalition $C$ profitably blocks alternative $a$ if there exists $a'\in E_C(a)$ such that $v_i(a') > v_i(a)$ for all $i\in C$. An alternative $a$ is a \textit{core alternative} if it cannot be profitably blocked by any coalition. A payoff vector $\tilde{v}$ is in the \textit{core} if there exists a core-alternative $a$ such that $\tilde{v}=v(a)$. For example, in the Roommates Problem, $ab|c$ fails to be a core alternative because Bella and Carol together can profitably block it.

We build on this notion in the repeated game: when players contemplate blocking the alternative $\sigma(h)$ specified by plan $\sigma$ at history $h$, they care not only about their instantaneous payoffs but also about how their choices today affect future behavior. 
\begin{definition} \label{Definition-PCE1}
Coalition $C$ \textbf{profitably blocks} plan $\sigma$ at history $h$ if there exists $a'\in E_C(\sigma(h))$ such that for all $i\in C$, 
\begin{align*}
    {(1-\delta)} v_i(a') \;{+} \;{\delta}{U_{i}(\, h, (a',\{C\}) \mid \sigma)} > {U_i(h\mid \sigma)}.
\end{align*}
\end{definition}

\begin{definition} \label{Definition-PCE2}
A plan $\sigma$ is a \textbf{perfect coalitional equilibrium} (PCE) if it cannot be profitably blocked by any coalition {at any history}.	
\end{definition}

A PCE is a ``self-enforcing'' plan in that, given the continuation play, no coalition finds it profitable to block.\footnote{An alternative definition of profitable blocking might stipulate that each coalition member gains weakly and at least one does so strictly. This modification would not affect our results.} In the language of self-generation, the alternative specified at each history is \emph{enforced} by continuation promises that themselves are credible given that the requirement is imposed at every history (including those off-path). Thus, the continuation of a PCE at every history must itself be a PCE. This recursive form implies that the set of PCE supportable payoffs is amenable to dynamic programming \`a la \cite{abreu1990toward}. If $\delta=0$, PCEs implement only core alternatives of the stage game. Hence, a PCE may not exist if the stage game has an empty core  and players are sufficiently myopic. If the core is nonempty, then a plan that specifies a core alternative $a^*$ after every history is a PCE for every $\delta\geq 0$.\footnote{The Bondareva-Shapley Theorem establishes that a cooperative game has a nonempty core if it is balanced. Thus, under these conditions, a PCE is also guaranteed to exist for every $\delta\geq 0$.} 

Before characterizing PCE, we briefly compare it to other solution concepts. Fundamentally, a PCE captures the idea that coalitions can block through temporary agreements but coalition partners cannot commit in advance to the agreements they will sign in the future. By contrast, Strong Nash Equilibrium \citep{aumann1959acceptable} and Strong Perfect Equilibrium \citep{rubinstein1980strong} allow deviating coalitions to commit to a strategy in the repeated game, including how they will behave in every future period. Relative to those concepts, PCE therefore takes a conservative stance on the agreements that blocking coalitions can form. 

From a different perspective, however, PCE may endow deviating coalitions with too much power, since a blocking coalition need not worry about further deviations by a subcoalition in the same period. To address this concern, we formalize perfect coalition-proof equilibrium (PCPE) in the Supplementary Appendix,  combining the ideas above with coalition proofness \citep{bernheim1987coalition}. Although the resulting solution concept is distinct, we find that it yields identical coalitional minmaxes and supports the same payoff set when players are patient.

\section{What Payoffs Are Supported by PCE?}\label{Section-NTU}
\vspace{-.05in}
\subsection{Coalitional Minmaxes}\label{Section-Alignment}

In the introduction, we mentioned how coalitions can withstand punishments if they share aligned interests, which in turn limits the scope of PCE-supportable payoffs. We illustrate this phenomenon using the common-interest game depicted in \cref{table:EU}, where we adopt the specification of coalitional moves stipulated in \Cref{Example-NormalFormGame}.

\begin{table}[ht]
\centering
\setlength{\extrarowheight}{2pt}
\begin{subtable}{0.4\textwidth}
\begin{center}
\hspace{-0.5in}
    \begin{tabular}{cc|c|c|}
    & \multicolumn{1}{c}{} & \multicolumn{1}{c}{$L$}  & \multicolumn{1}{c}{$R$} \\\cline{3-4}
    & $U$ & $1,1$ & $0,0$ \\\cline{3-4}
    & $D$ & $0,0$ & $0,0$ \\\cline{3-4}
    \end{tabular}
\end{center}\vspace{-.1in}
\caption{\label{table:EU}}
\end{subtable}
\begin{subtable}{0.4\textwidth}
\begin{center}
\hspace{-0.5in}
    \begin{tabular}{cc|c|c|}
    & \multicolumn{1}{c}{} & \multicolumn{1}{c}{$L$}  & \multicolumn{1}{c}{$R$} \\\cline{3-4}
    & $U$ & $1,1$ & $-\epsilon, \epsilon$ \\\cline{3-4}
    & $D$ & $0,0$ & $0,0$ \\\cline{3-4}
    \end{tabular}
\end{center}\vspace{-.1in}
\caption{\label{table:NEU}}
\end{subtable}
 \caption{Payoffs in (A) are perfectly aligned while those in (B) are slightly misaligned ($\epsilon >0$).}\label{table:2player-equivalent}
\vspace{-.1in}
\end{table}
This game has a unique PCE, which prescribes the efficient action profile \((U, L)\) at every history guaranteeing each player a payoff of \(1\). To see why, let $\underline{w}$ denote the infimum of the normalized discounted payoffs from all PCEs, and consider an arbitrary PCE in which each player accrues $w\in [0,1]$. Since the continuation of a PCE at any history must itself be a PCE, if the pair blocks the alternative in the first period and chooses $(U,L)$ instead, each player receives at least $(1-\delta) + \delta \underline{w}$. The pair profits from the deviation unless $w \ge (1-\delta) + \delta \underline{w}$. Because this inequality must hold for $w$ arbitrarily close to $\underline{w}$, it then follows that $\underline{w}\ge 1$.

In this common-interest game, there is a gap between PCE and subgame perfect equilibrium (henceforth SPE) payoffs: as $(D,R)$ is  a Nash equilibrium of the stage game, all payoffs in $[0,1]$ can be supported by SPE of the repeated game. It turns out that the complete alignment of preferences is both sufficient and necessary for this gap. We find that coalitions of perfectly ``like-minded'' players can guarantee themselves a high coalitional payoff across all PCEs, regardless of discount factors. But misaligning preferences ever so slightly disrupts coalitional power. 

For instance, our analysis shows that in the stage game of \cref{table:NEU}, PCE can support payoffs arbitrarily close to $0$ when players are patient. We discuss why the logic used for the common-interest game does not apply here. Let $\underline{w}_i$ denote the infimum of player $i$'s PCE payoffs. In a PCE, the pair do not block and choose $(U,L)$ if at least one of them finds it unprofitable. Hence, a PCE payoff $(w_1,w_2)$ must satisfy at least one, but not necessarily both, of the following inequalities: 
\begin{align*}
w_1 \ge (1-\delta) + \delta \underline{w}_1 \;\;\text{or}\;\;
w_2 \ge (1-\delta) + \delta \underline{w}_2.    
\end{align*}
Because players' payoffs are not perfectly aligned, there exists a sequence of PCE payoff profiles $(w_1,w_2)$ such that $w_1 \to \underline{w}_1$ while $w_2$ stays bounded away from $\underline{w}_2$, and along this sequence the first inequality fails but the second holds. In these PCE, player $1$ would profit from blocking with player $2$ but player $2$ does not. By using continuation play that punishes one player and not the other, such PCE can push players' payoffs below $1$ without inviting a coalitional deviation. In contrast, both players are simultaneously rewarded and punished in the common-interest game, which makes it impossible to push payoffs below $1$. 

Generalizing these ideas, the right formulation of alignment here uses \cite{abreu1994folk}'s notion of \emph{equivalent utilities}: players $i$ and $j$ have {equivalent utilities} if there exist $k>0$ and $c\in \Re$ such that $v_j(a) =  k v_i(a) +c \text{ for all } a\in \alternatives$; otherwise, their utilities are not equivalent. We partition the set of players on this basis; let $C(i)$ be the set of players with whom player $i$ shares equivalent utilities.

This set leads to player $i$'s \textit{coalitional minmax}, namely the lowest payoff that she can be pushed down to when coalition $C(i)$ collectively best responds.
\begin{equation}
    \tag{Player $i$'s coalitional minmax}
    \label{Equation-coalitionminmax}
	\cm_i \equiv \min_{a\in \alternatives} \max_{a'\in E_{C(i)}(a)} v_i(a').
\end{equation}
The pure minmax is relevant because plans are pure, specifying an alternative following every history.\footnote{The term above is also well-defined as $A$ is compact, $v(\cdot)$ is continuous, and $E_C$ is continuous and compact-valued. Assuming that $E_C$ is monotone in $C$ simplifies the expression; absent monotonicity, the coalitional minmax is 
$\cm_i \equiv \min_{a\in \alternatives} \max_{C\subseteq C(i)} \max_{a'\in E_{C}(a)} v_i(a')$. This expression coincides with that above if $E_C(a)\subseteq E_{C(i)}(a)$ for every player $i$, coalition $C\subseteq C(i)$, and alternative $a$.} Using $\feasible$ to denote the convex hull of stage-game payoffs, we define $\fcr\equiv \left\{v\in\feasible:v_i>\cm_i\text{ for every }i=1,\ldots,n\right\}$ as the set of \emph{strictly coalitionally rational  payoffs}.\footnote{Although our setup does not have public randomization, the convex hull is relevant because these payoffs may be reached through intertemporal averaging \citep{sorin1986,fudenbergmaskin1991}.} We distinguish this from the individual minmax:
\begin{equation}\tag{Player $i$'s individual minmax}
\im_i\equiv \min_{a \in \alternatives}\max_{a' \in E_{\{i\}}(a)}v_i(a').
\end{equation}
Generally, $\cm_i$ is higher than $\im_i$. The two coincide if player $i$ does not share equivalent utilities with any other player, i.e., $C(i)=\{i\}$.  More strongly, $\fcr$ coincides with the set of strictly individually rational payoffs if no two players have equivalent utilities. The Roommates Problem lies in this class as do non-cooperative games that satisfy the NEU condition or full-dimensionality.

\subsection{The Power of Scapegoat Schemes}\label{Section-FolkTheoremNTU}

With these preliminaries in place, we state our first result. 
\begin{theorem} \label{Theorem-NTU}
For every $\delta \geq 0$, every PCE gives each player $i$ a payoff of at least $\cm_i$. Moreover, for every $v\in {\fcr}$, there is a $\underline{\delta}<1$ such that for every $\delta\in(\underline{\delta},1)$, there exists a PCE with discounted payoff equal to $v$.
\end{theorem}

The first part of the result identifies the coalitional minmax, $\cm_i$, as relevant when coalitions can block. The second part  shows that every strictly coalitionally rational payoff vector can be supported if players are sufficiently patient. An ancillary implication of this result is that, if $\fcr$ is nonempty, PCE are guaranteed to exist in the repeated game when players are patient, even for stage games with an empty core.

We compare \Cref{Theorem-NTU} to the folk theorem for SPE in repeated games with perfect monitoring. For this comparison, suppose that no two players share equivalent payoffs. Then \Cref{Theorem-NTU} has the same implication as the folk theorems of \cite{fudenberg1986folk} and \cite*{abreu1994folk}. We highlight two differences. First, the result applies for both repeated cooperative and non-cooperative games, including settings such as repeated matching. Second, a PCE is robust to both coalitional and individual deviations. Thus, when players are patient and have misaligned preferences, deterring  coalitional deviations is no harder than deterring individual deviations.

Why? The proof constructs punishments that crack coalitions. Since a coalition blocks only if all of its members agree, it suffices to make just one member unwilling to participate. We do so by singling out and punishing only one member of each coalition as though she were the sole deviator, while granting amnesty to the rest. Specifically, consider the construction in \cite{fudenberg1986folk} in which individual deviations are deterred by minmaxing a player before shifting to her player-specific punishment; such a construction is feasible whenever players have non-equivalent utilities. We adapt that scheme by assigning to each history and each (non-singleton) coalition $C$ a designated member---a ``scapegoat''---who is punished as if she were the sole deviator. This punishment makes the scapegoat unwilling to block with coalition $C$. By associating each coalition with a scapegoat, a PCE can ensure that no coalition finds it profitable to block whenever players are sufficiently patient.\footnote{We observe that a coalition can be cracked even if all coalition members share the same \emph{ordinal} rankings over alternatives; a PCE can nevertheless create player-specific punishments and isolate a coalition member as a scapegoat through its choice of how to sequence alternatives.}

This divide-and-conquer scheme fails if all members of a coalition $C$ share equivalent utilities. A higher minmax then applies for those players. To see why, consider such a coalition $C$ and suppose towards a contradiction that some PCE $\sigma$ could push the payoffs of its coalition members below their coalitional minmax. Members of coalition $C$ could guarantee their coalitional minmax if they could commit to a long-run plan in which they collectively best respond to the alternative specified by the plan after every history. Because payoffs are equivalent, all their gains and losses move in sync so we can effectively treat coalition $C$ as a unitary actor. Using an argument like that for the One-Shot Deviation Principle, it follows that there then exists a history at which coalition $C$ could profitably block, which precludes $\sigma$ from being a PCE.

Thus, \cref{Theorem-NTU} highlights that a coalition withstands the force of repeated games if and only if its members share \emph{completely} aligned preferences. On the one hand, this conclusion might appear to emphasize a ``knife-edge'' consideration. Yet, it applies to common-interest games, a commonly studied setting, in which we find that only efficient action profiles are chosen in a PCE.\footnote{For these games, the coalitional minmax therefore departs from---and is generally higher than---\cite{wen1994folk}'s \emph{effective minmax} for SPE, which would be $\min_{a\in A} \max_{j\in C(i)}\max_{a'_j\in A_j}v_i(a_{-j},a'_j)$.} More importantly, as we show in our study of strongly symmetric equilibria (\cref{Theorem-Symmetric}) and secret side-payments (\cref{Theorem-Secret}), this theme of alignment emerges even if players do not have equivalent utilities.

\subsection{Structural Properties}\label{Section-Structure}

\subsubsection{When Do Stationary PCEs Suffice?}\label{Section-Stationary}

Herein, we highlight how, in a rich class of games, all PCE payoffs can be achieved using stationary PCE. A plan $\sigma$ is \emph{stationary} if following every history $h$, the plan $\sigma$ specifies the same alternative in each period so long as the plan is not blocked, i.e., $\sigma(h,(\sigma(h),\emptyset))=\sigma(h)$.
After a block, a stationary plan may transition to a different alternative but would prescribe that alternative in every subsequent period. We prove that it suffices to study stationary PCE in \emph{convex} games with \emph{default-independent power}. As the latter notion is new, we turn to it first. 

For a coalition $C$ and alternative $a$, let $v_C(a) \equiv \{ (v_i(a))_{i\in C} \}$ be the projection of $v(a)$ onto $C$'s payoff space. Let $\payoffs_C(a)\equiv \{v_C(a'):a'\in E_C(a)\}$ be the set of $|C|$-dimensional payoff vectors that coalition $C$ can obtain from blocking that alternative.  
\begin{definition}\label{Definition-DIP}
    A stage game exhibits \textbf{default-independent power} if for every coalition $C$, there exists a $|C|$-dimensional payoff set $D_C$ such that $\payoffs_C(a)=D_C\cup \{v_C(a)\}$  for every alternative $a$.
\end{definition}
\Cref{Definition-DIP} asserts that the payoffs a coalition can achieve through blocking an alternative does not depend on that alternative. Although this notion rules out strategic form games (\cref{Example-NormalFormGame}), several well-studied coalitional games exhibit this property.

For instance, consider any characteristic function game studied in the classical cooperative game theory literature (\Cref{Example-CoalitionalGame}). In such a game, if a coalition blocks a default partition $\pi$, the set of payoffs it can achieve is independent of  the partition. An important application is matching  without externalities; there too, the set of payoffs that a coalition obtains from blocking an assignment does not hinge on the assignment. 

Our result identifies how stationary PCEs suffice in default-independent games if the stage game is \emph{convex}, i.e., $\{v(a):a\in \alternatives\}$ is a convex set.\footnote{We view convexity to be suitable for applications in which players can transfer utility or face a pure distribution problem. Alternatively, the set may be convex if players can access public correlation devices and make a choice to block before the realization of those lotteries.} 
\begin{theorem} \label{Theorem-Stationary}
If the stage game is convex and exhibits default-independent power, then for every $\delta\geq 0$, the set of PCE-supportable payoffs coincides with that supported by stationary PCEs.
\end{theorem}
\Cref{Theorem-Stationary} offers a conclusion that would be unexpected of subgame perfect equilibria of repeated games; optimal penal codes often involve non-stationary play \citep{abreu1988opc}. The proof invokes both convexity and default-independent power: the former enables us to replace a non-stationary path of play with a stationary path and the latter assures that the replacement does not affect any coalition's incentives. Our result generalizes \cite{BS2009} who obtain this conclusion for Dynamic Condorcet Winners. By clarifying that default-independent power is the key underlying property, \Cref{Theorem-Stationary} establishes that this conclusion holds more broadly.

\subsubsection{An Anti-Folk Theorem for Strongly Symmetric PCE}\label{Section-Symmetric} 
\cite{greenporter1984} and \cite*{fudenberg1994folk} elucidate how with  monitoring imperfections, strongly symmetric equilibria are inefficient but asymmetric play can support near-efficient payoffs. The theory here offers a complementary rationale for asymmetric play that applies even with perfect monitoring: a strongly symmetric PCE cannot credibly punish players because it aligns their interests. 

To make this point, we study a strategic form stage game (\cref{Example-NormalFormGame}) that is \emph{symmetric}: $A_i=A_j$ for all players $i$ and $j$, and for each permutation $\mu$ of $\{1, \ldots, n\}$, $v_i(a_{\mu(1)}, \ldots, a_{\mu(n)})=v_{\mu(i)}\left(a_1, \ldots, a_n\right)$ for every action profile $a$ and player $i$.
The set of symmetric action profiles is $A^S \equiv \{a\in \alternatives: a_i = a_j \text{ for all } i,j\in N \}$ and $V^S\equiv \{v(a): a\in A^S \}$ denotes the associated symmetric payoffs. Given that $V^S$ is compact and totally ordered, a maximal element exists denoted $\hat{v}$, the highest symmetric payoff. 

A plan $\sigma$ is \emph{strongly symmetric} if it specifies a symmetric action profile, $\sigma(h)$ in $A^S$, after every history $h$. \Cref{Theorem-Symmetric} characterizes strongly symmetric PCE.

\begin{theorem} \label{Theorem-Symmetric}
A strongly symmetric PCE exists if and only if the stage game has a symmetric core alternative $\hat{a}$ that attains the highest symmetric payoff $(v(\hat{a})= \hat v)$; then, $\hat v$ is the unique payoff supported by a strongly symmetric PCE.
\end{theorem}

This result reflects a collapse of intertemporal incentives: a strongly symmetric PCE exists if and only if the highest symmetric payoff lies in the core and thus can be supported without any carrots or sticks. This requirement is demanding, excluding games like the prisoner's dilemma or oligopolistic collusion. In such games, a PCE must instead rely on asymmetric punishments.

\begin{table}[h!]
\centering
\setlength{\extrarowheight}{2pt}
\begin{subtable}{0.4\textwidth}
\begin{center}
\hspace{-0.5in}
    \begin{tabular}{cc|c|c|}
    & \multicolumn{1}{c}{} & \multicolumn{1}{c}{$E$}  & \multicolumn{1}{c}{$S$} \\\cline{3-4}
    & $E$ & $2,2$ & $-1,3$ \\\cline{3-4}
    & $S$ & $3,-1$ & $0,0$ \\\cline{3-4}
    \end{tabular}
\end{center}\vspace{-.1in}
\caption{\label{table:PD}}
\end{subtable}
\begin{subtable}{0.4\textwidth}
\begin{center}
\hspace{-0.5in}
    \begin{tabular}{cc|c|c|}
    & \multicolumn{1}{c}{} & \multicolumn{1}{c}{$E$}  & \multicolumn{1}{c}{$S$} \\\cline{3-4}
    & $E$ & $2,2$ & $-1,1$ \\\cline{3-4}
    & $S$ & $1,-1$ & $0,0$ \\\cline{3-4}
    \end{tabular}
\end{center}\vspace{-.1in}
\caption{\label{table:notPD}}
\end{subtable}
 \caption{(A) is a Prisoner's Dilemma, which fails the condition of \Cref{Theorem-Symmetric}; (B) is a strategic form game that satisfies it.}\label{table:2player-SS}
\vspace{-.1in}
\end{table}

We illustrate this finding in \Cref{table:2player-SS}. In the Prisoner's Dilemma depicted in (A), the highest symmetric payoff $(2,2)$ can only be generated by the action profile $EE$, which is not a core alternative. Hence, the game admits no strongly symmetric PCE for any discount factor. Grim trigger, a strongly symmetric plan (and subgame perfect equilibrium if $\delta\geq 1/3$), fails to be a PCE because players would find it profitable to block during the punishment phase. By contrast, in the game depicted in (B), the payoff $(2,2)$ is in the core; \Cref{Theorem-Symmetric} asserts then that the \textit{unique} strongly symmetric PCE prescribes $EE$ after every history.

The key step in the proof of \Cref{Theorem-Symmetric} shows that strongly symmetric PCE payoffs, if they exist, must be unique. The logic is that once punishments are required to be symmetric, the game effectively becomes a common-interest game: unless players are already at the highest symmetric payoff profile $\hat{v}$, they would profitably block and move to it.\footnote{Suppose $\underline{v}$ is the lowest strongly symmetric PCE profile. Then players do not profitably block to an action that generates the  payoff $\hat{v}$ only if $\underline{v}\geq (1-\delta)\hat{v}+\delta\underline{v}$, which implies that $\underline{v}\geq \hat{v}$.} Therefore, the continuation payoffs for any strongly symmetric PCE are pinned down independently of history, which implies that such equilibria can enforce only those actions that are myopically optimal.

\section{Do Transfers Align Incentives?} \label{Section-Transfers}

\subsection{Transferable Utility Framework} \label{Section-TUFramework}

We turn to the question of whether transferable utility aligns incentives. We model transfers separately from alternatives because we vary their observability. \Cref{Section-PublicTransfers} studies publicly observed transfers and \Cref{Section-SecretTransfers} models secret side-payments. Like the bonus in relational contracting \citep[e.g.,][]{levin2003relational}, transfers are discretionary and not contracted upon ahead of time. In introducing transfers, our primary goal is to evaluate the extent to which coalitions can use transfers to profitably block. 

We represent transfers by \(\transfers \equiv [T_{ij}]_{i,j \in \players}\) where \( \transfers_{ij} \in [0,\infty)\) is the utility that player \(i\)  transfers to player \(j\). A player's \emph{experienced payoff} is the sum of the payoff from the chosen alternative and net transfers: $u_i(a, \transfers) \equiv v_i(a) + \sum_{j\in \players} \transfers_{ji} -  \sum_{j\in \players} \transfers_{ij}$. Let $\mathcal{T}$ be the set of all $(n\times n)$ transfer matrices in which entries along the main diagonal equal $0$ (so that a player cannot  transfer utility to herself). We denote transfers paid by members of coalition $C$ by \(\transfers_C \equiv [\transfers_{ij}]_{i \in C,j\in \players}\); $\mathcal{T}_C$ is the set of $|C|\times n$ transfer matrices.

An outcome of the stage game now includes the chosen alternative, the identity of blocking coalitions (if any), and the chosen transfers. The set of stage-game outcomes is \(\outcomesTU \equiv  \alternatives\times \mathcal{B} \times \mathcal{T}\). Histories and paths are defined analogous to the NTU case with the addition of transfers. We denote the set of histories with transfers by \(\historiesTU\). A plan \(\sigma: \historiesTU \rightarrow A \times \mathcal{T}\) specifies an alternative and configuration of transfers, based on history. We use \(a(h | \sigma)\) and \(\transfers(h | \sigma)\) to denote the default alternative and transfers in \(\sigma(h)\). We modify the definition of \(U_i(h | \sigma)\) to reflect the influence of transfers.

By blocking the default $(a,T)$, a coalition $C$ can choose a different alternative $a'\in E_C(a)$ and change their transfers to any $\transfers'_{C}$. A question we have to tackle is: \emph{if a coalition blocks, what transfers do players outside the coalition make?} Two distinct answers strike us as reasonable. The first hews to a ``simultaneous noncooperative'' formulation in which coalition $C$'s block surprises players outside the coalition, who therefore make transfers $T_{-C}$ as was specified by the plan. The second models a ``cooperative'' approach in which if a coalition blocks, its members can transfer utility among themselves but players outside that coalition do not transfer any utility to them. To accommodate both answers, we formulate the transfers of others abstractly. 
\begin{assumption} \label{Assumption-Bounded-Transfers}
For each coalition $C$, if $C$ blocks a default transfers matrix $T$, the transfers made by players outside of $C$ is $\hchi^{C}(T)$ where $\hchi^{C}: \mathcal{T} \rightarrow \mathcal{T}_{\players\backslash C}$ satisfies:
    \begin{enumerate}
        \item For each bounded set $S\subseteq \mathcal{T}$, the image $\hchi^{C}(S) \subseteq \mathcal{T}_{\players\backslash C}$ is also bounded.
        \item If $T$ satisfies $T_{ij}=0$ for all $i\notin C, j\in C$, then $\hchi^{C}_{ij}(T) = 0$ for all $i\notin C, j\in C$.        
    \end{enumerate}
\end{assumption}
\Cref{Assumption-Bounded-Transfers} encompasses the two specifications described above: the former corresponds to $\hchi^{C}(T) = T_{-C}$ whereas the latter corresponds to $\hchi^{C}_{ij}(T) = T_{ij}$ for all $i,j\in \players\backslash C$, and $\hchi^{C}_{ij}(T) = 0$ for all $i\in \players\backslash C$ and $j\in C$.

Thus, if coalition $C$ blocks, chooses actions $a'$ and changes transfers to $\transfers'_{C}$,  the realized outcome is then $(a', \{C\}, \transfers'_{C}, \hchi^{C}(\transfers))$. We now define the versions of profitable blocking and PCE appropriate for this setting. 

\begin{definition} \label{Definition-PCET1}
Coalition $C$ \textbf{profitably blocks} plan $\sigma$ at history $h$ if there exists an alternative $a'\in E_{C}(\,a(h|\sigma)\,)$ and transfers $\transfers '_C = [\transfers '_{ij} ]_{i\in C, j\in \players }$ such that for all $i\in C$,
\begin{align*}
    (1-\delta ) u_i\big(a', \transfers '_C , \hchi^{C}(T(h|\sigma)) \big) + \delta U_i\big( h, a', \{C\}, \transfers '_C , \hchi^{C}(T(h|\sigma)) \mid \sigma  \big) > U_i(h | \sigma).
\end{align*}
\end{definition}

\begin{definition}\label{Definition-PCET2}
A plan $\sigma$ is a \textbf{perfect coalitional equilibrium} if it cannot be profitably blocked by any coalition at any history.    
\end{definition}

To rule out Ponzi schemes, we make the following technical assumption.
\begin{assumption}\label{Assumption-Bounded}
	We consider plans $\sigma$ such that continuation values are bounded across histories: $\left\{U(h|\sigma): h\in \historiesTU \right\}$ is a bounded subset of $\Re^n$.
\end{assumption}

\subsection{Publicly Observed Transfers}\label{Section-PublicTransfers}

Transfers allow blocking coalitions to distribute gains among their members. One might intuit that transfers would then align coalition members' incentives. However, we find that public transfers have the opposite effect, undermining coalitions. Even those coalitions whose payoffs would be aligned absent transfers can now be splintered. Our result below establishes that all payoffs that are feasible and strictly \emph{individually} rational can be supported. 

To state this result, we re-define the set of feasible payoffs to account for transfers: 
\begin{align*}\feasiblet\equiv \conv \Big( \Big \{ u\in \Re^n:\exists a\in\alternatives\text{ such that }\sum_{i\in \players} u_i= \sum_{i\in \players} v_i(a)\Big\}\Big) .
\end{align*} 
The set of feasible and strictly individually rational payoffs is 
\begin{align*}
	\firt\equiv\left\{u\in \feasiblet:u_i>\im_i\text{ for every }i=1,\ldots,n\right\}.
\end{align*}

\begin{theorem} \label{Theorem-PTU}
For every $\delta\ge 0$, every PCE gives each player $i$ a payoff of at least $\im_i$. Moreover, for every $u \in {\firt}$, there is a $\underline{\delta}<1$ such that for every $\delta\in(\underline{\delta},1)$, there exists a PCE with discounted payoff equal to $u$.
\end{theorem}

We omit a formal proof of \Cref{Theorem-PTU} because it follows as a special case of \Cref{Theorem-Secret} in \Cref{Section-SecretTransfers}. Comparing \cref{Theorem-NTU,Theorem-PTU} reveals that rather than aligning coalition members' incentives, public transfers undermine any existing preference alignment that was present before the transfers. The key idea is that public transfers make the distribution of utilities within any blocking coalition transparent to all. Therefore, a PCE can tailor the selection of a scapegoat in a blocking coalition to the post-transfer payoffs of coalition members so as to deter blocking alongside all transfer choices. We illustrate this logic in \Cref{Example-Procurement} below. Moreover, players $i$ and $j$ have misaligned interests (or non-equivalent utilities) once one can pay transfers to the other. This misalignment allows us to construct player-specific punishments.

\begin{exampleu} \label{Example-Procurement}
Consider a setting in which, in each period, a single firm $F$ procures up to 10 units of perfectly divisible goods each period from two suppliers, $S_1$ and $S_2$, each of whom faces zero marginal cost. Each unit is worth \$1 to the firm, and prices can be tailored to each supplier. Thus, the total surplus each period is \$$10$. 

\begin{figure}[ht]
\centering
\begin{subfigure}{.45\linewidth}
\centering
\begin{tikzpicture}[scale=0.6]

\tikzset{every path/.append style={line width=1pt}}

\path (0,2) node[draw, shape=circle, line width=1pt, inner sep=0pt, minimum size=20pt] (u) {$F$};

\path (3,4) node (d1) [draw, shape=circle, line width=1pt, inner sep=0pt, minimum size=20pt]{$S_1$};
\path (3,0) node (d2) [draw, shape=circle, line width=1pt, inner sep=0pt, minimum size=20pt]{$S_2$};

\draw[thick] (u) -- (d1) node[pos=0.4, above=1.8mm] {};
\draw[thick] (u) -- (d2) node[pos=0.4, below=2.8mm] {};


\draw (u) node[left=4.5mm] {$10$};
\draw (d1) node[right=4.5mm] {$0$};
\draw (d2) node[right=4.5mm] {$0$};
\end{tikzpicture}
\caption{}\label{figure: core}
\end{subfigure}
\begin{subfigure}{.45\linewidth}
\centering
\begin{tikzpicture}[scale=0.6]

\tikzset{every path/.append style={line width=1pt}}

\path (0,2) node[draw, shape=circle, line width=1pt, inner sep=0pt, minimum size=20pt] (u) {$F$};

\path (3,4) node (d1) [draw, shape=circle, line width=1pt, inner sep=0pt, minimum size=20pt]{$S_1$};
\path (3,0) node (d2) [draw, shape=circle, line width=1pt, inner sep=0pt, minimum size=20pt]{$S_2$};

\draw[thick] (u) -- (d1) node[pos=0.4, above=1.8mm] {};
\draw[thick] (u) -- (d2) node[pos=0.4, below=2.8mm] {};


\draw (u) node[left=4.5mm] {$4$};
\draw (d1) node[right=4.5mm] {$3$};
\draw (d2) node[right=4.5mm] {$3$};
\end{tikzpicture}
\caption{}\label{figure: surplus extraction}
\end{subfigure}
\caption{(A) shows the unique core of the stage game, where the firm keeps the entire surplus; (B) shows an allocation where the suppliers retain some surplus.}
\end{figure}

Because the firm $F$ can meet its demand through either supplier individually, each pair $\{F,S_1\}$ and $\{F,S_2\}$ must secure a total surplus of $10$ in the stage-game core. Otherwise, at least one pair could profitably block the allocation by dividing the share of the other supplier between themselves.\footnote{We consider the specification of transfers from \Cref{Assumption-Bounded-Transfers} in which players outside a blocking coalition make no transfers to members of a blocking coalition.} Thus, the firm captures the entire surplus in the stage-game core, as depicted in \Cref{figure: core}.

A PCE captures how repeated play could distribute surplus to the suppliers. Suppose $\delta = 2/3$. Then the payoffs in \Cref{figure: surplus extraction} can be supported by a PCE using a modified core-reversion scheme: players play the allocation in (B) indefinitely unless some pair blocks. Should the coalition $\{F, S_1\}$ block and the resulting allocation give $F$ a stage-game payoff of $x\geq 4$, then play reverts permanently to the stage-game core in every subsequent period; otherwise, their block is ignored and play continues as in (B). A similar response applies to blocking by $\{F,S_2\}$. Observe that this scheme conditions the continuation play not only on the identity of the blocking coalition but also on the distribution of surplus within it. Finetuning punishments in this manner allows the scheme to ensure that at least one player is made worse off in each (potential) blocking coalition. For example, if $\{F,S_1\}$ contemplates a block where $x \ge 4$, then $S_1$ suffers from participating in this block because her payoff can be no more than $(1-\delta)(6) +\delta(0)$, which is less than $3$ given that $\delta=2/3$; by contrast, if $x < 4$, the scheme deters firm $F$ from blocking because doing so leads to a lower payoff today without affecting its continuation payoff. \qed
\end{exampleu}

\vspace{-.15in}
\subsection{Secret Transfers}\label{Section-SecretTransfers} 

In light of the analysis above, we ask: \emph{what if some coalitions can make secret side-payments when they block? Could their incentives then be aligned?} We view this question to be of conceptual and practical import given that, in many contexts, transfers within coalitions are not public. For instance, a firm when poaching another firm's employees might offer a contract whose terms are observed by the worker and firm alone. These contracts are often confidential, a point to which we return in \Cref{Section-Applications} in our discussion of wage transparency. More broadly, groups of players often seek and find ways to transfer money under the table when defecting from a social arrangement. Our analysis here identifies the benefits that coalitions accrue from making secret transfers even if their blocking decision is observable.

We consider a setting in which some but not all coalitions can make secret transfers; denote the set of all non-singleton coalitions that can by $\seccoal \subseteq \coalitions$. In our leading application, we consider firms that can offer contracts to workers with private wage terms. A secret side-payment is observed within a coalition but not outside it. Aligned with this idea, we define the outcomes that are publicly observed by all parties.

\begin{definition}\label{Definition-SecretTransfersSome}
Given the set $\seccoal \subseteq \coalitions$ and a stage-game outcome $o=(a, B, \transfers) \in \outcomesTU$, the public transfers, denoted $T^p$, exclude those  made within any blocking coalition in $\seccoal$:
\begin{equation*}
    T^p_{ij} =
\begin{cases}
\,\verb|#| & \text{\normalfont if } \exists \,C \in \seccoal \cap B \text{ \normalfont such that } \{i,j\}\subseteq C,\\
 T_{ij} & \text{\normalfont otherwise.} 
\end{cases}
\end{equation*}
The public component of $o$, denoted by $o_{p}$, is $o_p \equiv (a, B, T^p)$. For any history $h=(o^{\tau})_{\tau=0}^{t}$, the public component of $h$ is ${h}_{p}=(o_p^{\tau})_{\tau=0}^{t}$. 
\end{definition}

The definition stipulates that if coalition $C$ can make secret transfers, then transfers within it are not recorded in the public history whenever it blocks; instead, those transfers are recorded as $\verb|#|$, indicating that they are missing. In this setting with imperfect public monitoring, we consider the analog of a \emph{perfect public equilibrium} \citep*{abreu1990toward,fudenberg1994folk}.

\begin{definition}
    A plan $\sigma$ is \textbf{public} if  $\sigma(h)=\sigma(h')$ for all $h,h'\in \historiesTU $ satisfying $h_p = h'_p$. A \textbf{public PCE} is a public plan $\sigma$ that constitutes a PCE of the repeated game.
\end{definition}

We argue below that secret transfers empower a coalition to act as if it were a single party, which in turn limits what a public PCE can support. In particular, a public PCE can condition only on the identity of a blocking coalition and not on how its members divide their surplus. Before proceeding to our results, we illustrate this tension by returning to \Cref{Example-Procurement}.

\begin{rexample}{\ref{Example-Procurement} (Continued)}
The PCE supporting the allocation in \Cref{figure: surplus extraction} relies on continuation play that conditions on transfers within the blocking coalition. Suppose transfers are instead secret: if the coalition $\{F,S_i\}$ blocks, transfers between $F$ and $S_i$ are unobserved by the other seller $S_{-i}$. We show that, in this setting, simple core reversion cannot support the allocation in (B). Consider a plan in which, on-path, the players obtain the payoffs in (B) but any block triggers permanent reversion to the stage-game core. Firm $F$ and seller $S_1$ can then find it profitable to block, generate the full surplus of $10$, and use an appropriate choice of transfer to make sure that each comes out ahead, given continuation play.\footnote{E.g., firm $F$ could transfer all of the current-period surplus to $S_1$, which results in a normalized discounted payoff of $(1-\delta)(10)+\delta(0)>3$ to $S_1$ and $(1-\delta)(0)+\delta(10)>4$ to $F$ (recall $\delta=2/3$).} Therefore, this plan is not a PCE. 

Would other plans work? The challenge is that every plan can be defeated by some choice of transfers between $F$ and $S_1$. In particular, \cref{Theorem-Secret} below implies that every public PCE must deliver a total payoff of $10$ to each of the coalitions $\{F,S_1\}$ and $\{F,S_2\}$. Only the stage-game core, wherein the firm captures the entire surplus, is compatible with this division of surplus. \qed
\end{rexample}

To formalize the general conclusion, suppose a coalition $C\in {\seccoal}$ acted as a unitary actor that maximizes the total utility $\sum_{i\in C} v_i(\cdot)$ equipped with an effectivity function $E_C(\cdot)$. We would then define its minmax as  
\begin{align}\tag{Coalition $C$'s minmax}\label{Equation-SecretCoalitionminmax}
	\cf_C \equiv \min_{a\in \alternatives} \max_{a'\in E_{C}(a)} \sum_{i\in C} v_i (a') 
\end{align}
Treating each coalition $C$ in $\seccoal$ in this way would lead to the set of feasible and strictly $\seccoal$-coalitionally rational  payoffs
\begin{align*}
\fcrt(\seccoal) \equiv \left\{ u \in \feasiblet: 
\begin{array}{r}
u_i > \im_i \text{ for every } i \in N, \\
\sum_{i \in C} u_i > \cf_C \text{ for every } C \in \seccoal
\end{array}
\right\}.
\end{align*}
The above set derives the set of feasible and ``individually'' rational payoffs in a fictitious game in which the set of players is $N\cup\seccoal$. Our result below shows that this set characterizes the limits of public PCE.
\begin{theorem}\label{Theorem-Secret}
For every $\delta \ge 0$, every public PCE gives each coalition $C\in {\seccoal}$ a total payoff of at least $\cf_C$ and every player $i$ a payoff of at least $\im_i$. Moreover, for every $u\in \fcrt(\seccoal)$, there is a $\underline{\delta}<1$ such that for every $\delta\in(\underline{\delta},1)$, there exists a public PCE with a discounted payoff equal to $u$.
\end{theorem}

\Cref{Theorem-Secret} identifies the significant gains that coalitions accrue from finding a channel to transfer utility secretly; all those in such a coalition can collectively enjoy a higher minmax while those outside a secret coalition can be pushed towards their individually rational payoffs. One might view these high coalitional minmaxes as conveying an ``anti-folk'' flavor. Indeed, in the example above and both applications, secret transfers reduce the supportable payoff set to the core of the stage game.

To see why \Cref{Theorem-Secret} holds, we first explain why each coalition $C\in \seccoal$ can guarantee its minmax for every $\delta\geq 0$. Consider a plan $\sigma$ and suppose towards a contradiction that coalition $C$ failed to achieve $\cf_C$. Were coalition $C$ a unitary actor, it could guarantee a payoff no lower than $\cf_C$ by executing a long-run plan where it best responds to the default alternative in each period, blocking when necessary. An argument similar to the one-shot deviation principle then establishes that the total utility of members of coalition $C$ must increase by blocking once at some history $h$. By apportioning that gain across the members of $C$ through secret side-payments, it can then be assured that each member simultaneously profits from the block at that history without affecting continuation play.\footnote{In addition to secrecy, this step uses the one-to-one transferability of utility. If utility were costly to transfer, then the coalition $C$ as a whole may be on a higher utility frontier from blocking but have no way to apportion those gains so that each player profits.} Such a block is then profitable, which means $\sigma$ is not a PCE.

We turn to why every payoff in $\fcrt(\seccoal)$ can be attained for patient players. Consider the fictitious game in which the set of players is $N\cup \seccoal$. In this game, we directly construct ``player-specific'' punishments for each player; one can see that such punishments exist because the payoffs across the players in this fictitious game satisfy the NEU condition. Using these punishments, the payoff of each player in $N\cup \seccoal$ can be pushed arbitrarily close to its minmax.

We contrast this result with \Cref{Theorem-PTU}, which corresponds to the special case in which $\seccoal$ is empty. Therein, we cracked coalitions by fine-tuning the selection of the scapegoat to the details of who pays whom. Such an approach fails here because the continuation play cannot vary the identify of the scapegoat with the transfers within coalitions in $\mathcal S$. Not only do these scapegoat schemes unravel but so do any other that pushes the payoff of one of these coalitions  below its minmax.

\vspace{-.15in}
\section{An Application to Labor Market Matching} \label{Section-Applications}

Many practices in labor markets, such as collective wage bargaining and firms' collusive wage-setting, are fundamentally driven by long-run incentives. We incorporate these considerations into the canonical model of \cite{kelsocrawford1982} (henceforth KC82). In the process, we obtain a new perspective on when and how wage transparency benefits workers. 

In this application, the set of players is $\players \equiv \firms \cup \workers$, where $\firms$ is the set of firms and $\workers$ is the set of workers. We use $f$ to denote a generic firm, $w$ to denote a generic worker; $i$ and $j$ denote generic players who could be workers or firms. We use $C\subseteq N$ to denote a generic coalition. We call a coalition \textit{essential} if it comprises a single firm and a nonempty set of workers; let $\esscoal \equiv \big \{\{f\}\cup W: f\in \firms, W\subseteq \workers, W\ne \emptyset \big \}$ be the set of all essential coalitions.

Each firm can hire multiple workers. An \textit{assignment} is a mapping $\phi: \firms \cup \workers \rightarrow 2^{\firms\cup\workers}$ such that (i) every worker $w$ is assigned to at most one firm; (ii) every firm $f$ is assigned to a (potentially empty) set of workers; and (iii) $w \in \phi(f)$ if and only if $f\in \phi(w)$.\footnote{Formally, the first two conditions stipulate that $\phi(w) \subseteq \firms $ with $|\phi(w)|\le 1$ and that $\phi(f)\subseteq {\workers}$. Throughout, we write $\phi(i)=\emptyset$ if player $i$ is unassigned.} The set of alternatives $A$ comprises all assignments between firms and workers. A \textit{matching} is an assignment of workers to firms and a specification of transfers made between players. Following KC82, we allow non-zero transfers to occur only between employers and their employees. Therefore, the set of matchings is $\mathcal{M} \equiv \big\{(\phi,T) \in A\times \mathcal{T}: T_{ij}\ne 0 \text{ only if } i\in \phi(j)\big\}$.

Each firm $f$ has a revenue function ${v}_f$ defined over all subsets of $\workers$, with ${v}_{f}(\emptyset)$ normalized to $0$; similarly, worker $w$ has a pre-renumeration utility function ${v}_{w}$ defined over singleton or empty subsets of $\firms$, where the payoff of being unemployed, ${v}_w(\emptyset)$, is normalized to $0$. Abusing notation, we use $v_i$ to also denote the utility player $i$ receives from an assignment, so $v_i(\phi) = v_i(\phi(i))$. Given a matching $(\phi,T)$, player $i$'s experienced payoff is $u_i(\phi,T) \equiv v_i(\phi) +\sum_{j\ne i} T_{ji} - \sum_{j\ne i} T_{ij}$.

If a coalition $C$ blocks, its members sever their links with players outside $C$ and may arbitrarily rearrange the matching within $C$.\footnote{KC82 study a more restrictive notion of blocking, allowing only essential coalitions to block. Our results apply also to that setting.}  Formally, for all $\phi\in \alternatives$, 
\[
E_C(\phi) = \{\phi'\in A: \phi'(i)\subseteq C  \text{ for all } i\in C, \text{ and }\phi'(i)= \phi(i)\backslash C \text{ for all } i \notin C\}.
\] 
This formulation specifies that if coalition $C$ blocks assignment $\phi$, any match between agents both outside $C$ remains intact while those involving a member of $C$ and a nonmember is severed. Because transfers happen only between matched players, those outside of $C$ make no transfers to those in $C$ if $C$ blocks. This specification adheres to the ``budget-balance'' case of \cref{Section-TUFramework}; the mapping $\{\hchi^C\}_{C \in \coalitions}$ denotes the transfers made across players outside of coalition $C$. 

We now state the definitions of profitable blocking and core.

\begin{definition}\label{Definition-KC-Core}
A matching $(\phi,T)$ is \textbf{profitably blocked by coalition $C$} if there exists an  assignment $\phi'\in E_C(\phi)$ and transfers $T'_C = [T'_{ij}]_{i\in C, j\in \players}$ such that 
all in $C$ are better off from the matching $(\phi',T'_C,\hchi^C(T))$: 
\[
    u_i(\phi',T'_C,\hchi^C(T)) > u_i(\phi,T) \text{ for all } i\in C.
\]
A matching $(\phi,T)$ is a \textbf{core allocation} if it cannot be profitably blocked by any coalition. The \textbf{stage-game core}, denoted by $\mathcal{K}$, are the payoffs of core allocations.
\end{definition}

KC82 show that when firms' revenue functions satisfy the \textit{gross substitutes} condition, formally defined in the footnote,\footnote{For a vector of wages from firm $f$, $T_f = (T_{fw})_{w\in \workers}$, define $Ch_f(T_f ) \equiv \argmax_{W \subseteq \workers} ( v_f(W ) - \sum_{w \in W} T_{fw}  )$. For every set of workers $W$ and pair of wage vectors $T_f$ and $T'_f$ such that $T'_{fw} \ge  T_{fw}$ for all $w \in \workers$, define $E(W, T_f, T'_f) \equiv \{w \in W: T'_{fw} = T_{fw} \}$. Firm $f$'s revenue function satisfies \emph{gross substitutes}  if $\hat{W} \in Ch_f(T_f)$ implies that there exists $\hat{W}' \in Ch_f(T'_f)$ such that $E(\hat{W}, T_f, T'_f) \subseteq \hat{W}'$.} the core is nonempty. We impose the same assumption and, although we permit blocking by non-essential coalitions, we prove that the resulting core remains unchanged.

Having described the stage game, we now consider the implications of repetition, using the framework and analyses of \Cref{Section-Transfers}. The concept of PCE defined in \cref{Definition-PCET2} naturally extends to this setting, where a plan specifies a stage-game matching at every history.\footnote{Our approach models firms and workers who interact via spot contracts and studies how intertemporal incentives discipline which spot contracts are signed over time. An alternative modeling approach might allow firms and workers to sign long-term contracts. \cite{Diamantoudi2015decentralized} and \cite{zhang2016} study two-period matching models with commitment, emphasizing how the availability of long-term contracts influences equilibrium outcomes.} The set of feasible payoffs in this repeated game is $\feasiblet^{\mathcal{M}}\equiv  \conv \Big( \Big \{ u\in \Re^n:\exists (\phi,T) \in\mathcal{M}\text{ such that }u=u(\phi,T) \Big\}\Big)$. 
Player $i$'s individual minmax payoff is $\im_i=0$, which is achieved through a matching that ostracizes her. Thus, the set of feasible and individually rational payoffs is $\firt^{\mathcal{M}}\equiv  \Big \{ u\in \feasiblet^{\mathcal{M}}:u_i> 0 \text{ for all } i\in \players\Big\}$. 

\paragraph{Public vs. Private Wages.} The set of matchings that can be supported in the repeated game hinges on whether past wage terms are publicly or privately observed. In the former, we find that many outcomes may be supported; in the latter, we see a collapse of intertemporal incentives leading to only payoffs in the core being tenable.

Suppose all transfers are public. Using \cref{Theorem-PTU}'s argument, we show that any feasible and individually rational payoff can be supported when players are patient.

\begin{proposition} \label{Proposition-KC-Public}
For every $\delta\ge 0$, every PCE gives each player $i$ a payoff of at least $0$. Moreover, for every $u \in \mathcal{U}_{IR}^{\mathcal{M}}$, there is a $\underline{\delta}<1$ such that for every $\delta\in(\underline{\delta},1)$, there exists a PCE with discounted payoff equal to $u$.
\end{proposition}

While we assume gross substitutes to maintain comparability with KC82, \Cref{Proposition-KC-Public} itself does not require gross substitutes. Even if the stage-game core is empty,  stable schemes exist in the repeated game if players are sufficiently patient.

Now suppose each firm can hire and offer private wage terms to a group of workers. Recall that a coalition $C$ is {essential} if it comprises a single firm and a nonempty set of workers, and $\mathcal E$ is the set of all essential coalitions. For the next result, we suppose that the set of secret coalitions, $\mathcal S$, includes all essential coalitions, $\mathcal E$. In this setting, we obtain a conclusion sharper than \Cref{Theorem-Secret}: all payoff vectors outside the core are untenable regardless of the players' patience.

\begin{proposition} \label{Proposition-KC-Secret}
    Suppose $\esscoal \subseteq \seccoal$. For every $\delta\ge 0$, a public PCE supports a discounted payoff vector if and only if that payoff vector is in the stage-game core, $\mathcal{K}$.
\end{proposition}
\Cref{Proposition-KC-Secret} is an anti-folk theorem that asserts that empowering essential coalitions to make secret transfers cripples a PCE's ability to go beyond the stage-game core. The ``if'' direction is immediate as the infinite repetition of a core allocation constitutes a public PCE. For the ``only if'' direction, observe that by the same logic as in \cref{Theorem-Secret}, every essential coalition is assured its minmax payoff in a public PCE. In other words, every firm $f$ and group of workers $W$ must achieve a total utility of at least what they would get from matching together, namely, $v_f(W)+\sum_{w\in W} v_w(f)$, which is this coalition's value. In our proof, we show that all payoff vectors that assure that each coalition obtains at least its value lie in the stage-game core.

\paragraph{Who Benefits from Wage Transparency?} To answer this question, we specialize to a setting in which workers are homogeneous (which KC82 also consider). Suppose that all workers have the same payoff function, $v_w(f)= \lambda(f)$ for each firm $f$ and worker $w$. Additionally, each firm's revenue depends only on the number of workers it hires: $v_f(W) = \tilde{v}_f(|W|)$. Let $\rho(f,l) \equiv \lambda(f) + \tilde{v}_f(l) - \tilde{v}_f(l-1)$ be the surplus generated from assigning the $l^{\text{th}}$ worker to firm $f$. We continue to assume that firm revenues satisfy gross substitutes, which KC82 show translates into a condition on {diminishing marginal returns}: $\rho(f,l)$ is then weakly decreasing in $l$ for each $f$. 

In this setting, an assignment $\phi^*$ that maximizes total social surplus is found by greedily assigning workers to firm slots in order of their contribution to total surplus. Formally, let $L\equiv |\mathcal {W}|$ be the total labor supply and $\eta(\ell)$ be the $\ell^{\text{th}}$ highest value of $\{\rho(f,l): f\in \firms, l\ge 1 \}$ for $\ell$ in $\{1,\ldots,L\}$, which represents the marginal value of assigning the $\ell^{\text{th}}$ worker optimally. To simplify our exposition, we assume that the set $\{\rho(f,l): f\in \firms, l\ge 1 \}$ has no ties and excludes $0$; every efficient assignment must then fills ``slots'' $\{(f,l):\rho(f,l)\ge  \max\{0, \eta(L)\}\}$ leaving all others vacant. We use $A^*$ to denote the set of efficient assignments and note that all its elements are identical up to a relabeling of workers. Finally, we assume that it would be inefficient for a single firm to hire all workers so that each firm faces some competition.

Let $\mathcal{U}^{\circ} \equiv \conv\{u(\phi^*,T): (\phi^*,T)\in \mathcal{M} \text{ and }\phi^*\in A^*\}$ denote the efficient frontier of the set of feasible payoff profiles. Within this set, let $\mathcal{U}^+ \equiv \{\tilde{u} \in \mathcal{U}^{\circ}: \tilde{u}_i> 0 \text{ for all } i\in \firms\cup\workers \}$ 
denote those surplus divisions under which every player obtains more than her minmax. We also consider the surplus divisions in which each worker obtains a net utility of approximately the ``marginal product'' of the last employee in the economy while firms are residual claimants. Let $\eta(L+1)$ be the $(L+1)^{\text{th}}$ highest value of $\rho(f,l)$ assuming that there were an additional worker in the economy. Then we define:
\[
\mathcal{U}^* \equiv \{\tilde{u}\in \mathcal{U}^{\circ}: \max\{0,\eta(L+1)\}\le \tilde{u}_w=\tilde{u}_{w'} \le \max\{0,\eta(L)\}\text{ for all }w,w'\in \workers \}.
\]
We show that $\mathcal{K} = \mathcal{U}^*$, which yields the following conclusion.

\begin{proposition} \label{Proposition-Wages}
If wages are public, for each $u\in \mathcal{U}^+$, there exists $\underline{\delta}<1$ such that for every $\delta\in (\underline{\delta},1)$, there exists a PCE with discounted payoff equal to $u$. By contrast, if wages are private, for every $\delta\geq 0$, the set of payoffs supported by public PCEs is $\mathcal{U}^*$.
\end{proposition}

\cref{Proposition-Wages} asserts that any surplus division from the efficient assignment $\phi^*$ in which individual rationality conditions hold can be supported if wages are public. Firms could collude to extract nearly all surplus from workers; alternatively, workers can collectively bargain to retain almost the entire surplus. By contrast, if wages are private, workers accrue the value of the marginal product of the least productive employee, and firms capture the remaining surplus.

Given \Cref{Proposition-Wages}, workers favor wage transparency if they are plentiful---i.e., $\eta(L)<0$---or their marginal product falls quickly. Without transparency, workers compete intensely for slots and thereby drive their earnings to near $0$. By contrast, wage transparency enables them to use collective bargaining to obtain higher wages for them all. In such a scheme, were a firm to try to poach workers in a way that is mutually profitable, a PCE would deter workers from accepting those offers by reverting to the stage-game core from the next period onwards. Thus, workers recognize that the future promise of high wages---and the continued success of their collective bargaining efforts---requires them to reject offers that are tempting today. 

By contrast, if workers are scarce or the marginal product of workers falls slowly---i.e., $\eta(1)\approx \eta(L)$---it is firms who favor wage transparency. All PCEs under private wages result in high wages, as firms compete heavily for workers. Wage transparency allows firms to collusively suppress wages, with all of them setting low wages and agreeing not to poach each other's workers. Such an agreement is viable given the continuation play in which poaching today triggers a ``salary war'' tomorrow. 
 
We depict this prediction in \Cref{Figure-Matching}: (A) shows the distribution of worker and firm surplus under private wages when the marginal product of labor falls slowly and (B) shows the same when the marginal product falls quickly. As the figure shows, workers are worse off absolutely and relatively in the latter case. Were wages transparent, workers or firms could obtain better terms. In (A), workers have little to gain but much to lose from wage transparency as firms could then suppress wages; by contrast, in (B), it is workers who can use wage transparency to secure a larger share of the pie.

\begin{figure}[t]
\centering
\begin{subfigure}{.45\linewidth}
        \centering
        \begin{tikzpicture}[
        declare function={
        f(\x)=
        0.041148938207776*(\x)^7
        -0.345651080945326*(\x)^6
        +1.152170269817777*(\x)^5 
        -1.935646053293930*(\x)^4
        +1.580534083292260* (\x)^3 
        -0.274538084465002*(\x)^2 
        -0.615313510019283* (\x)
        + 1.5;},
        >=stealth,
        scale=2.5
        ]
        
        \newcommand{\WD}{2.1};
        \draw[->] (0,0) -- ({\WD+0.15},0) node[anchor=west] { $\ell$}; 
        \draw[->] (0,0) -- (0,1.65) node[anchor=east] {$\eta(\ell)$}; 
        
        
        \pgfmathsetmacro{\WDD}{\WD+0.1};
        \foreach \x in {0.1, 0.2,...,\WD} {
            \draw (\x,0.01) -- (\x,-0.01);
        }
        
        \node[below] at (0, 0) {\footnotesize $0$};
        \node[below] at (\WD, 0) {\scriptsize $L$};
        \node at (\WD/2,-0.15) {\text{\scriptsize Number of Workers}};
        \draw (-.15,0.825) node[rotate=90]{\text{\small \scriptsize{Marginal Product}}};
        
        \draw[line width=1.1pt, black, dashed, domain=0.1:\WD] plot (\x, {f(\x-0.1)+0.005}) node[right] {};

        \pgfmathsetmacro{\WDD}{\WD-0.1};
        \foreach \x in {0.1, 0.2, ...,\WDD} {
        \draw[red,thick] (\x-0.1,{f(\x-0.1)}) -- (\x,{f(\x-0.1)}) -- (\x, {f(\x)});
        }
        \draw[red,thick] (\WD-0.1,{f(\WD-0.1)}) -- (\WD,{f(\WD-0.1)});
        
        \foreach \x  in {0.1, 0.2, ...,\WD}
            {
                \filldraw[red,opacity=0.5] (\x-0.1,{f(\x-0.1)}) -- (\x,{f(\x-0.1)}) -- (\x, {f(\WD-0.1)})  -- (\x-0.1, {f(\WD-0.1)}) -- cycle;
            }
        
        
        \filldraw[thick, blue,opacity=0.3] (0,{f(\WD-0.1)}) -- (\WD,{f(\WD-0.1)}) -- (\WD,0) -- (0,0);
        
        \node (WS) at (\WD/2, {f(\WD)/2}) {\footnotesize Worker Surplus};
        \node (FS) at (\WD/2, {1.4}) {\footnotesize Firm Surplus};
        \draw[thick,->] (FS) -- (2*\WD/5+0.1, 0.98);
\end{tikzpicture}
\caption{Marginal product falls slowly}\label{Figure-Scarce}
\end{subfigure}
\begin{subfigure}{.45\linewidth}
        \centering
        \begin{tikzpicture}[
        declare function={
        f(\x)=
        0.041148938207776*(\x)^7
        -0.345651080945326*(\x)^6
        +1.152170269817777*(\x)^5 
        -1.935646053293930*(\x)^4
        +1.580534083292260* (\x)^3 
        -0.274538084465002*(\x)^2 
        -0.845313510019283* (\x)
        + 1.5;},
        >=stealth,
        scale=2.5
        ]
        
        \newcommand{\WD}{2.1};
        \draw[->] (0,0) -- ({\WD+0.15},0) node[anchor=west] { $\ell$}; 
        \draw[->] (0,0) -- (0,1.65) node[anchor=east] {$\eta(\ell)$}; 
        
        \pgfmathsetmacro{\WDD}{\WD+0.1};
        \foreach \x in {0.1, 0.2,...,\WD} {
            \draw (\x,0.01) -- (\x,-0.01);
        }
        
        \node[below] at (0, 0) {\footnotesize $0$};
        \node[below] at (\WD, 0) {\scriptsize $L$};
        \node at (\WD/2,-0.15) {\text{\scriptsize Number of Workers}};
        \draw (-.15,0.825) node[rotate=90]{\text{\small \scriptsize{Marginal Product}}};

        \draw[line width=1.1pt, black, dashed, domain=0.1:\WD] plot (\x, {f(\x-0.1)+0.005}) node[right] {};

        \pgfmathsetmacro{\WDD}{\WD-0.1};
        \foreach \x in {0.1, 0.2, ...,\WDD} {
        \draw[red,thick] (\x-0.1,{f(\x-0.1)}) -- (\x,{f(\x-0.1)}) -- (\x, {f(\x)});
        }
        \draw[red,thick] (\WD-0.1,{f(\WD-0.1)}) -- (\WD,{f(\WD-0.1)});
        
        \foreach \x  in {0.1, 0.2, ...,\WD}
            {
                \filldraw[red,opacity=0.5] (\x-0.1,{f(\x-0.1)}) -- (\x,{f(\x-0.1)}) -- (\x, {f(\WD-0.1)})  -- (\x-0.1, {f(\WD-0.1)}) -- cycle;
            }
        
        \filldraw[thick, blue,opacity=0.3] (0,{f(\WD-0.1)}) -- (\WD,{f(\WD-0.1)}) -- (\WD,0) -- (0,0);
        
        \node (WS) at (\WD/2, {f(\WD)/2+0.03}) {\footnotesize Worker Surplus};
        \node (FS) at (\WD/2, {1.4}) {\footnotesize Firm Surplus};
        \draw[thick,->] (FS) -- (2*\WD/5, 0.65);
\end{tikzpicture}
\caption{Marginal product falls quickly}\label{Figure-Plentiful}
\end{subfigure}\vspace{0.1in}
\caption{(A) and (B) show the distribution of surplus under private wages when the marginal productivity falls slowly or quickly. In the latter case, workers have more to gain from wage transparency.}\label{Figure-Matching}
\end{figure}

Formalizing this comparative statics prediction, consider two markets $M_1$ and $M_2$ that are identical in all respects but one: they differ in the productivity of labor as captured in firms' revenues. In market $i$, firm $f$'s revenue function is $\tilde{v}_{f,i}$, and the marginal value of assigning worker $\ell$ optimally is then $\eta_{i}(\ell)$. We assume that labor is valuable in each market, in that $\eta_i(L+1)>0$ for each $i$, and that the two markets accrue the same gain from hiring the first worker, $\eta_1(1)=\eta_2(1)$. 
\begin{definition}\label{Definition-SteeplyDecreasing}
Market $M_2$ exhibits \textbf{more steeply decreasing returns to labor} than market $M_1$ if $   \eta_2(\ell)-\eta_2(\ell+1) \ge \eta_1(\ell)-\eta_1(\ell+1)$ for every $\ell$ in $\{1,\ldots,L\}$.\footnote{If the inequality is strict for some $\ell$, then we add the qualifier ``strictly.''}
\end{definition}
In each market, given gross substitutes, the marginal product of labor falls with each incremental worker; \Cref{Definition-SteeplyDecreasing} asserts that this fall is always more pronounced in $M_2$. Modulo integer issues, this definition translates into a standard condition on the second derivative of the total product being more negative in $M_2$.\footnote{The ``total product'' in each market with $\ell$ units of labor would be $\hat\Pi_i(\ell)\equiv\sum_{l=1}^\ell \eta_i(l)$. As the first worker in markets $M_1$ and $M_2$ generates the same gain, \Cref{Definition-SteeplyDecreasing} implies that $\hat\Pi_2(\cdot)$ must be a concave transformation of $\hat\Pi_1(\cdot)$, which is tantamount to the standard Arrow-Pratt comparison.}

We turn to the implications for how surplus is divided between workers and firms. Let $\Pi_i\equiv \sum_{\ell=1}^L \eta_i(\ell)$ denote the total surplus from the efficient assignment in market $M_i$. Let $\Pi^{\mathcal W}_{i}\equiv [L\eta_i(L+1),L\eta_i(L)]$ denote the set of potential workers' total surplus under private wages; recall from \Cref{Proposition-Wages} that each worker is paid the same, which is around the marginal product of the least productive worker. Firms capture the gap between total and workers' surplus; let $\Pi_i^{\mathcal F}$ denotes the set of potential firms' total surplus. We compare these surplus divisions between the two markets; when comparing sets, we use the strong set order denoted $\succeq_{SSO}$.\footnote{For sets $A,B\subseteq \Re$, $A\succeq_{SSO} B$ if for every $a\in A$ and $b\in B$, $\max\{a,b\}\in A$ and $\min\{a,b\}\in B$.}
\begin{proposition} \label{Proposition-Comparative-Wages}
Suppose $M_2$ exhibits more steeply decreasing returns to labor than $M_1$. Then the following hold about the distribution of surplus under private wages:
\begin{enumerate}[label=\emph{(\alph*)}]
    \item The total surplus in market $M_1$ is higher: $\Pi_1\geq \Pi_2$.
    \item Worker surplus in market $M_1$ must be higher: $\Pi^{\mathcal W}_{1}\succeq_{SSO} \Pi^{\mathcal W}_2$.
    \item Firm surplus in market $M_1$ must be lower: $\Pi^{\mathcal F}_{1}\preceq_{SSO} \Pi^{\mathcal F}_2$.
\end{enumerate}
Furthermore, all the orders above are strict if $M_2$ exhibits strictly more steeply decreasing returns to labor.
\end{proposition}
Although more steeply decreasing returns reduce both total surplus and worker surplus under private wages, \Cref{Proposition-Comparative-Wages} shows that this effect is more pronounced for workers.  Hence, as seen in (c), the residual surplus captured by firms is actually higher in $M_2$ than in $M_1$. If wage transparency enables workers to capture firms' profits, then workers have more to gain (and less to lose) in $M_2$ than $M_1$.

\section{Conclusion}\label{Section-Discussion}

This paper develops a portable framework for coalitional repeated games, which enables us to evaluate the role of dynamic incentives in coalitional behavior across a range of settings. Our analysis uncovers the importance of alignment: history dependence keeps coalitions in line if coalition members' interests are even slightly misaligned. Simple scapegoat schemes then deter coalitional deviations. However, if players in a coalition have completely aligned interests, they can secure a higher minmax payoff by effectively acting as a unitary agent. 

This perspective delivers additional insights. Strongly symmetric schemes do not deter coalitional deviations, thereby pushing towards the use of asymmetric punishments. Being able to transfer utility alone does not align interests; to the contrary, publicly observed transfers create a wedge between coalition partners and thereby undermine coalitions. However, the ability to make transfers under the table forges strong ties: a coalition that can do so is assured a high net payoff across PCE. 

In our applications, these secret side-payments cripple intertemporal incentives, reducing the set of supportable outcomes to the stage-game core. We use these results to identify conditions under which workers favor wage transparency in repeated labor-market matching. In the Supplementary Appendix, we also study repeated negotiations and show that history dependence can counter the tendency of veto players to become de facto dictators; we show therein that secret transfers restore their dictatorial power.

{ 
	\addcontentsline{toc}{section}{References}
	\setlength{\bibsep}{0.25\baselineskip}
	\bibliographystyle{jpe}
	\bibliography{coalitions}

@article{hatfield2005matching,
  title={Matching with Contracts},
  author={Hatfield, John William and Milgrom, Paul R},
  journal={American Economic Review},
  volume={95},
  number={4},
  pages={913--935},
  year={2005},
  publisher={American Economic Association}
}

@article{kelsocrawford1982,
  title={Job Matching, Coalition Formation, and Gross Substitutes},
  author={Kelso, Alexander S and Crawford, Vincent P},
  journal={Econometrica},
  pages={1483--1504},
  year={1982},
  publisher={JSTOR}
}

@article{cullen2023equilibrium,
  title={Equilibrium effects of pay transparency},
  author={Cullen, Zo{\"e} B and Pakzad-Hurson, Bobak},
  journal={Econometrica},
  volume={91},
  number={3},
  pages={765--802},
  year={2023},
  publisher={Wiley Online Library}
}

@article{cullen2024jep,
Author = {Cullen, Zo{\"e}},
Title = {Is Pay Transparency Good?},
Journal = {Journal of Economic Perspectives},
Volume = {38},
Number = {1},
Year = {2024},
Month = {February},
Pages = {153–80}}

@article{bernheim1987coalition,
  title={Coalition-proof nash equilibria i. concepts},
  author={Bernheim, B. Douglas and Peleg, Bezalel and Whinston, Michael D.},
  journal={Journal of Economic Theory},
  volume={42},
  number={1},
  pages={1--12},
  year={1987},
  publisher={Elsevier}
}

@article{greenporter1984,
  title={Noncooperative collusion under imperfect price information},
  author={Green, Edward J and Porter, Robert H},
  journal={Econometrica},
  pages={87--100},
  year={1984},
  publisher={JSTOR}
}

@article{Diamantoudi2015decentralized,
title = {Decentralized matching: The role of commitment},
journal = {Games and Economic Behavior},
volume = {92},
pages = {1-17},
year = {2015},
issn = {0899-8256},
author = {Effrosyni Diamantoudi and Eiichi Miyagawa and Licun Xue}
}

@unpublished{zhang2016,
    author = {Mu Zhang and Jie Zheng},
    title = {Multi-period Matching with Commitment},
    note = {Working Paper},
    year = {2016}
}

@article{fudenberg1994folk,
  title={The Folk Theorem with Imperfect Public Monitoring},
  author={Fudenberg, Drew and Levine, David and Maskin, Eric},
  journal={Econometrica},
  volume={62},
  number={5},
  pages={997--1039},
  year={1994}
}

@article{liu2023stability,
  title={Stability in Repeated Matching Markets},
  author={Liu, Ce},
  journal={Theoretical Economics},
  volume={18},
  number={4},
  pages={1711--1757},
  year={2023},
  publisher={Wiley Online Library}
}

@article{miller2013theory,
  title={A theory of disagreement in repeated games with bargaining},
  author={Miller, David A and Watson, Joel},
  journal={Econometrica},
  volume={81},
  number={6},
  pages={2303--2350},
  year={2013},
  publisher={Wiley Online Library}
}

@article{liu2024self,
  title={Self-Enforced Job Matching},
  author={Liu, Ce and Wang, Ziwei and Zhang, Hanzhe},
  journal={American Economic Journal: Microeconomics},
  year={2025}
}

@unpublished{bardhi2024early,
  title={Early-Career Discrimination: Spiraling or Self-Correcting?},
  author={Bardhi, Arjada and Guo, Yingni and Strulovici, Bruno},
  note={Working Paper},
  year={2025}
}

@article{greenberg1989application,
  title={An application of the theory of social situations to repeated games},
  author={Greenberg, Joseph},
  journal={Journal of Economic Theory},
  volume={49},
  number={2},
  pages={278--293},
  year={1989},
  publisher={Elsevier}
}

@book{greenberg1990theory,
  title={The theory of social situations: an alternative game-theoretic approach},
  author={Greenberg, Joseph},
  year={1990},
  publisher={Cambridge University Press}
}

@article{wen1994folk,
  title={The ``Folk Theorem'' for Repeated Games with Complete Information},
  author={Wen, Quan},
  journal={Econometrica},
  pages={949--954},
  year={1994},
  publisher={JSTOR}
}

@article{moulin1982cores,
	Author = {Moulin, Herve and Peleg, Bezalel},
	Journal = {Journal of Mathematical Economics},
	Number = {1},
	Pages = {115--145},
	Publisher = {Elsevier},
	Title = {Cores of effectivity functions and implementation theory},
	Volume = {10},
	Year = {1982}}

@unpublished{safronov2018contestable,
  title={Contestable norms},
  author={Safronov, Mikhail and Strulovici, Bruno},
  note={Working Paper},
  year={2018}
}

@book{von1945theory,
	Author = {Von Neumann, John and Morgenstern, Oskar},
	Publisher = {Princeton University Press Princeton, NJ},
	Title = {Theory of Games and Economic Behavior},
	Year = {1945}}

@article{vartiainen2011dynamic,
	Author = {Vartiainen, Hannu},
	Journal = {Journal of Economic Theory},
	Number = {2},
	Pages = {672--698},
	Publisher = {Elsevier},
	Title = {Dynamic Coalitional Equilibrium},
	Volume = {146},
	Year = {2011}}

@article{sorin1986,
	Author = {Sylvain Sorin},
	Issn = {0364765X, 15265471},
	Journal = {Mathematics of Operations Research},
	Number = {1},
	Pages = {147-160},
	Publisher = {INFORMS},
	Title = {On Repeated Games with Complete Information},
	Volume = {11},
	Year = {1986}}

@article{rubinstein1980strong,
	Author = {Rubinstein, Ariel},
	Journal = {International Journal of Game Theory},
	Number = {1},
	Pages = {1--12},
	Publisher = {Springer},
	Title = {Strong Perfect Equilibrium in Supergames},
	Volume = {9},
	Year = {1980}}

@article{rosenthal1972cooperative,
	Author = {Rosenthal, Robert W.},
	Journal = {Journal of Economic Theory},
	Number = {1},
	Pages = {88--101},
	Publisher = {Elsevier},
	Title = {Cooperative Games in Effectiveness Form},
	Volume = {5},
	Year = {1972}}

@article{ray2015farsighted,
	Author = {Ray, Debraj and Vohra, Rajiv},
	Journal = {Econometrica},
	Number = {3},
	Pages = {977--1011},
	Publisher = {Wiley Online Library},
	Title = {The Farsighted Stable Set},
	Volume = {83},
	Year = {2015}}

@book{mailathrepeated,
	Address = {New York, NY},
	Author = {Mailath, George and Samuelson, Larry},
	Publisher = {Oxford University Press},
	Title = {Repeated Games and Reputations},
	Year = {2006}}

@article{mailath2024trust,
    author = {Cole, Harold L and Krueger, Dirk and Mailath, George J and Park, Yena},
    title = {Trust in Risk Sharing: A Double-Edged Sword},
    journal = {Review of Economic Studies},
    volume = {91},
    number = {3},
    pages = {1448-1497},
    year = {2024},
    month = {07},
    issn = {0034-6527}
}

@article{KR2003,
	Author = {Konishi, Hideo and Ray, Debraj},
	Journal = {Journal of Economic Theory},
	Number = {1},
	Pages = {1--41},
	Publisher = {Elsevier},
	Title = {Coalition Formation as A Dynamic Process},
	Volume = {110},
	Year = {2003}}

@article{kotowski2024,
title = {A perfectly robust approach to multiperiod matching problems},
journal = {Journal of Economic Theory},
volume = {222},
pages = {105919},
year = {2024},
issn = {0022-0531},
author = {Maciej H. Kotowski}
}

@unpublished{rostek2024matching,
  title={Matching with Strategic Consistency},
  author={Rostek, Marzena J and Yoder, Nathan},
  note={Working Paper},
  year={2024}
}

@article{kadam2018time,
	Author = {Kadam, Sangram V. and Kotowski, Maciej H.},
	Journal = {Games and Economic Behavior},
	Pages = {1--20},
	Publisher = {Elsevier},
	Title = {Time Horizons, Lattice Structures, and Welfare in Multi-Period Matching Markets},
	Volume = {112},
	Year = {2018}}

@article{kadam2018multiperiod,
	Author = {Kadam, Sangram V. and Kotowski, Maciej H.},
	Journal = {International Economic Review},
	Number = {4},
	Pages = {1927--1947},
	Publisher = {Wiley Online Library},
	Title = {Multiperiod Matching},
	Volume = {59},
	Year = {2018}}

@article{gomes2005dynamic,
	Author = {Gomes, Armando and Jehiel, Philippe},
	Journal = {Journal of Political Economy},
	Number = {3},
	Pages = {626--667},
	Publisher = {JSTOR},
	Title = {Dynamic Processes of Social and Economic Interactions: On the Persistence of Inefficiencies},
	Volume = {113},
	Year = {2005}}

@article{fudenbergmaskin1991,
	Author = {Fudenberg, Drew and Maskin, Eric},
	Journal = {Journal of Economic Theory},
	Number = {2},
	Pages = {428---438},
	Publisher = {Elsevier},
	Title = {On the Dispensability of Public Randomization in Discounted Repeated Games},
	Volume = {53},
	Year = {1991}}

@article{fudenberg1986folk,
	Author = {Fudenberg, Drew and Maskin, Eric},
	Journal = {Econometrica},
	Pages = {533--554},
	Publisher = {JSTOR},
	Title = {The Folk Theorem in Repeated Games with Discounting or with Incomplete Information},
	Year = {1986}}

@article{farrell1989renegotiation,
	Author = {Farrell, Joseph and Maskin, Eric},
	Journal = {Games and Economic Behavior},
	Number = {4},
	Pages = {327--360},
	Publisher = {Elsevier},
	Title = {Renegotiation in Repeated Games},
	Volume = {1},
	Year = {1989}}

@article{levin2003relational,
  title={Relational incentive contracts},
  author={Levin, Jonathan},
  journal={American Economic Review},
  volume={93},
  number={3},
  pages={835--857},
  year={2003},
  publisher={American Economic Association}
}

@unpublished{aliliu2026note,
	Author = {Ali, S. Nageeb and Ce Liu},
	Note = {Working Paper},
	Title = {On Conservative Stable Standard of Behavior and Perfect Coalitional Equilibrium},
	Year = {2026}}

@article{doval2022dynamically,
  title={Dynamically Stable Matching},
  author={Doval, Laura},
  journal={Theoretical Economics},
  volume={17},
  number={2},
  pages={687--724},
  year={2022},
  publisher={Wiley Online Library}
}

@article{demarzo1992coalitions,
	Author = {DeMarzo, Peter M.},
	Journal = {Games and Economic Behavior},
	Number = {1},
	Pages = {72--100},
	Publisher = {Elsevier},
	Title = {Coalitions, Leadership, and Social Norms: The Power of Suggestion in Games},
	Volume = {4},
	Year = {1992}}

@article{genicotray2003,
    author = {Genicot, Garance and Ray, Debraj},
    title = {Group Formation in Risk-Sharing Arrangements},
    journal = {Review of Economic Studies},
    volume = {70},
    number = {1},
    pages = {87-113},
    year = {2003},
    month = {01},
    }

@article{damiano2005stability,
	Author = {Damiano, Ettore and Lam, Ricky},
	Journal = {Games and Economic Behavior},
	Number = {1},
	Pages = {34--53},
	Publisher = {Elsevier},
	Title = {Stability in Dynamic Matching Markets},
	Volume = {52},
	Year = {2005}}

@article{corbae2003directed,
	Author = {Corbae, Dean and Temzelides, Ted and Wright, Randall},
	Journal = {Econometrica},
	Number = {3},
	Pages = {731--756},
	Publisher = {Wiley Online Library},
	Title = {Directed Matching and Monetary Exchange},
	Volume = {71},
	Year = {2003}}

@article{chwe1994farsighted,
	Author = {Chwe, Michael},
	Journal = {Journal of Economic Theory},
	Number = {2},
	Pages = {299--325},
	Publisher = {Elsevier},
	Title = {Farsighted Coalitional Stability},
	Volume = {63},
	Year = {1994}}

@article{BS2009,
	Author = {Bernheim, B. Douglas and Slavov, Sita N.},
	Journal = {Review of Economic Studies},
	Number = {1},
	Pages = {33--62},
	Publisher = {Oxford University Press},
	Title = {A Solution Concept for Majority Rule in Dynamic Settings},
	Volume = {76},
	Year = {2009}}

@article{bernheim1989collective,
	Author = {Bernheim, B. Douglas and Ray, Debraj},
	Journal = {Games and Economic Behavior},
	Number = {4},
	Pages = {295--326},
	Publisher = {Elsevier},
	Title = {Collective Dynamic Consistency in Repeated Games},
	Volume = {1},
	Year = {1989}}

@article{barron2021use,
  title={The use and misuse of coordinated punishments},
  author={Barron, Daniel and Guo, Yingni},
  journal={Quarterly Journal of Economics},
  volume={136},
  number={1},
  pages={471--504},
  year={2021},
  publisher={Oxford University Press}
}

@article{baron1989bargaining,
	Author = {Baron, David P. and Ferejohn, John A.},
	Journal = {American Political Science Review},
	Number = {4},
	Pages = {1181--1206},
	Publisher = {JSTOR},
	Title = {Bargaining in Legislatures},
	Volume = {83},
	Year = {1989}}

@incollection{aumann1959acceptable,
	Address = {Princeton, NJ},
	Author = {Aumann, Robert J.},
	Booktitle = {Contributions to the Theory of Games IV},
	Editor = {H. W. Kuhn and R. D. Luce},
	Pages = {287},
	Publisher = {Princeton University Press},
	Title = {Acceptable Points in General Cooperative n-Person Games},
	Volume = {4},
	Year = {1959}}

@article{abreu1990toward,
	Author = {Abreu, Dilip and Pearce, David and Stacchetti, Ennio},
	Journal = {Econometrica},
	Pages = {1041--1063},
	Publisher = {JSTOR},
	Title = {Toward A Theory of Discounted Repeated Games with Imperfect Monitoring},
	Year = {1990}}

@article{abreu1994folk,
	Author = {Abreu, Dilip and Dutta, Prajit K. and Smith, Lones},
	Journal = {Econometrica},
	Number = {4},
	Pages = {939--948},
	Publisher = {JSTOR},
	Title = {The Folk Theorem for Repeated Games: A NEU Condition},
	Volume = {62},
	Year = {1994}}

@article{abreu1988opc,
	Author = {Abreu, Dilip},
	Journal = {Econometrica},
	Number = {2},
	Pages = {383--396},
	Publisher = {Citeseer},
	Title = {On the Theory of Infinitely Repeated Games with Discounting},
	Volume = {56},
	Year = {1988}}
}

\addtocontents{toc}{\protect\setcounter{tocdepth}{1}}
\appendix
\renewcommand{\theequation}{\arabic{equation}}
  \renewcommand{\thesection}{\Alph{section}}

\section{Appendix} \label{Section-MainProofs}

The main appendix contains proofs for Theorems \ref{Theorem-NTU}, \ref{Theorem-Stationary}, \ref{Theorem-Symmetric}, and \ref{Theorem-Secret}. All other proofs are in the Supplementary Appendix. Throughout our analysis, we use sequences of play to convexify payoffs, following standard arguments from \cite{sorin1986} and \cite{fudenbergmaskin1991}. Below, we reproduce the statement that we invoke in our arguments.
\begin{lemma} \label{lemma:NTU-folk-sequence-approx}
	\textbf{(Lemma 2 of \citealp{fudenbergmaskin1991})} Let $X$ be a convex polytope in $\Re^n$ with vertices $x^1, \ldots, x^K$. For all $\epsilon>0$, there exists a $\underline{\delta}<1$ such that for all $\underline{\delta} <\delta<1$, and any $x \in X$, there exits a sequence  $(x_\tau )_{\tau =0}^\infty$ drawn from $\{x^1, \ldots, x^K\}$, such that $(1-\delta) \sum_{\tau =0}^\infty \delta^\tau x_\tau = x$ and at any $t$, $\norm{x - (1-\delta )\sum_{\tau = t}^\infty \delta^{\tau-t} x_\tau} < \epsilon$.
\end{lemma}

\subsection{Proof of \cref{Theorem-NTU} on p. \pageref{Theorem-NTU}} \label{Section-Proof-NTU}
\paragraph{A Preliminary Result.} A \textit{blocking plan} by coalition ${C}$ from a plan $\sigma$ is a function $\alpha:{\histories} \rightarrow \alternatives$ such that $\alpha(h) \in E_{C} ( \sigma (h) ) $ for every history ${h}\in {\histories}$. After each history $h$, the blocking plan $\alpha$ generates a path $( \alpha(h), \alpha(h,\alpha(h), \{C\} ),\ldots  )$ that is distinct from the one generated by $\sigma$. We use $U_i(h|\alpha)$ to denote player $i$'s normalized discounted payoff from that path. The blocking plan $\alpha$ is {profitable} if there exists a history ${h}$ such that $U_i({h} |\alpha) > U_i({h}|\sigma)$ for all $i\in C$. Below, we say that a coalition is in the alignment partition if it corresponds to a coalition $C(i)$ for some player $i$.

\begin{lemma}\label{Lemma:One-Shot-Deviation-NTU}
If $\sigma$ is a PCE, then no coalition in the alignment partition has a profitable blocking plan.
\end{lemma}

\begin{proof}
Consider a  plan $\sigma$ from which coalition $C(i^*)$ has a profitable blocking plan $\alpha$ for some $i^*\in \players$. In particular, there exists a history $h \in \histories$ such that $U_{i}(h|\alpha) > U_i(h|\sigma)$ for every $i\in C(i^*)$. We show that coalition $C(i^*)$ must then have a profitable block from the plan $\sigma$ at some history, so $\sigma$ is not a PCE.

Since the set of alternatives $\alternatives$ is compact and $v:A\rightarrow \Re^n$ is continuous, the plan $\sigma$ has bounded continuation values for all players.
Given discounting, the standard one-shot deviation principle applies. Therefore, there exists a  history $\hat{h} \in \histories$ such that 
\begin{align*}
(1-\delta)  v_{i^*} \big (\alpha(\hat{h}) \big ) + \delta U_{i^*} \big ( \hat{h}, \alpha (\hat{h}), \{C(i^*)\} \big| \sigma \big ) > U_{i^*}(\hat{h}|\sigma).
\end{align*}
Since players in $C(i^*)$ have equivalent payoffs, for each $j \in C(i^*)$ there exists $\lambda_{ji^*}>0$ and $\mu_{ji^*}\in \Re$ such that $v_j(a) = \lambda_{ji^*} v_{i^*}(a) + \mu_{ji^*}$ and all alternatives $a\in \alternatives$; in addition, for every $j\in C(i^*)$, the discounted payoffs satisfy $U_j(h|\sigma) = \lambda_{ji^*} U_{i^*}(h|\sigma) + \mu_{ji^*}$ at every history $h\in \histories$. Substituting into the inequality above, we have
\begin{align*}
(1-\delta)  v_{j } \big (\alpha(\hat{h}) \big ) + \delta U_{j } \big ( \hat{h}, \alpha (\hat{h}), \{C(i^*)\} \big| \sigma \big ) > U_{j }(\hat{h}|\sigma) \text{ for all } j\in C(i^*).
\end{align*}
Therefore, coalition $C(i^*)$ has a profitable block at history $\hat{h}$.
\end{proof}

\paragraph{Proof of \cref{Theorem-NTU}.}

{ \noindent \underline{Part 1}: \emph{For every $\delta \geq 0$, every PCE gives each player $i$ a payoff of at least $\cm_i$.}}

\noindent We establish this claim by proving its contrapositive: let $\sigma$ be a plan, and suppose there exists a player ${i}^*$ that satisfies $U_{i^*}(\emptyset|\sigma) < \cm_{i^*}$. We show that $\sigma$ cannot be a PCE. Given \cref{Lemma:One-Shot-Deviation-NTU},  it suffices to show that coalition  $C(i^*)$ has a profitable blocking plan.

Consider the following {blocking plan} $\alpha$ for coalition $C(i^*)$: at every history $h$, coalition $C(i^*)$ chooses its myopic best response to the default alternative, $\alpha(h) \in \argmax_{a'\in E_{C(i^*)}(\sigma(h))}  v_{i^*}(a')$. By the definition of $\cm_{i^*}$, $v_{i^*}( \alpha (h) )  \ge  \cm_{i^*}$ for every history $h$, so player $i^*$'s continuation value from period $0$ must be higher: $U_{i^*} (\emptyset|\alpha ) > U_{i^*} (\emptyset|\sigma )$. Given that all players $j \in C(i^*)$ have equivalent utilities,  $U_j(\emptyset|\alpha) > U_{j}(\emptyset|\sigma)$ for all $j\in C(i^*)$, so $\alpha$ is a profitable blocking plan for coalition $C(i^*)$. 
\medskip

{\noindent \underline{Part 2}: \emph{For every $v\in {\fcr}$, there is a $\underline{\delta}<1$ such that for every $\delta\in(\underline{\delta},1)$, there exists a PCE with discounted payoff equal to $v$.}}

\noindent Fix $v^* \in \fcr$. First, observe that for any pair of players $i,j$ such that $j\notin C(i)$, their payoffs satisfy the non-equivalent utilities (NEU) condition. By Lemma $1$ and Lemma $2$ of \cite*{abreu1994folk}, we can find \textit{coalition-specific punishments} for $v^*$: there exist payoff vectors $\{v^{C(i)}\}_{i=1}^{n}\subseteq \fcr $ such that $v_{i}^{C(i)}<v^*_{i}$ for all $i\in \players$, and $v_{j}^{C(i)} > v_{j}^{C(j)}$ for all $j\notin  C(i)$.

Second, let us define \emph{coalitional minmaxing alternatives}: for each coalition $C(i)$, let $\cma_{C(i)} \in \argmin_{a\in \alternatives} \max_{a'\in E_{C(i)}(a)} v_j(a')$ for some $j\in C(i)$---note that the specific choice of $j\in C(i)$ in the definition does not matter given the equivalent payoffs within $C(i)$---as the alternative that will be used to minmax coalition $C(i)$. Since $A$ is compact, $v$ is continuous, and $E_{C(i)}(\cdot)$ is continuous and compact-valued, by Berge's maximum theorem, $\cma_{C(i)}$ is well-defined for each $i\in \players$. By construction, for all $i\in \players$ and $a'\in E_{C(i)}(\cma_{C(i)})$, $v_i(a') \le \cm_i$. Observe that the reflexivity of effectivity correspondences implies that $\cma_{C(i)}~\in~E_{C(i)}(\cma_{C(i)})$ and therefore, ${v}_{i}(\cma_{C(i)} ) \le \cm_i$.

Given these payoffs and punishments, let $\kappa\in (0,1)$ be such that for every $\tilde\kappa \in [\kappa,1]$, the following is true for every $i$:  
\begin{align}
	(1-\tilde\kappa)v_i(\cma_{C(i)})+\tilde\kappa v_i^{C(i)}&>\cm_i \label{Equation-Kapppa}\\
\text{For every $i\in \players$ and $j\notin C(i)$:}\hspace{0.5cm}(1-\tilde\kappa)v_{j}(\cma_{C(i)})+\tilde\kappa v_j^{C(i)} &> (1-\tilde\kappa)\cm_j+\tilde\kappa v_j^{C(j)}	\label{Equation-Kappa}
\end{align}

Inequality \eqref{Equation-Kapppa} implies that every player $j\in C(i)$ is willing to bear the cost of $v_j(\cma_{C(i)})$ with the promise of transitioning into their coalition-specific punishment rather than staying at their coalitional minmax, where the promise is discounted at $\tilde\kappa$. Similarly, inequality \eqref{Equation-Kappa} implies that player $j$ is willing to bear the cost of minmaxing any coalition with whom $j$ does not share equivalent utilities, given the promise of transitioning into coalition $C(i)$'s specific punishment rather than her own, when the post-minmaxing phase payoffs are discounted at $\tilde\kappa$. Each inequality holds at $\tilde\kappa=1$ for all $i$ and $j\notin C(i)$. Since the set of players is finite, there exists a value of $\kappa \in (0,1)$ such that the inequality holds for all $\tilde\kappa\in [\kappa,1]$, $i\in \players$ and $j\notin C(i)$. Let $L(\delta)\equiv \left\lceil{\frac{\log \kappa}{\log \delta}}\right\rceil$ where $\left\lceil\cdot\right\rceil$ is the ceiling function. Observe that $ \delta^{L(\delta)} \in [\delta^{ \frac{\log \kappa}{ \log \delta} +1 } , \delta^{ \frac{\log \kappa}{ \log \delta}}] = [\delta \kappa, \kappa ]$. Therefore, $\lim_{\delta\rightarrow 1} \delta^{L(\delta)}=\kappa$.\medskip 

Since $\{v^{C(i)}\}_{i=1}^n\cup\{v^*\} \subseteq \conv{\{v(a): a\in \alternatives\} } \subseteq \Re^n$, by Carath\'{e}odory's theorem, there exist $\{\hat{a}^1,\ldots, \hat{a}^K\} \subseteq A$ for some integer $K$, such that $\{v^{C(i)}\}_{i=1}^n\cup\{v^*\} \subseteq \conv{\{v(\hat{a}^k): k=1,\ldots,K \} }$. Define $ \mathcal{I} \equiv \{C(i)\}_{i=1}^n,\; \text{ and }\; \hat{\mathcal{I}} \equiv \{C(i)\}_{i=1}^n \cup \{ * \}$.
\cref{lemma:NTU-folk-sequence-approx} then guarantees that for any $\epsilon >0$, there exists $\underline\delta\in (0,1)$ such that for all $\delta \in (\underline\delta,1) $, there exist sequences $\big \{ (a^{S,\tau})_{\tau =0}^\infty : S \in \hat{\mathcal{I}} \big \}$ such that for each $S \in \hat{\mathcal{I}}$ and $t$, $(1-\delta) \sum_{\tau =0}^\infty \delta^\tau v(a^{S, \tau}) = v^S$ and $\norm{v^S - (1-\delta )\sum_{\tau = t}^\infty \delta^{\tau-t} v(a^{S,\tau}) } < \epsilon$. We fix an 
\[
\epsilon<\left(1-\kappa\right)\min\Big\{\min_{S\in \hat{\mathcal{I}}, i\in \players \backslash S} (v_i^{S}-v_i^{C(i)}),\; \min_{i \in \players}v^{C(i)}_i-\cm_i \Big\},
\]
and given that $\epsilon$, consider $\delta$ exceeding the appropriate $\underline\delta$.\medskip

\noindent We now describe the plan that supports $v^*$. Consider the automaton $(W, w(*,0),f,\gamma )$: 
\begin{itemize}
	\item $W\equiv \big \{ w(d, \tau) | d  \in  \hat{\mathcal{I}}, \tau \ge 0 \big \} \cup\{ \underline{w}(S, \tau )|S \in \mathcal{I}, 0\le \tau <L(\delta)\}$ is the set of possible states and $w(*,0)$ is the initial state;
	\item $f:W \rightarrow \outcomesNTU $ is the output function, where $f(w(d, \tau )) = (a^{d,\tau},\emptyset )$ and $f(\underline{w}(S, \tau ))= ( \cma_{S}, \emptyset )$.
	\item $\gamma: W \times \outcomesNTU \rightarrow W$ is the transition function. For any collection of blocking coalitions $B\in \mathcal{B}$, let $\hat{C}(B) = \cup_{C\in B}C$ denote their union. For states of the form $w(d,\tau)$, the transition is
    	\begin{align*}
    	\gamma \big( w(d,\tau), (a, B) \big) =
    	\begin{cases}
    	\underline{w}(C(j^{*}),0) & \text{if $B\ne \emptyset$ , where $j^{*} = \min_{ }\hat{C}(B) $ } \\
    		w(d, \tau+1) & \text{otherwise}
    	\end{cases}
    	\end{align*}		
For states of the form $\{\underline{w}(S,\tau)|S\in \mathcal{I}, 0\le \tau <L(\delta)-1\}$,
    	\begin{align*}
    	\gamma \big( \underline{w} (S, \tau), (a, B) \big) =
    	\begin{cases}
    	\underline{w}(C(j^{*}),0) & \text{if $\hat{C}(B) \nsubseteq S$ , where $j^{*} = \min_{} (\hat{C}(B)\backslash S) $ } \\
    	\underline{w}(S,0) & \text{if $\hat{C}(B) \subseteq S$ and $\hat{C}(B)\ne \emptyset$}\\
    	\underline{w}(S, \tau+1) & \text{otherwise}
    	\end{cases}
    	\end{align*}
For states of the form $\{\underline{w}(S,L(\delta)-1)|S\in \mathcal{I}\}$, the transition is
		\begin{align*}
		\gamma \big( \underline{w}(S,L(\delta)-1), (a, B) \big) =
		\begin{cases}
		\underline{w}(C(j^{*}),0) & \text{if }\hat{C}(B) \nsubseteq S,\\
        &\; \text{ where } j^{*} = \min_{}(\hat{C}(B)\backslash S)   \\
		\underline{w}(S,0) & \text{if $\hat{C}(B) \subseteq S$ and $\hat{C}(B)\ne \emptyset$}\\
		w(S,0) & \text{otherwise}
		\end{cases}
		\end{align*}
	\end{itemize}
	The plan represented by the above automaton yields payoff profile $v^*$. 
	By construction, $\norm{v^d - V(w(d,\tau ))}  < \epsilon$ for all $(d,\tau)$; in addition, for $\tau = 0, 1, \ldots, L(\delta)$,
\begin{align*}
	 V(\underline{w}(S, \tau)) = (1-\delta^{L(\delta) - \tau})  v(\cma_S ) & + \delta^{L(\delta)-\tau}V( w(S, 0) ).
\end{align*}

\noindent Below, we show that this plan is a PCE by showing that no coalition can profitably block in any state of this automaton.
	
{\flushleft \textbf{States of the form $w(d, \tau )$:}}
Set $b > \max_{ a\in \alternatives,i\in \players} v_{i} (a)$. Consider coalition $C$ blocking and implementing the alternative $a$. Let $j^{*} = \min C $. For all $\tau$, without the blocking $j^{*}$ obtains a payoff greater than $v_{j^{*}}^{d} - \epsilon$. By participating in the blocking, $j^{*}$ obtains a payoff less than
$
        (1-\delta)b+\delta V_{j^{*}}( \underline{w} (C(j^{*}),0)) =  (1-\delta)b+\delta[(1-\delta^{L(\delta)}) {v}_{j^*}(\cma_{C(j^*)} ) +\delta^{L(\delta)} v^{C(j^*)}_{j^*}].
$
The blocking is not profitable if the preceeding term is no more than $v_{j^{*}}^{d}  -\epsilon$. We prove that this is true in two separate cases.

First consider the case where $d\in \hat{\mathcal{I}}\backslash \{C(j^*)\}$. Observe that 
\begin{align*}
	&\lim_{\delta\rightarrow 1} (1-\delta)b + \delta\left[(1-\delta^{L(\delta)}) {v}_{j^*}(\cma_{C(j^*)} ) +  \delta^{L(\delta)} v^{C(j^*)}_{j^*}\right]\\
 = & \lim_{\delta\rightarrow 1}\left[(1-\delta^{L(\delta)}) {v}_{j^*}(\cma_{C(j^*)}) +\delta^{L(\delta)} v^{C(j^*)}_{j^*}\right] <  v^{C(j^*)}_{j^*},
\end{align*}
where the inequality follows from ${v}_{j^*}(\cma_{C(j^*)} ) \le \cm_{j^*} < v^{C(j^*)}_{j^*}$. Because $\epsilon$ by construction is strictly less than $v_{j^{*}}^{d}-	v^{C(j^*)}_{j^*}$, it follows that payoff from blocking is less than $v_{j^{*}}^{d}-\epsilon$ when $\delta$ is sufficiently large. 

Now suppose that $d = C(j^*)$. The blocking payoff being less than $v_{j^{*}}^{C(j^*)}-\epsilon$ can be re-written as $(1- \delta )( b - v_{j^*}^{C(j^*)} ) + \epsilon \le \delta(1- \delta^{L(\delta)} ) (v^{C(j^*)}_{j^*} - {v}_{j^*}(\cma_{C(j^*)} ) )$.
As $\delta \rightarrow 1$, the LHS converges to $\epsilon$. Because $\lim_{\delta\rightarrow 1}\delta^{L(\delta)}=\kappa$, the RHS converges to $(1-\kappa)(v^{C(j^*)}_{j^*} - {v}_{j^*}(\cma_{C(j^*)} ) )$. By definition of $\epsilon$, the above inequality holds, and therefore, no coalition can profitably block if $\delta$ is sufficiently high.
	
{\flushleft \textbf{States of the form $\underline{w}(S, \tau )$:}} We first consider the case where $C \subseteq S$ and $C\ne \emptyset$. Choose an arbitrary $i\in C$ and we will show that the blocking is not profitable for $i$. By the definition $\cma_S$, coalition $C$ cannot generate a payoff of more than $\cm_i$ for player $i$, so $i$ finds the blocking to be unprofitable if 
\begin{align}\label{Equation-iminmaxnodeviate}
    (1-\delta^{L(\delta)- \tau}) v_{i} (\cma_{S})  + \delta^{L(\delta) - \tau}v_{i}^{S}\geq 	(1-\delta) \cm_i + \delta(1-\delta^{L(\delta)})v_{i} (\cma_S)+\delta^{L(\delta)+1}v_{i}^{S}.
\end{align}
Because $v_i^S>	\cm_i \geq v_{i} (\cma_S)$, it suffices to show that the inequality above holds at $\tau=0$. Re-arranging terms yields that $(1-\delta)(1-\delta^{L(\delta)}) v_{i} (\cma_S)  + (1-\delta)\delta^{L(\delta) }v_{i}^{S}\geq 	(1-\delta) \cm_i$, and then dividing by $(1-\delta)$ yields $(1-\delta^{L(\delta)}) v_{i} (\cma_S)  + \delta^{L(\delta) }v_{i}^{S}\geq 	\cm_i$. Let us verify that this inequality holds for sufficiently high $\delta$. Taking $\delta\rightarrow 1$ yields \eqref{Equation-Kapppa}, which is true. Hence \eqref{Equation-iminmaxnodeviate} holds for sufficiently high $\delta$.
	
Next we consider the case where $C\nsubseteq S$. By construction, $j^*\notin S$. Player $j^{*}$ finds blocking to be unprofitable if
\begin{equation}\label{Equation-MinmaxingIC}
    (1-\delta^{L(\delta)- \tau}) v_{j^{*}} (\cma_S)  + \delta^{L(\delta) - \tau}v_{j^{*}}^{S}\geq 	(1-\delta) b + \delta(1-\delta^{L(\delta)}) {v}_{j^*}(\cma_{C(j^*)} ) + \delta^{L(\delta)+1}v_{j^{*}}^{C(j^{*})}.
    \end{equation}
We prove that this inequality is satisfied if $\delta$ is sufficiently high. Examining the LHS, observe that for all $\tau$ such that $0\leq\tau\leq L(\delta)-1$, 
\begin{align*}
	\lim_{\delta\rightarrow 1}\left[	(1-\delta^{L(\delta)- \tau}) v_{j^{*}} (\cma_S)  + \delta^{L(\delta) - \tau}v_{j^{*}}^{S}\right]&=\lim_{\delta\rightarrow 1}	\left[\left(1-\frac{\kappa}{\delta^\tau}\right) v_{j^{*}} (\cma_S)  + \frac{\kappa}{\delta^\tau}v_{j^{*}}^{S}\right]\\
	&=(1-\tilde{\kappa}) v_{j^{*}} (\cma_S)  + \tilde\kappa v_{j^{*}}^{S}
\end{align*}
for some $\tilde\kappa \in [\kappa,1]$. Examining the RHS of \eqref{Equation-MinmaxingIC}, observe that
\begin{align*}
	&\lim_{\delta\rightarrow 1}\left[	(1-\delta) b + \delta(1-\delta^{L(\delta)}) {v}_{j^*}(\cma_{C(j^*)}) +\delta^{L(\delta)+1} v_{j^{*}}^{C(j^{*})} \right]\\
 = & \lim_{\delta\rightarrow 1}\left[(1-\delta^{L(\delta)}){v}_{j^*}(\cma_{C(j^*)}) + \delta^{L(\delta)}v_{j^{*}}^{C(j^{*})}\right]\\
	= &\; (1-\kappa) {v}_{j^*}(\cma_{C(j^*)}) +\kappa v_{j^*}^{C(j^*)}\; \le \; (1-\kappa) \cm_{j^*} +\kappa v_{j^{*}}^{C(j^*)} \; \leq \; (1-\tilde\kappa) \cm_{j^*} + \tilde\kappa v_{j^*}^{C(j^*)},
\end{align*}
where the first equality follows from taking limits, the second from  $\lim_{\delta\rightarrow 1} \delta^{L(\delta)}=\kappa$, the first weak inequality follows from ${v}_{j^*}(\cma_{C(j^*)}) \le \cm_{j^*}$, the second weak inequality follows from $\tilde\kappa\geq \kappa$ and $\cm_{j^*}<v_{j^{*}}^{C(j^{*})}$. Since $\tilde\kappa\in  [\kappa,1]$ and $j^*\notin S$, \eqref{Equation-Kappa} delivers that $(1-\tilde{\kappa}) v_{j^{*}} (\cma_S)  + \tilde\kappa v_{j^{*}}^{S}$ is strictly higher than $(1-\tilde\kappa) \cm_{j^*}+\tilde\kappa v_{j^{*}}^{C(j^*)}$. This term guarantees that \eqref{Equation-MinmaxingIC} holds for sufficiently high $\delta$. \qed

\subsection{Proof of \cref{Theorem-Stationary} on p. \pageref{Theorem-Stationary}}
Since the stage game exhibits default-independent power, for each $C\in \coalitions$ there exists set $D(C)\subseteq \Re^{|C|}$ such that $\payoffs_C(a) = D(C)\cup \{v_C(a)\}$ for all $a\in A$. To study stationary PCEs, we define the analogue of the self-generation map  \citep{abreu1990toward}. In a stationary PCE, at any history the prescribed current-period payoff and continuation value are the same. Accordingly, we define the stationary self-generation map as follows. Let $\payoffs^{\circ} \equiv \left\{v(a):a\in \alternatives\right\}$ denote the set of stage-game payoff profiles. For any set $Y\subseteq \payoffs^{\circ} \subseteq \Re^n$, define
\[
\Phi_{\delta}(Y) \equiv \big\{y \in Y: \forall C\in \coalitions \text{ and } 
z_C\in D(C), \;\exists y'\in Y \text{ and }  i\in C \text{ s.t. } y_i \ge (1-\delta)z_i + \delta y'_i\big\}.
\]
The self-generation map identifies the set of discounted payoffs that are supportable when continuation payoffs must lie in the set $Y$, so for any blocking coalition $C$ contemplating payoff $z_C\in D(C)$ for its members, there is an alternative continuation payoff $y'$ that deters $C$ from doing so.

\begin{lemma} \label{Lemma-Self-Generation}
    If $Y\subseteq \payoffs^{\circ}$ and $Y \subseteq \Phi_\delta(Y)$, then every $y\in Y$ can be supported by a stationary PCE.
\end{lemma}
\begin{proof}
Consider any $y \in Y \subseteq \Phi_\delta(Y)$. We will construct a stationary PCE $\sigma: \bigcup_{\tau=0}^{\infty} \histories^{\tau} \rightarrow A$ such that $U(\emptyset|\sigma)= y$. Since $y\in \payoffs^{\circ}$, there exists $\tilde{a}\in \alternatives$ such that $v(\tilde{a}) = y$. Define $\sigma(\emptyset)= \tilde{a}$. We will extend $\sigma$'s domain to $\bigcup_{\tau=0}^{\infty}\histories^{\tau}$ while making sure that for all $\tau \geq 0, h^{\tau} \in H^{\tau}$, $\sigma$ satisfies (i) stationarity: $\sigma(h^\tau,  \sigma(h^\tau), \emptyset) = \sigma\left(h^\tau\right)$ and $v(\sigma(h^{\tau}))\in Y$, and (ii) no profitable block: for all $ C\in \coalitions$  and $a'\in E_C(\sigma(h^{\tau}))$, there exists $i\in C$ such that $v_i( \sigma(h^{\tau})) \ge (1-\delta)v_i(a' ) + \delta v_i(\sigma(h^{\tau},a', \{C\}))$.

Since $y \in \Phi_\delta(Y)$, we know that for all $C\in \coalitions$, and $a\in E_C(\tilde{a})$, there exists $y'[a,C] \in Y$ such that $y_i \ge (1-\delta)v_i(a) + \delta y'_i[a,C]$ for some $i\in C$. Furthermore since $y'[a,C] \in Y\subseteq \payoffs^{\circ}$, this implies the existence of $a'[a,C]\in \alternatives$ such that $y'[a,C]=v(a'[a,C])$. We extend $\sigma$ to the domain $\{\emptyset\} \cup \histories^1$ as follows: $\sigma(a,B)=\tilde{a}$ if $B$ is either an empty set or comprises more than one coalitions; and $\sigma(a,\{C\})=a'[a,C]$ for all $C\in \coalitions$ and $a\in E_C(\tilde{a})$. Clearly, this satisfies properties (i) and (ii) for $\tau=0$.

Now we complete the definition of the plan $\sigma$ through induction on $t$. Fix $t>1$, and assume we've defined the function $\sigma: \bigcup_{\tau=0}^{t-1} \histories^{\tau} \rightarrow A$ satisfying properties (i) and (ii) for $\tau=0, \ldots, t-2$ and all $h^\tau \in \histories^{\tau}$. Consider any $h^{t-1} \in \histories^{t-1}$. Since $v(\sigma(h^{t-1})) \in Y\subseteq \Phi_\delta(Y)$, we know that for all $C\in \coalitions$, and $a\in E_C( \sigma(h^{t-1}) )$, there exists $y^{h^{t-1}}[a,C] \in Y$ such that $v_i(\sigma(h^{t-1})) \ge (1-\delta)v_i(a) + \delta y^{h^{t-1}}_i[a,C]$ for some $i\in C$. In addition, since $y^{h^{t-1}}[a,C] \in Y\subseteq \payoffs^{\circ}$, this implies the existence of $a^{h^{t-1}}[a,C]$ such that $y^{h^{t-1}}[a,C]=v(a^{h^{t-1}}[a,C])$. Extend $\sigma$ to the domain $\histories^t$ by defining $\sigma(h^{t-1}, a, B )$ as follows: $\sigma(h^{t-1}, a,B)=\sigma(h^{t-1})$ if $B$ is either an empty set or comprises more than one coalitions; and $\sigma(h^{t-1}, a,\{C\})=a^{h^{t-1}}[a,C]$ for all $h^{t-1}\in \histories^{t-1}$, $C\in \coalitions$ and $a\in E_C(\sigma(h^{t-1}))$. Note that, by construction, the function satisfies properties (i) and (ii) for $\tau=0, \ldots, t-1$ and $h^\tau \in \histories^\tau$. This completes the induction step. 

By property (i), $\sigma$ is stationary, which implies that at any history $h^t$ it delivers the discounted payoff $U(h^t|\sigma) = v(\sigma(h^t))$. In particular, $U(\emptyset|\sigma)=y$, as required. Property (ii) then implies that $\sigma$ is a PCE. 
\end{proof}

\paragraph{Proof of \cref{Theorem-Stationary}.}
We show that any payoff that can be supported by a PCE can be supported by a stationary PCE (given that the converse holds by definition). Take a PCE $\sigma$, and let $\mathcal{U}(\sigma) \equiv \{U(h|\sigma):h\in \histories \}$ denote the set of continuation values associated with $\sigma$. Since $\payoffs^{\circ}$ is convex, it follows that $\mathcal{U}(\sigma)\subseteq \payoffs^{\circ}$. Given \cref{Lemma-Self-Generation}, it suffices to show that $\mathcal{U}(\sigma)\subseteq \Phi_{\delta}(\mathcal{U}(\sigma))$.

Consider any $y\in \mathcal{U}(\sigma)$ and let $h^t$ be some $t$-history such that $U(h^t|\sigma)=y$. Since $\sigma$ is a PCE, we know that for any $C\in \coalitions$ and any $a'\in E_C(\sigma(h^t))$, there exists $y'\in \mathcal{U}(\sigma)$ and $i\in C$ such that $(1-\delta)v_i(a') + \delta y'_i \le y_i$. Since the stage game exhibits default-independent power, this is equivalent to the statement that for all $ C\in \coalitions $ and $ z_C\in D(C)$, there exists $ y'\in \mathcal{U}(\sigma) $ and $ i\in C $ such that $ y_i \ge (1-\delta)z_i + \delta y'_i$, which implies $\mathcal{U}(\sigma) \subseteq \Phi_{\delta}(\mathcal{U}(\sigma))$. The claim then follows from \cref{Lemma-Self-Generation}.\qed

\subsection{Proof of \cref{Theorem-Symmetric} on p. \pageref{Theorem-Symmetric}}

\paragraph{A Preliminary Result.} Let $\feasible^S$ denote the set of discounted payoff profiles from strongly symmetric PCEs. The following lemma is useful for proving \cref{Theorem-Symmetric}.

\begin{lemma}\label{Lemma-Symmetric}
If $\feasible^S$ is nonempty, then $\feasible^S$ is the singleton set $\{\hat{v}\}$.    
\end{lemma}
\begin{proof}
Suppose $\feasible^S$ is nonempty. Since players accrue identical payoffs from symmetric action profiles, we have $v_i=v_j$ for all players $i,j$ and $v\in \feasible^S$. Let $\hat{x}\equiv \max_{a\in \alternatives^S}v_1(a)$ be the highest feasible symmetric payoff, so $\hat{v} = (\hat{x},\ldots,\hat{x})$. Let $\underline{x} \equiv \inf\{x: (x,\ldots,x)\in \feasible^S \}$ denote the lowest symmetric PCE payoff. 

Consider a sequence $\big( (x^k,\ldots,x^k)\big)_{k=1}^\infty \subseteq \feasible^S$ that converges to $(\underline{x},\ldots,\underline{x})$ and let $\sigma^k$ be the PCE that supports payoff profile $({x}^k,\ldots,{x}^k)$. As a PCE, $\sigma^k$ cannot be profitably blocked by the grand coalition $N$ choosing $\hat{a}\in \argmax_{a\in A^S} v_1(a)$, which would generate the stage-game payoff profile $(\hat{x},\ldots,\hat{x})$. So for each $k$ we have $x^k \ge (1-\delta)\hat{x} +\delta \underline{x}$. Since $x^k\rightarrow \underline{x}$, it follows that $\underline{x}\ge \hat{x}$. However, by definitions $\underline{x}  \le   \hat{x}$, so $\feasible^S = \{\hat{v}\}$.
\end{proof}

\paragraph{Proof of \cref{Theorem-Symmetric}.} 
For the ``only if'' direction: Suppose there exists a strongly symmetric PCE $\sigma$. By \cref{Lemma-Symmetric}, $\sigma$'s continuation values satisfy $U(h|\sigma) = \hat{v}$ for all $h\in\histories$. Since $\hat{v}$ is the maximal feasible payoff from symmetric action profiles, and $\sigma$ must prescribe symmetric action profiles after every history, it follows that the default action profile $\sigma(h)$ must satisfy $v(\sigma(h)) = \hat{v}$ for all $h\in \histories$.


Take an arbitrary $\hat{h}\in \histories$ and let $\hat{a}\equiv \sigma(\hat{h})$ be the default action profile. As argued above $\hat{a}$ is symmetric and $v(\hat{a}) = \hat{v}$, so it remains to show that $\hat{a}$ is a core alternative. Suppose otherwise, namely there exists coalition $C$ and $a'\in E_C(\hat{a})$ such that $v_i(a')>v_i(\hat{a})$ for all $i\in C$. Then,
\begin{align*}
U_i(\hat{h}|\sigma) & = (1-\delta) v_i(\hat{a}) + \delta U_i(\hat{h},\hat{a},\emptyset ) = (1-\delta) v_i(\hat{a}) + \delta \hat{v}_i\\
& < (1-\delta)v_i(a') + \delta \hat{v}_i = (1-\delta)v_i(a') + \delta U_i(\hat{h},a',\{C\}|\sigma)
\end{align*}
for all $i\in C$, which contradicts $\sigma$ being a PCE. 

For the ``if'' direction: suppose $\hat{v} = v(\hat{a})$ for some symmetric core alternative $\hat{a}$. A plan that specifies $\hat{a}$ as default at all histories is a strongly symmetric PCE that supports the discounted payoff $\hat{v}$. Uniqueness follows from \cref{Lemma-Symmetric}. \qed

\subsection{Proof of \cref{Theorem-Secret} on p. \pageref{Theorem-Secret}}

\paragraph{Preliminary Results.} To prove our claim, we first introduce the transferable-utility analogue of the concept of a blocking plan, as defined in \cref{Section-Proof-NTU}.

A \textit{(transferable-utility) blocking plan} by coalition ${C}$ from a plan $\sigma$ is a pair $(\alpha, \beta)$, where $\alpha:{\historiesTU} \rightarrow \alternatives$ and $\beta:\historiesTU\rightarrow \mathcal{T}$ satisfy $\alpha(h) \in E_{C} ( a(h|\sigma) )$ and $\beta_{-C}(h) = \hchi^C ({T}(h|\sigma))$ for every history ${h}\in {\historiesTU}$. After each history, the blocking plan $(\alpha,\beta)$ generates a path 
\[
\Big( 
\alpha(h), \beta(h),  
\; \alpha\big(h, \alpha(h), \{C\}, \beta(h) \big), \beta \big(h, \alpha(h), \{C\}, \beta(h)\big), \ldots  \Big)
\] 
that is distinct from the one generated by $\sigma$. We will use $U_i(h|\alpha, \beta)$ to denote player $i$'s discounted payoff from that path. The blocking plan $(\alpha, \beta)$ is {profitable} if there exists a history ${h}$ such that $U_i({h} |\alpha,  \beta) > U_i({h}|\sigma)$ for all $i\in C$.

\begin{lemma}\label{Lemma:One-Shot-Deviation-Secret}
If a plan $\sigma$ is a public PCE, then no coalition $C\in \seccoal\cup \sincoal$ has a profitable blocking plan.
\end{lemma}

\begin{proof}
Consider a {public} plan $\sigma$ from which coalition $C\in \seccoal\cup \sincoal$ has a profitable blocking plan $(\alpha,\beta)$. In particular, there exists a history $h \in \histories$ such that $U_{i}(h|\alpha,\beta) > U_i(h|\sigma)$ for every $i\in C$. We will show that this implies that coalition $C$ has a profitable block from the plan $\sigma$, so $\sigma$ is not a PCE.

By \cref{Assumption-Bounded}, the plan $\sigma$ has bounded continuation values. Moreover, as proven in \cref{lemma:coalition bounded wlog} in the Supplementary Appendix, it is without loss to assume that the blocking plan $(\alpha,\beta)$ also has bounded continuation values. Treating coalition $C$ a fictitious player whose payoff is the sum of those of its members, we can see that $C$ faces a decision tree with bounded values. Applying the standard one-shot deviation principle to the fictitious player $C$ yields the existence of $\hat{h}\in \historiesTU$ such that 
such that 
\begin{align*}
(1-\delta) \sum_{i\in C} u_i \Big ( \alpha(\hat{h}), \; \beta(\hat{h}) \Big ) + \delta \sum_{i\in C} U_i \Big ( \hat{h}, \; \alpha( \hat{h} ), \; \{C\}, \; \beta(\hat{h}) \Big| \sigma \Big ) > \sum_{i\in C} U_i(\hat{h}|\sigma).
\end{align*}
If $C \in \sincoal$, we have already shown that $\sigma$ is not a PCE, which completes the proof. If instead $C \in \seccoal$, showing that $C$ can profitably block $\sigma$ at $\hat{h}$ amounts to showing that the total payoff can be divided so that every member of $C$ is made better off.

Let $\transfers^*$ be the transfers matrix such that for all $(j,k) \notin C\times C$, $\transfers^*_{jk} = \beta_{jk}(\hat{h})$; but for $(j,k) \in C\times C$,  $\transfers^*_{jk}$ satisfies for every $i\in C$,
\begin{align}\label{inequality: transfers rearrange}
(1-\delta) u_i \Big (\alpha(\hat{h} ), \; \transfers^* \Big ) + \delta U_i \Big (\hat{h}, \; \alpha( \hat{h} ), \; \{C\}, \;\beta(\hat{h}) \Big| \sigma \Big ) > U_i(\hat{h}|\sigma).
\end{align} 
Consider the two histories $h^1 \equiv \big (\hat{h}, \; \alpha( \hat{h}  ), \; \{C\}, \; \beta(\hat{h} ) \big )$  and  $h^2 \equiv \big ( \hat{h}, \; \alpha( \hat{h}  ), \; \{C\}, \; T^* \big )$.

By the construction of $\transfers^*$ and the fact that $C\in \seccoal\cup \sincoal$, $h^1$ and $h^2$ share the same public component $h^1_p=h^2_p$. Since the plan $\sigma$ is public, it follows that for all $i\in C$, $U_i \big (\hat{h}, \, \alpha( \hat{h} ), \, \{C\}, \, \beta(\hat{h} ) \big| \sigma \big ) = U_i \big (\hat{h}, \, \alpha( \hat{h} ), \, \{C\}, \, T^* \big| \sigma \big )$. Inequality (\ref{inequality: transfers rearrange}) can therefore be re-written as $(1-\delta) u_i \big (\alpha(\hat{h}), \, \transfers^* \big ) + \delta U_i \big (\hat{h}, \, \alpha( \hat{h} ), \, \{C\}, \, T^* \big| \sigma \big ) > U_i(\hat{h}|\sigma)$ for every $i\in C$, which implies that $\sigma$ is not a PCE.
\end{proof}

The next result shows that for any payoff profile in $\fcrt(\seccoal)$, we can construct  ``${(\seccoal\cup \sincoal)}$-specific punishments'' for all coalitions in ${\seccoal\cup \sincoal}$.

\begin{lemma}\label{Lemma:Coalition-specific Punishment}
	For any $u^*\in \fcrt(\seccoal)$, there exist ${(\seccoal\cup \sincoal)}$-specific punishments $\{u^C: C \in {\seccoal\cup \sincoal}  \} \subseteq \fcrt(\seccoal)$ such that $\sum_{i \in C} u^C_i < \sum_{i\in C} u_i^* $ for all  $C\in \seccoal\cup \sincoal$, and $\sum_{i \in C} u_i^C < \sum_{i \in C} u_i^{C'}$ for all  $C,C'\in \seccoal\cup \sincoal, C' \ne C$.
\end{lemma}

\begin{proof}
	For any coalition $C \in \seccoal\cup \sincoal$, consider the vector $u^C$ defined by
	\[
	u^C_i = 
	\begin{cases}
	u_i^* -  \frac{\epsilon}{\abs{C}} & i\in C \\
	u_i^* +\frac{\epsilon}{ \abs{\players \backslash C}}  & i\notin C
	\end{cases}
	\]
Compared to the payoff vector $u^*$, in $u^C$ every player in $C$ is taxed equally, with a total summing up to $\epsilon$; by contrast, players outside of $C$ are paid equally, with a total also summing up to $\epsilon$. The $\epsilon$ may be set sufficiently small to ensure all $u^C$'s are in $\fcrt(\seccoal)$. 

We show that these vectors satisfy the required inequalities. By construction, $\sum_{i\in C} u_i^C = \sum_{i\in C} u^*_i - \epsilon < \sum_{i\in C} u^*_i $ for all $C \in \seccoal\cup \sincoal$, which verifies the first set of inequalities in the lemma. 

Now consider two coalitions $C, C' \in \seccoal\cup \sincoal$ with $C \ne C'$. Coalition $C$ can be partitioned as $C = (C \backslash C') \cup (C \cap C')$. Compared to $u^*$, $u^{C'}$ gives everyone outside $C'$ an extra $\frac{\epsilon}{\abs{N\backslash C'}}$, while lowering the payoff of everyone inside $C'$ by $\frac{\epsilon}{\abs{C'}}$, so
	\begin{equation*}
	\sum_{i\in C} u^{C'}_i = \sum_{i\in C\backslash C'} u^{C'}_i + \sum_{i\in C \cap C'} u^{C'}_i	 =\Big [ \sum_{i\in C\backslash C'} u^*_i + \frac{\abs{ C\backslash C' }}{\abs{ \players \backslash C'}} \epsilon  \Big] + \Big [ \sum_{i\in C \cap C'} u^*_i - \frac{\abs{C\cap C' }}{\abs{C'}} \epsilon \Big ]. 
	\end{equation*}
    Combining terms above yields $\sum_{i\in C} u^{C'}_i= \sum_{i\in C} u^*_i - \big[ \frac{\abs{C\cap C' }}{\abs{C'}} - \frac{\abs{ C\backslash C' }}{\abs{ \players \backslash C'}} \big] \epsilon$. Note that since $C \ne C'$, either $C \backslash C' \ne \emptyset$ or $C \cap C' \ne C'$ (or both) must be true; in other words, either $\frac{\abs{ C\backslash C' }}{\abs{ \players \backslash C'}} > 0$ or $\frac{\abs{C\cap C' }}{\abs{C'}} < 1$. In either case, $\sum_{i\in C} u^{C'}_i> \sum_{i\in C} u^*_i - \epsilon = \sum_{i\in C} u^C_i$, which gives us the second set of inequalities in the lemma.
\end{proof}

\paragraph{Proof of \cref{Theorem-Secret}.} The result comprises two parts. 

{\noindent \underline{Part 1}: \emph{For all $\delta \ge 0$, public PCEs give each player $i$ at least $\im_i$ and each coalition $C\in {\seccoal}$ a total payoff of at least $\cf_C$.}}\smallskip

{\noindent \underline{Part 2}: \emph{For every $u\in \fcrt(\seccoal)$, there $\underline{\delta}<1$ such that for all $\delta\in(\underline{\delta},1)$, there exists a public PCE supporting $u$.}}\medskip

We prove Part 1 here and Part 2 in the Supplementary Appendix. In fact, we establish a statement stronger than Part 1: every public PCE $\sigma$ guarantees that at every history $h\in \historiesTU$, $U_i (h|\sigma ) \ge \im_i$ for every $i\in \players$ and $\sum_{i \in C} U_i (h|\sigma ) \ge \cf_C$  for every $C\in \seccoal$. The fact that $U_i (h|\sigma ) \ge \im_i$ for every $i\in \players$ follows from similar arguments as in the proof of \cref{Theorem-NTU}. Towards a contradiction, suppose $\sigma$ is a public plan  such that there exists a coalition $C\in {\seccoal}$ and history $\hat{h}$ where $\sum_{i\in  C} U_i( \hat{h} |\sigma ) < \cf_C$. We prove that $\sigma$ cannot be a PCE. 

To this end, we construct a profitable blocking plan from $\sigma$ for coalition $C$.  At every history $h\in \historiesTU$, let $a(h|\sigma)$  denote the default and $\alpha(h) \in \argmax_{a\in E_C(a(h|\sigma)) } \sum_{i\in C} v_i(a)$ be an alternative in coalition $C$'s ``best response.'' By the definition of $\cf_{C}$, it follows that $\sum_{i\in C} v_i( \alpha(h) )  \ge  \cf_{C} > \sum_{i\in C } U_i (\hat{h}|\sigma )$, so we can find transfers among players in ${C}$ such that when combined with $\alpha(h)$, these transfers give each $i \in {C}$ higher payoff than $U_i(\hat{h}|\sigma )$. Formally, at every history $h\in \historiesTU$, there exist transfers $\tilde{\transfers}_{C}(h) \equiv [\tilde{\transfers}_{ij}(h) ]_{i\in {C}, j\in \players }$ such that $\tilde{\transfers}_{ij}(h) =0 $ for all $j \in \players\backslash {C}$, and $v_i ( \alpha(h) ) +\sum_{j\in {C}} \tilde{\transfers}_{ji}(h) - \sum_{j\in {C}} \tilde{\transfers}_{ij} (h) > U_i (\hat{h}|\sigma )$ for all $i\in C$. As a result, for each player $i \in {C}$, the experienced payoff from the stage-game outcome satisfies 
\begin{equation*}
u_i \Big(\alpha(h),\; \tilde{\transfers}_{C}(h), \hchi^{C} \big(\, T(h|\sigma ) \,\big) \Big)\ge v_i ( \alpha(h) ) +\sum_{j\in {C}} \tilde{\transfers}_{ji}(h) - \sum_{j\in {C}} \tilde{\transfers}_{ij} (h)> U_i (\hat{h}|\sigma ),
\end{equation*}
where the weak inequality follows because $\hchi^C_{ji}(\, T(h|\sigma)\, ) \ge 0$  and $ \tilde{\transfers}_{ij}(h) = 0 $ for all $j\in \players \backslash {C}$. Observe that the inequality above can hold at every history, including $\hat{h}$ and those that follow. These steps prove that the blocking plan $(\alpha,\beta)$ by coalition $C$, where $\beta(h) \equiv \big[\tilde{\transfers}_{C}(h) , \hchi^C(\, T(h| \sigma ))\big]$
for every history $h\in \historiesTU$, is profitable: $U_i(\hat{h}|\alpha,\beta ) > U_i(\hat{h}|\sigma )$ for every  $i\in C$. \cref{Lemma:One-Shot-Deviation-Secret} then implies that $\sigma$ is not a PCE.\medskip

\newpage

\section*{Supplementary Appendix}
\setcounter{secnumdepth}{2}

The Supplementary Appendix is organized as follows:
\begin{itemize}[noitemsep]
    \item \Cref{Section-SupplemenntaryPreliminary} completes the proof of \Cref{Theorem-Secret}.
    \item \Cref{Section-KCProofs} proves results for labor market matching (\Cref{Proposition-KC-Public,Proposition-KC-Secret,Proposition-Wages,Proposition-Comparative-Wages}).
    \item \Cref{Section-PCPE} formalizes the concept of Perfect Coalition Proof Equilibrium and establishes the counterpart of \Cref{Theorem-NTU}.
    \item \Cref{Section-Distributive} describes our application to bargaining games with veto players.
\end{itemize}

\section{Completing the Proof of \Cref{Theorem-Secret}}\label{Section-SupplemenntaryPreliminary}
We first prove a few preliminary results used in the argument. The first result shows that when checking profitable blocking plans, we can WLOG focus on those with bounded total continuation values.

\begin{lemma} \label{lemma:coalition bounded wlog}
Let $\sigma$ be a PCE. Suppose coalition $C$ has blocking plan $(\alpha,\beta)$ such that $\sum_{i\in C} U_i(\overline{h}| \alpha,\beta ) > \sum_{i\in C} U_i(\overline{h}| \sigma )$ for some $\overline{h}\in \historiesTU$, then $C$ has blocking plan $(\alpha',\beta')$ such that $\sum_{i\in C} U_i(\overline{h}| \alpha',\beta' ) > \sum_{i\in C} U_i(\overline{h}| \sigma )$, and $\{\sum_{i\in C} U_i(h| \alpha',\beta'): h\in  \historiesTU \}$ is bounded.
\end{lemma}

\begin{proof} We break this argument into two parts.\medskip

\noindent\underline{Part 1:} We show that the set $\{\sum_{i \in C} U_i(h| \alpha,\beta ): h\in  \historiesTU \}$ is bounded from above. To this end, it suffices to show that the set of stage-game payoffs from the blocking plan, $\{ \sum_{i\in C} u_i ( \alpha(h),\beta(h)): h\in \historiesTU \}$ is bounded from above.

Consider an arbitrary coalition $C\in\coalitions$ and an arbitrary history $h\in \historiesTU$. Let $\tilde{a} = a ( h |\sigma )$ denote the default alternative specified by $\sigma$ and $\tilde{\transfers} = \transfers( h|\sigma)$ denote the default transfers. By the definition of a blocking plan, $\alpha(h)\in E_{C}(\tilde{a})$ and $\beta(h) = ( {T}'_C, \hchi^C(\tilde{T}))$ for some ${T}'_C\in \mathcal{T}_C$. Since the transfers ${T}'_C$ may involve nonzero transfers to players outside of $C$, we have
\begin{align} \label{Equation-average-bound}
\sum_{i\in C} u_i(\alpha(h),\beta(h) ) = & \sum_{i\in C} v_i(\alpha(h)) + \sum_{i\in C, j\notin C} \hchi^C_{ji}(\tilde{\transfers}) - \sum_{i\in C,j\notin C} {T}'_{ij} \nonumber \\
\le & \sum_{i\in C} v_i(\alpha(h)) + \sum_{i\in C, j\notin C} \hchi^C_{ji}(\tilde{\transfers}) 
\end{align}

Now suppose the coalition $C$ blocks at history $h$ and chooses alternative $\alpha(h)\in E_C(\tilde{a})$; however, instead of $\beta(h) = ( {T}'_C, \hchi^C(\tilde{T}))$, $C$ chooses transfers $(T''_C,\hchi^C(\tilde{T}) )$, where the transfers $T''_C$ are such that members of $C$ make zero payment to players outside of $C$ while splitting the total payoff within $C$ evenly. If $C$ carries out this block, each member $i\in C$ obtains a discounted utility of at least
\[
    (1-\delta)\frac{1}{\abs{C}} \Big[\sum_{i\in C} v_i(\alpha(h)) + \sum_{i\in C, j\notin C} \hchi^C_{ji}(\tilde{\transfers})\Big] + \delta\; \inf_{ h \in \historiesTU,i\in \players }   U_i(h|\sigma ),
\]
whereas adhering to $\sigma$ at $h$ yields each member at most $\sup_{h \in \historiesTU,i\in \players} U_i(h|\sigma )$. Since $\sigma$ is a PCE, $\big(\,\alpha(h), T''_C,\hchi^C(\tilde{T})\,\big)$ cannot be a profitable block for $C$, so it must be true that
\begin{equation*} 
	(1-\delta)\frac{1}{\abs{C}} \Big[\sum_{i\in C} v_i(\alpha(h)) + \sum_{i\in C, j\notin C} \hchi^C_{ji}(\tilde{\transfers})\Big] + \delta\; \inf_{ h \in \historiesTU,i\in \players }   U_i(h|\sigma )   \le \sup_{h \in \historiesTU,i\in \players} U_i(h|\sigma ).
\end{equation*}
Combining the inequality above with \eqref{Equation-average-bound} yields
\[
	(1-\delta)\frac{1}{\abs{C}} \Big[\sum_{i\in C} u_i(\alpha(h),\beta(h) )\Big] + \delta\; \inf_{ h \in \historiesTU,i\in \players }   U_i(h|\sigma )   \le \sup_{h \in \historiesTU,i\in \players} U_i(h|\sigma ).
\]
Rearranging terms, we have
\begin{align*}
    \sum_{i\in C} u_i(\alpha(h),\beta(h) ) & \le  \frac{|C|}{1-\delta} \left[\sup_{h \in \historiesTU,i\in \players} U_i(h|\sigma ) - \delta\; \inf_{ h \in \historiesTU,i\in \players }   U_i(h|\sigma ) \right]\\
    & \le \frac{|C|}{1-\delta}  \abs{\sup_{h \in \historiesTU,i\in \players} U_i(h|\sigma )} +  \frac{|C|\delta}{1-\delta}\; \abs{\inf_{ h \in \historiesTU,i\in \players }   U_i(h|\sigma ) }.
\end{align*}
Since $\{ U(h|\sigma ): h \in \historiesTU  \}$ is bounded by \cref{Assumption-Bounded}, there exists $L>0$ such that 
\[
\abs{  \sup_{h \in \historiesTU,i\in \players}  U_i(h|\sigma ) } \le L \;\;\; \text{ and }\;\;\;  \abs{ \inf_{h \in \historiesTU,i\in \players}  U_i(h|\sigma ) } \le L. 
\]
Therefore,
\[
\sum_{i\in C} u_i(\alpha(h),\beta(h) ) \le \frac{1+\delta}{1-\delta} \abs{C}L.
\]
Note that the inequality above holds for all $h\in \historiesTU$ while the right hand side does not depend on $h$, so our claim follows.
\medskip

\noindent\underline{Part 2:} It is without loss to assume that $\{\sum_{i \in C} U_i(h| \alpha,\beta ): h\in  \historiesTU \}$ is bounded from below. If not, we can construct another blocking plan $(\alpha',\beta')$ such that $\sum_{i\in C} U_i(\overline{h}| \alpha',\beta' ) > \sum_{i\in C} U_i(\overline{h}| \sigma )$ while ensuring $\{\sum_{i\in C} U_i(h| \alpha',\beta' ): h\in  \historiesTU \}$ is bounded from below: if $\sum_{i\in C} U_i(\hat{h}|\alpha,\beta )$ falls below $\min_{a\in \alternatives} \sum_{i\in C} v_i(a)$ for some history $\hat{h}\in\historiesTU$, we will ask $C$ to refuse all outgoing transfers at all histories following $\hat{h}$.
	
Formally, for a history $\hat{h} \in \historiesTU$, let $F(\hat{h} ) \equiv \{ \hat{h}h: h\in \historiesTU \}$ denote the set of histories that can follow from $\hat{h}$. Let $\underline{H}_C \equiv \{ h \in \historiesTU: \sum_{i\in C} U_i (h |\alpha,\beta) < \min_{a \in \alternatives} \sum_{i\in C} v_i(a) \}$. Let $\textbf{0}_C$ denote the vector of zero-valued transfers made from players in $C$. Set $\alpha'=\alpha$, and define
	\[
	\beta'(h) =
	\begin{cases}
	\big(\textbf{0}_C, \hchi^{C}(T(h|\sigma) ) \big) & \text{if } h \in F(\hat{h} ) \text{ for some } \hat{h}\in \underline{H}_C,  \\
	\beta(h) & \text{otherwise.}
	\end{cases}
	\]
	By construction, the blocking plan $(\alpha', \beta')$ has continuation values bounded below by $\min_{a\in \alternatives} \sum_{i\in C} v_i(a)$. In addition, compared to  $(\alpha, \beta)$, the blocking plan $(\alpha', \beta')$ gives coalition $C$ weakly higher total continuation value following any history, so $\sum_{i\in C} U_i(\overline{h}| \alpha',\beta' ) > \sum_{i\in C} U_i(\overline{h}| \sigma )$.
\end{proof}

Next we argue that there exists a finite set of payoff vectors whose convex hull contains the set $\firt$.

\begin{lemma} \label{lemma:TUPM-vertices}
For each alternative $a\in \alternatives$ let $\feasiblet(a)\equiv \{u\in \Re^n: \sum_i u_i = \sum_i v_i(a)\}$ denote the set of payoff profiles that can be generated by playing alternative $a$ and redistributing through transfers. Let $\overline{a}\in \argmax_{a\in \alternatives}\sum_{i\in \players} v_i(a)$ and $\underline{a}\in \argmin_{a\in \alternatives}\sum_{i\in \players} v_i(a)$ be two alternatives that maximize and minimize players' total generated payoffs, respectively. There exist payoff vectors $ \{ \tilde{u}^1, \ldots, \tilde{u}^M \} \subseteq  \feasiblet(\overline{a} )\cup \, \feasiblet(\underline{a} ) $, such that $\firt \subseteq \conv(\tilde{u}^1, \ldots, \tilde{u}^M) $.
\end{lemma}

\begin{proof}
By definition,
\[
	\firt \subseteq  \overline{\mathcal{U}}_{IR} \equiv \left \{u\in \Re^n: \sum_{i \in \players}  v_i(\underline{a} ) \le \sum_{i \in \players }u_i \le \sum_{i \in \players } v_i(\overline{a} ) \text{ and } u_i \ge \im_i \forall i\in \players \right \}.
\]
Since $\overline{\mathcal{U}}_{IR}$ is a bounded polyhedron, it is also a polytope. Let $x^1,\ldots, x^K $ be its vertices. Any point inside $\firt$ can then be expressed as convex combinations of these vertices.  Since $x^k \in \conv (\,\feasiblet(\overline{a} ) \cup \,\feasiblet(\underline{a} )) $ for all $1\le k \le K$, for each $k$, there exist $\{ \tilde{u}^{k,1},  \ldots,  \tilde{u}^{k,m_k} \} \subseteq \feasiblet(\overline{a}) \cup \feasiblet(\underline{a} )$ such that $x^k \in \conv (\tilde{u}^{k,1}, \ldots,  \tilde{u}^{k,m_k})$. As a result $\firt \subseteq \conv  (\cup_{1\le k\le K}  \left\{ \tilde{u}^{k,1},  \ldots,  \tilde{u}^{k,m_k} \right\}) $.
\end{proof}

\paragraph{Proof of Part 2 of \Cref{Theorem-Secret}.} Here, we establish that {for every $u\in \fcrt(\seccoal)$, there $\underline{\delta}<1$ such that for all $\delta\in(\underline{\delta},1)$, there exists a public PCE supporting $u$.}

To economize on notation, we define $\cf_{\{i\}}\equiv \im_i$ for all  $i\in \players$, so $\fcrt(\seccoal)$ can be written as $\fcrt(\seccoal) = \left\{ u \in \feasiblet: \sum_{i \in C} u_i > \cf_C \text{ for every } C \in \seccoal \cup \sincoal\right\}$. For every $C\in {\seccoal\cup \sincoal} $, let $\underline{a}_C \in \argmin_{a\in  \alternatives} \max_{a'\in E_{C}(a)} \sum_{i\in C} v_i(a')$ be an alternative that can be used to minmax $C$. Observe that, by Berge's maximum theorem, $\underline{a}_C$ is well-defined because $A$ is compact, $v$ is continuous, and $E_{C}(\cdot)$ is continuous and compact-valued. By the definition of $\cf_C$, $\sum_{i\in C}v_i(a') \le \cf_C$ for all $a'\in E_C(\underline{a}_C)$. Given the reflexivity of effectivity correspondences, $\underline{a}_C\in E_C(\underline{a}_C)$, and therefore, $\sum_{i\in C} v_i(\underline{a}_C)~\le~\cf_C$.\footnote{This is the only step of the argument in which we invoke reflexivity. As this step shows, our proof does not require a global form of reflexivity but only that it holds for minmaxing alternatives. In our application to labor market matching, we establish that this weaker form of reflexivity holds, which allows us to apply the argument here. \label{footnote: reflexivity}}

Fix any payoff vector $u^* \in \fcrt (\seccoal)$, and let $\{u^C:C\in \seccoal\cup \sincoal\}$ be the ${(\seccoal\cup \sincoal)}$-specific punishments from \cref{Lemma:Coalition-specific Punishment}. Given these punishments, let $\kappa\in (0,1)$ be such that for every $\tilde\kappa \in [\kappa,1]$, the following is true for all $C\in {\seccoal\cup \sincoal}$ and $C'\in {\seccoal\cup \sincoal}\backslash\{C\}$:  
\begin{align}
(1-\tilde\kappa)\sum_{i\in C} v_i(\underline{a}_C)+\tilde\kappa \sum_{i\in C} u_i^C& > \cf_C \label{Equation-Kapppa-Some}\\
(1-\tilde\kappa)\sum_{i\in C'} v_{i}(\underline{a}_C)+\tilde\kappa \sum_{i\in C'} u_i^C &> (1-\tilde\kappa) \sum_{i\in C'} v_i(\underline{a}_{C'} ) +\tilde\kappa \sum_{i\in C'} u_i^{C'}.
\label{Equation-Kappa-Some}
\end{align} 
By an argument identical to that in \Cref{Theorem-NTU}, there exists $\kappa \in (0,1)$ such that the inequalities above hold for all $\tilde\kappa\in [\kappa,1]$, $C\in \seccoal\cup \sincoal$ and $C'\in \seccoal\cup \sincoal\backslash\{C\}$. Let $L(\delta)\equiv \left\lceil{\frac{\log \kappa}{\log \delta}}\right\rceil$. As before, we use the property that $\lim_{\delta\rightarrow 1} \delta^{L(\delta)}=\kappa$.

For each alternative $a\in \alternatives$ let $\feasiblet(a)\equiv \{u\in \Re^n: \sum_i u_i = \sum_i v_i(a)\}$ denote the set of payoff profiles that can be generated by playing alternative $a$ and redistributing through transfers. Let $\overline{a}\in \argmax_{a\in \alternatives}\sum_{i\in \players} v_i(a)$ and $\underline{a}\in \argmin_{a\in \alternatives}\sum_{i\in \players} v_i(a)$ be alternatives that maximize and minimize total payoffs, respectively. Since $\fcrt(\seccoal ) \subseteq \firt$, by \cref{lemma:TUPM-vertices} in the Supp. Appendix, there exist payoff vectors $\{ \tilde{u}^1, \ldots, \tilde{u}^M \} \subseteq \mathcal{U}(\overline{a}) \cup \mathcal{U}(\underline{a})$ such that $\fcrt(\seccoal) \subseteq \conv( \tilde{u}^1, \ldots, \tilde{u}^M)$, where each $\tilde{u}^m = u(\tilde{a}^m, \tilde{\transfers}^m )$ for some $\tilde{a}^m\in \{\overline{a}, \underline{a} \}$ and $\tilde{\transfers}^m$. Let $\tilde{\mathcal{T}} \equiv \{ \tilde{T}^m \}_{m=1}^M$ be the set comprising these transfer matrices.

By \cref{lemma:NTU-folk-sequence-approx}, for any $\epsilon >0$, there exists $\underline\delta\in (0,1)$ such that for all $\delta \in (\underline\delta,1) $, there exist sequences $\big \{ (a^{d,\tau}, \transfers^{d,\tau} )_{\tau =0}^\infty : d \in {\seccoal\cup \sincoal}\cup\{*\} \big \} \subseteq \big( \{\overline{a},\underline{a} \}\times \tilde{\mathcal{T}} \big)^{\infty}$ such that for each $d$ and $t$, $(1-\delta) \sum_{\tau =0}^\infty \delta^\tau u(a^{d, \tau}, \transfers^{d,\tau} ) = u^d$ and $\norm{u^d - (1-\delta )\sum_{\tau = t}^\infty \delta^{\tau-t} u(a^{d,\tau}, \transfers^{d,\tau} ) } < \epsilon$. We fix an $\epsilon$ such that
\begin{equation*}
    \epsilon < (1-\kappa) \min \Big \{\allowbreak   \min_{d \in {\seccoal\cup \sincoal}, d' \in {\seccoal\cup \sincoal}\cup\{*\}, d'\neq d} \Big(\sum_{i\in d} u_i^{d'} - \sum_{i\in d} u_i^d \big ), \min_{d\in {\seccoal\cup \sincoal} } \sum_{i\in d} u^d_i - \cf_d \Big \},
\end{equation*}
and given that $\epsilon$, consider $\delta$ exceeding the appropriate $\underline\delta$.

We now describe the public plan that we use to support $u^*$. Let $\mathbf{0}$ denote the transfer matrix where all players make no transfers. For any non-singleton coalition \(C \in \coalitions \setminus \sincoal\), let $\sincoal(C) \equiv \{\{i\} : i \in C\}$ denote the set of singleton coalitions formed by members of $C$. For any collection of blocking coalitions $B\in \mathcal{B}$, define $\hat{C}(B) = [ B\cap (\seccoal\cup \sincoal ) ] \cup [\cup_{C\in B\backslash(\seccoal\cup\sincoal)}\sincoal(C) ]$. Note that $B\cap(\seccoal\cup\sincoal)$ is the set of secret or individual coalitions in $B$, while $\cup_{C\in B\backslash(\seccoal \cup\sincoal)}\sincoal(C)$ consists of singleton coalitions converted from members of $\cup_{C\in B\backslash(\seccoal \cup\sincoal)}C$, so $\hat{C}(B)$ is the collection of ``players'' in $B$ if coalitions in $\seccoal\cup\sincoal$ are treated as fictitious players. Consider the plan represented by the automaton $(W, w(*,0),f,\gamma )$, where 
\begin{itemize}
	\item $W\equiv \big \{ w(d, \tau) | d\in { \seccoal\cup \sincoal} \cup\{*\}, \tau \ge 0 \big \} \cup\{ \underline{w}(S, \tau )\,|\, S \in {\seccoal\cup \sincoal} ,0\le \tau <L(\delta)\}$ is the set of possible states and $w(*,0)$ is the initial state;	
	\item $f:W \rightarrow \outcomesTU $ is the output function, where $f(w(d, \tau )) = (a^{d,\tau },\emptyset, \transfers^{d,\tau} ) $ and $f(\underline{w}(S, \tau ))= ( \underline{a}_{S}, \emptyset, \mathbf{0} )$;
	\item $\gamma: W \times \outcomesTU \rightarrow W$ is the transition function. For each $S\in \seccoal$ let $u_S(a,T)=\sum_{i\in S}u_i(a,T)$ denote the total utility accruing to $S$. For states of the form $\{\underline{w}(S,\tau) | 0 \le \tau < L(\delta)-1, S \in \sincoal\cup \seccoal\}$, the transition is
	\[
	\gamma (\underline{w}(S, \tau),(a, B, \transfers) )=
	\begin{cases}
	\underline{w}(S^*, 0) & \text{where } S^* \in \argmin_{ S'\in \hat{C}(B) \backslash\{S\} } u_{S'}(a,T),\\	
    & \;\text{ if } B \ne \emptyset \text{ but either }\{S \notin \hat{C}(B) \} \\
    &\; \text{ or } \{u_S(a, \transfers )  > \cf_S \} \text{ is true.} \\
    \underline{w}(S,0) &   \text{if } B \ne \emptyset \text{ and both }\{S \in \hat{C}(B) \} \\
    & \;\text{ and } \{u_S(a, \transfers )  \le \cf_S \}  \text{ are true.}\\
	\underline{w}(S, \tau + 1) & \text{if $B = \emptyset$.}
	\end{cases}
	\]
    For states of the form $\{\underline{w}(S,L(\delta)-1 )| S \in \sincoal\cup \seccoal\}$, the transition is
	\[
	\gamma (\underline{w}(S, L(\delta) -1),(a, B, \transfers) )=
	\begin{cases}
	\underline{w}(S^*, 0) & \text{where } S^* \in \argmin_{S'\in \hat{C}(B) \backslash\{S\} } u_{S'}(a,T),\\	
    & \;\text{ if } B \ne \emptyset \text{ but either }\{S \notin \hat{C}(B) \} \\
    & \;\text{ or } \{u_S(a, \transfers )  > \cf_S \} \text{ is true.} \\
    \underline{w}(S,0) &   \text{if } B \ne \emptyset \text{ and both }\{S \in \hat{C}(B) \} \\
    & \;\text{ and } \{u_S(a, \transfers )  \le \cf_S \}  \text{ are true.}\\
	{w}(S, 0) & \text{if }B = \emptyset.	
	\end{cases}
	\]

    For states of the form $w(d,\tau)$, the transition is
    \[
	\gamma \big( w(d,\tau), (a, B, \transfers) \big)=
    \begin{cases}
	\underline{w}(S^*,0) & \text{if } B\ne \emptyset, \\
    w(d, \tau+1) & \text{if $B = \emptyset$,} 
	\end{cases}
	\] 
    where $S^* \in \argmin_{S' \in \hat{C}(B)} u_{S'}(a,T)$.
\end{itemize}

The plan represented by this automaton yields payoff profile $u^*$. The plan is also public since the transition relies only on $B$ and $\{u_S(a,T): S\in \hat{C}(B)\}$, both of which are public information. By construction, $\norm{u^d - V(w(d,\tau ))} < \epsilon$ and $V(\underline{w}(S, \tau)) = (1-\delta^{L(\delta) - \tau})  v(\underline{a}_S ) + \delta^{L(\delta) - \tau} V( w(S,0) )$
 for all $\tau$ in $\{0,\ldots,L(\delta)-1\}$ and $S \in  {\seccoal\cup \sincoal}$. As the arguments from here on are standard, we verify in the Supplementary Appendix that no coalition can profitably block in any state of this automaton. 

Recall that $\tilde{\mathcal{T}} = \{ \tilde{T}^m \}_{m=1}^M$ and $\big \{ (a^{d,\tau}, \transfers^{d,\tau} ): \tau\ge 0, d \in {\seccoal\cup \sincoal}\cup\{*\} \big \} \subseteq \{\overline{a},\underline{a} \}\times \tilde{\mathcal{T}}$, so all default transfers in the plan are selected from a finite set. By \cref{Assumption-Bounded-Transfers}, if  a coalition $C$ blocks a default transfers matrix $T\in \tilde{\mathcal{T}}$, there exists $\tilde{b}>0$ such that 
\begin{equation} \label{Equation-Bounded-Transfers1}
\sum_{i\notin C,j\in C}\hchi_{ij}^C({T}) \le \tilde{b} \text{ for all } T\in \tilde{\mathcal{T}} \text{ and } C\in \coalitions.    
\end{equation}
In addition, since $\alternatives$ is compact and $v(.)$ is continuous, there exists $\hat{b}$ such that 
\begin{equation} \label{Equation-Bounded-Gen1}
    \max_{C\in \coalitions} \max_{a\in \alternatives} \sum_{i\in C} v_i(a) \le \hat{b}.
\end{equation}
We verify the incentives in the automaton states below.

{\flushleft \textbf{States of the form $w(d, \tau )$:}} Suppose coalition $C$ blocks and the outcome $(\hat{a}, \{C\} , \hat{\transfers} )$ is realized. Recall that $\hat{C}(B) = [ B\cap (\seccoal\cup \sincoal ) ] \cup [\cup_{C\in B\backslash(\seccoal\cup\sincoal)}\sincoal(C) ]$, so if $C\in \seccoal\cup\sincoal$, then $\hat{C}(\{C\})=\{ C\}$ is a singleton set containing $C$ as a unitary player; however, if $C\notin \seccoal\cup\sincoal$, then $\hat{C}(\{C\}) = \sincoal(C)$, which is a set of singleton coalitions made up of players in $C$. The plan punishes $S^* \in \argmin_{S' \in \hat{C}(\{C\})} u_{S'}(a,T)$, so $S^*$ is either $C\in\seccoal$ as a unitary player or some singleton coalition $\{i\}$ where $i\in C$. In either case, the (total) stage-game payoff for $S^*$ satisfies
\begin{align*} 
	u_{S^*}(\hat{a}, \hat{\transfers})  & \le \frac{1}{|\hat{C}(\{C\})|}  \sum_{S' \in \hat{C}(\{C\})} u_{S'}(\hat{a}, \hat{\transfers})   \le   \max_{a \in \alternatives} \sum_{j\in C} v_j(a) + \sum_{j \in C} \sum_{k \notin C} \hchi^C_{kj} (T^{d,\tau}) \le  \hat{b} + \tilde{b},
\end{align*}
where the first inequality follows since the minimum among a set of numbers is less than their average; the second inequality follows since $C$'s total payoff comes from the generated payoffs plus the net transfers paid by players outside of $C$; lastly, the third inequality follows from \eqref{Equation-Bounded-Transfers1} and \eqref{Equation-Bounded-Gen1} and the fact that all $\transfers^{d,\tau}$ are drawn from $\{\tilde{\transfers}^m\}_{m=1}^M$.

Thus, we can find a uniform bound $\overline{b}_1 \equiv \hat{b} + \tilde{b}$ such that the total stage-game payoff of $S^*$ satisfies $u_{S^*}(\hat{a}, \hat{\transfers}) \le \overline{b}$ for all $C$, $\delta$ and $(d,\tau)$. Following the same steps as those in the analogous part of \Cref{Theorem-NTU}, we can show that $S^*$ obtains lower total payoff after coalition $C$ blocks. Since $S^*\subseteq C$, there exists player $i\in C$ who is not better off, so this is not a profitable block for $C$.

{\flushleft \textbf{States of the form $\underline{w}(S, \tau )$ where $S\in \sincoal \cup \seccoal$:}} Suppose coalition $C$ blocks and the outcome is $(\hat{a}, \{C\}, \widehat{\transfers} )$. Just like above, depending on whether $C\in \seccoal\cup\sincoal$, $\hat{C}(\{C\})$ is either $\{C\}$ containing $C$ as a unitary player or the set $\sincoal(C)$ of singleton coalitions  made up of $C$'s members. There are 2 cases to consider.\smallskip

\noindent  \underline{Case I: $S \in \hat{C}(\{C\})$, and $u_S(\hat{a}, \hat{\transfers} )  \le \cf_S$}. In this case, the plan punishes the current scapegoat $S$, where $S$ is either $C$ or a singleton set containing a member of $C$. Using \eqref{Equation-Kapppa-Some} for sufficiently high $\delta$ and following steps identical to the analogous argument in \Cref{Theorem-NTU}, we can show that
\begin{align*}
(1-\delta^{L(\delta)- \tau}) v_{S} (\ima_S)  + \delta^{L(\delta) - \tau} u_{S}^{S}\geq 	(1-\delta) \cf_S + \delta(1-\delta^{L(\delta)})v_{S} (\ima_S)+\delta^{L(\delta)+1} u_{S}^{S},
\end{align*}
where $v_S(.) = \sum_{i\in S}v_i(.)$ and $u^S_S  = \sum_{i\in S}u^S_i$ denote $S$'s total payoff. If $S$ contains only a member of $C$, then the inequality above shows that the blocking is not profitable for $C$; if $S$ is $C$ itself, it implies that blocking does not improve $C$'s total value, so there exist $i\in C$ who is not better off, so again the blocking is not profitable for $C$.
\smallskip

\noindent \underline{Case II: either $S \notin \hat{C}(\{C\})$ or $u_S(\hat{a}, \hat{\transfers} )  > \cf_S$}. 
In this case the plan punishes $S^* \in \argmin_{S'\in \hat{C}(\{C\}) \backslash\{S\} } u_{S'}(\hat{a},\hat{T})$ as scapegoat. 

First observe that if $C$ blocks in state $\underline{w}(S,\tau)$ and the stage-game payoff satisfies $u_S(\hat{a}, \hat{\transfers} )  > \cf_S$, then no matter if $S\in \sincoal$ or $S\in \seccoal$, it must be that $C\ne S$ and therefore $\hat{C}(\{C\}) \ne \{S\}$; otherwise the definition of $\cf_S$ would ensure $u_S(\hat{a}, \hat{\transfers} ) \le \cf_S$. As a result, under either of the conditions defining the current case (i.e. $S \notin \hat{C}(\{C\})$ or $u_S(\hat{a}, \hat{\transfers} )  > \cf_S$), $\hat{C}(\{C\}) \ne \{S\}$ must be true, so the scapegoat $S^* \in \argmin_{S'\in \hat{C}(\{C\}) \backslash\{S\} } u_{S'}(\hat{a},\hat{T})$ is well defined (i.e. the $\argmin$ is not taken over an empty set). 

We show that the (total) stage-game payoff of $S^*$  is bounded. If $S \notin \hat{C}(\{C\})$, then 
\[
S^* \in \argmin_{S'\in  \hat{C}(\{C\}) \backslash\{S\}  }   u_{S'}(\hat{a},\hat{T}) = \argmin_{ S' \in \hat{C}(\{C\}) }   u_{S'}(\hat{a},\hat{T}), \text{ and }
\] 
\begin{align} \label{Equation-j*-Bounded1-Some}
	u_{S^*}(\hat{a},\hat{\transfers} )  & \le \frac{1}{| \hat{C}(\{C\}) |} \sum_{S' \in \hat{C}(\{C\})} u_{S'}(\hat{a},\hat{\transfers} )  \le \frac{1}{|\hat{C}(\{C\})|} \max_{a\in \alternatives} \sum_{j \in C} v_j(a),
\end{align}
where the last inequality above follows from \cref{Assumption-Bounded-Transfers} and the fact that the default transfers are $0$ in states $\underline{w}(S,\tau)$. Alternatively, if $S\in \hat{C}(\{C\})$ and $u_S( \hat{a},\hat{\transfers}) > \cf_S$, then
\begin{align*}
	u_{S^*}(\hat{a},\hat{\transfers} )  & \le \frac{1}{|\hat{C}(\{C\})| - 1} \sum_{S' \in \hat{C}(\{C\})\backslash\{S\} } u_{S'}(\hat{a},\hat{\transfers} ) \nonumber \\
    & =  \frac{1}{|\hat{C}(\{C\})|-1} \Big[ \sum_{S' \in \hat{C}(\{C\})} u_{S'}( \hat{a},\hat{\transfers} )  -  u_S(\hat{a},\hat{\transfers}) \Big] \nonumber \\
    & \le \frac{1}{|\hat{C}(\{C\})| - 1} \Big[ \sum_{S' \in \hat{C}(\{C\})} u_{S'}(\hat{a},\hat{\transfers})  - \cf_S \Big],
\end{align*}
where the last inequality follows because we are considering the case $u_{S}( \hat{a},\hat{\transfers}) > \cf_S$. Since the plan specifies zero default transfers, \cref{Assumption-Bounded-Transfers} ensures
\begin{align} \label{Equation-j*-Bounded2-Some}
	u_{S^*}(\hat{a},\hat{\transfers} ) \le \frac{1}{|\hat{C}(\{C\})|-1} \Big[ \max_{a\in \alternatives} \sum_{i \in C} v_{i}(a)  - \cf_S \Big].
\end{align}
Comparing the RHS of \eqref{Equation-j*-Bounded1-Some} and \eqref{Equation-j*-Bounded2-Some} to the bounds obtained in \eqref{Equation-Bounded-Transfers1} and \eqref{Equation-Bounded-Gen1},  it follows that we can find $\overline{b}_2$ such that $u_{S^*}(\hat{a},\hat{\transfers} ) <\overline{b}_2$.
 
Finally, to show that the blocking by $C$ is not profitable, note that the (total) payoff of $S^*$ is not improved by blocking if
	\begin{align*}
		(1-\delta^{L(\delta)- \tau}) v_{S^{*}} (\ima_S)  + \delta^{L(\delta) - \tau}u_{S^{*}}^{S}\geq 	(1-\delta) \overline{b}_2 + \delta(1-\delta^{L(\delta)}) {v}_{S^{*}}(\ima_{S^*}) + \delta^{L(\delta)+1} u_{S^{*}}^{S^{*}}.
		\end{align*}
This inequality follows for sufficiently high $\delta$ from the same steps as that of the analogous part of \Cref{Theorem-NTU}. Based on the same arguments as in previous cases, the blocking is not profitable for $C$.\qed

\section{Proofs for Labor Market Matching}\label{Section-KCProofs}

\subsection{Proof of \cref{Proposition-KC-Public} on p. \pageref{Proposition-KC-Public}}\label{Section-SupplementaryProp1}
To establish this conclusion, we apply \cref{Theorem-PTU}, or more specifically, the proof of \cref{Theorem-Secret} in which $\seccoal=\emptyset$. Note that for each player $i\in \firms\cup \workers$, the individual minmax is $\im_i=0$. In this setting, to minmax a player, it suffices to use any assignment in which that player is unmatched. For specificity, we consider the assignment in which all players are unmatched, and denote it by $\phi^\circ$. All steps of the argument of \Cref{Theorem-Secret} go through without adaptation but one, namely that which invokes the reflexivity of the effectivity correspondence. Observe that, as discussed in \cref{footnote: reflexivity}, the weaker local reflexivity holds for the minmaxing alternative: for every player $i$, $E_{\{i\}}(\phi^\circ)=\{\phi^\circ\}$. Therefore, this step also goes through.  \qed

\subsection{Proof of \cref{Proposition-KC-Secret} on p. \pageref{Proposition-KC-Secret}} \label{section: proof of prop 2}\label{Section-SupplementaryProp2}

As we argued in the proof of \Cref{Proposition-KC-Public}, the minmax for an individual player $i$ is $\im_i=0$, attained using the assignment $\phi^\circ$. Like the proof of \cref{Theorem-Secret}, we use $\cf_{\{i\}}\equiv \im_i = 0$ to denote the minmax for singleton coalitions consisting of individual firms and workers, so $\cf_C=0$ for all $C\in \sincoal$. For an essential coalition $C =  \{f\}\cup W \in \esscoal$, $\cf_C =  \sum_{w\in W} v_w(\{f\}) + v_f(W)$, which is the total value generated when $C$ matches.

\paragraph{Preliminary Results.}
We establish several results that allow us to characterize $\mathcal{K}$.

\begin{lemma} \label{Lemma-K-Equivalence-to-KC}
A matching $(\phi,T)$ is a core allocation, as defined in \cref{Definition-KC-Core}, if and only if it admits no profitable blocking by any singleton firm, singleton worker, or essential coalition.
\end{lemma}
\begin{proof}
The ``only if'' direction is immediate. We prove the ``if'' direction by contradiction.
Suppose that the matching $(\phi,T)$ admits no profitable blocking by any coalition $C\in \esscoal\cup \sincoal$ but can be profitably blocked by some larger coalition $C\notin \esscoal\cup \sincoal$. It follows then that there exists an assignment $\phi' \in E_C(\phi)$ and transfers $T'_C = [T'_{ij}]_{i\in C, j\in \players}$ such that every member of $C$ is better off from the matching $(\phi',T'_C,\hchi^C(T))$, i.e. $u_i(\phi',T'_C,\hchi^C(T)) > u_i(\phi,T)$ for all $i\in C$. By the definition of the effectivity function $E_C$, $\phi'$ induces a partition $\pi'_C$ of $C$ such that $\pi'_C \subseteq \esscoal\cup \sincoal$. Let $u_{C'}(m) \equiv \sum_{i\in C'}u_i(m)$ denote the total utility a coalition $C'$ obtains from a matching $m$, we have
\[
\sum_{C'\in \pi'_C} u_{C'}(\phi',T'_C,\hchi^C(T)) > \sum_{C'\in \pi'_C} u_{C'}(\phi,T),
\]
so there exists some $\hat{C}\in \esscoal\cup \sincoal$ such that $ u_{\hat{C}}(\phi',T'_C,\hchi^C(T)) > u_{\hat{C}}(\phi,T)$. Since in the matching $(\phi',T'_C,\hchi^C(T))$ transfers are only made between matched firm and workers, $ u_{\hat{C}}(\phi',T'_C,\hchi^C(T)) = \sum_{i\in \hat{C}}v_i(\phi'(i)) $, which is the total value generated by coalition $\hat{C}$. This value can be secured by $\hat{C}$ when it blocks alone, so $\hat{C}\in \esscoal\cup \sincoal$ can profitably block the matching $(\phi,T)$, which is a contradiction.
\end{proof}

\begin{lemma} \label{Lemma-K-Efficiency}
Let $\hat{x}= \max_{\phi\in A} \sum_{i\in \players} v_i(\phi)$. If $u\in \Re^n$ satisfies $\sum_{i\in \players}u_i \le \hat{x}$ and $\sum_{i\in C} u_i  \ge \cf_C $ for all $C\in \esscoal\cup\sincoal$, then $\sum_{i\in \players}u_i = \hat{x}$.
\end{lemma}
\begin{proof}

Take any $\tilde{u}\in \Re^n$ satisfying $\sum_{i\in \players} \tilde{u}_i \le \hat{x}$ and $\sum_{i\in C} \tilde{u}_i  \ge \cf_C$ for all $C\in \esscoal\cup \sincoal$. Towards a contradiction, suppose that $\tilde{u}$ is not utilitarian efficient; that is, suppose $\sum_{i\in \players} \tilde{u}_i < \hat{x}$. Then there exists an assignment $\phi' \in A$ such that $\sum_{i \in \players} v_i(\phi') > \sum_{i \in \players} \tilde{u}_i$. Let $\pi'$ denote the partition of players induced by the matching $\phi'$. Note that $\pi'$ consists of either essential coalitions or singletons so $\pi'\subseteq \esscoal\cup \sincoal$. It follows that there exists $C' \in \pi' \subseteq \esscoal\cup \sincoal$ such that $\cf_{C'} = \sum_{i \in C'} v_i(\phi') > \sum_{i \in C'} \tilde{u}_i$, which is a contradiction to the assumption that $\sum_{i\in C} \tilde{u}_i \ge \cf_C$ for all $C\in \esscoal\cup \sincoal$. So $\tilde{u}$ must be utilitarian efficient.
\end{proof}

\begin{lemma}\label{Lemma-K-Characterization}
Let $\hat{x}= \max_{\phi\in A} \sum_{i\in \players} v_i(\phi)$. The set $\mathcal{K}$ is characterized by
\begin{equation} \label{Equation-K-Characterization}
    \mathcal{K} = \{u\in \Re^n: \sum_{i\in \players}u_i = \hat{x}, \, \sum_{i\in C} u_i  \ge \cf_C \text{ for all } C\in \esscoal\cup \sincoal \}.
\end{equation}
\end{lemma}
\begin{proof}
Take any $\tilde{u} \in \mathcal{K}$. Suppose, for the sake of contradiction, that there exists some $C \in \esscoal\cup\sincoal$ such that $\sum_{i \in C} \tilde{u}_i < \cf_C$, then $\tilde{u}$ would be blocked by $C$, which contradicts the assumption that $\tilde{u} \in \mathcal{K}$. So $\sum_{i \in C} \tilde{u}_i \geq \cf_C$ must hold for all $C \in \esscoal\cup\sincoal$. \cref{Lemma-K-Efficiency} then implies that $\tilde{u}$ is utilitarian efficient, so $\tilde{u}$ satisfies the conditions in \eqref{Equation-K-Characterization}.

For the converse, take any $\tilde{u}$ that satisfies the conditions in \eqref{Equation-K-Characterization}. We will show that $\tilde{u} \in \mathcal{K}$, i.e., there exists a core allocation $(\phi, T)$ such that $\tilde{u} = u(\phi, T)$. 

By \cref{Lemma-K-Equivalence-to-KC}, the set of core allocations is the set of matchings that cannot be profitable blocked by singletons can essential coalitions, which is nonempty under our assumptions on firm preferences \citep{kelsocrawford1982}. Since $\mathcal{K}$ is nonempty, there exists a core allocation $(\tilde{\phi}, \tilde{T})$, which by the arguments above must satisfy $\sum_{i \in \players} v_i(\tilde{\phi}) = \hat{x}$. Since $\sum_{i\in \players}\tilde{u}_i = \hat{x}$, there exists $\tilde{T}' \in \mathcal{T}$ such that $\tilde{u} = u(\tilde{\phi}, \tilde{T}')$. Note however that $\tilde{T}'$ may involve nonzero transfers between players who are not in an employment relationship, so $(\tilde{\phi}, \tilde{T}')$ may not constitute a matching. Nevertheless, let $\tilde{\pi}$ denote the partition of players induced by $\tilde{\phi}$. For every $C \in \tilde{\pi}$, it must hold that
\[
\sum_{i \in C, j \notin C} \tilde{T}'_{ij} - \sum_{i \in C, j \notin C} \tilde{T}'_{ji} = 0,
\]
for otherwise we would have $\sum_{i \in C'} \tilde{u}_i < \sum_{i \in C'} v_i(\tilde{\phi})$ for some $C' \in \tilde{\pi}$, contradicting the fact that $\tilde{u}$ satisfies \eqref{Equation-K-Characterization}. Therefore, we can construct $\tilde{T}'' \in \mathcal{T}$ such that
\[
\tilde{u} = u(\tilde{\phi}, \tilde{T}''), \quad \text{and} \quad \tilde{T}''_{ij} \neq 0 \text{ only if } i = \tilde{\phi}(j) \text{ or } i \in \tilde{\phi}(j),
\]
so $(\tilde{\phi}, \tilde{T}'')$ is a matching that induces payoff profile $\tilde{u}$. Since $\sum_{i\in C} \tilde{u}_i  \ge \cf_C \text{ for all } C\in \esscoal\cup\sincoal$,  $(\tilde{\phi}, \tilde{T}'')$ cannot be blocked by any coalition $C\in \esscoal\cup\sincoal$, so by \cref{Lemma-K-Equivalence-to-KC} $(\tilde{\phi}, \tilde{T}'')$ is a core allocation, and therefore $\tilde{u} \in \mathcal{K}$.
\end{proof}

\begin{lemma} \label{Lemma-FCR-K-Equivalence}
Let
\[
\feasiblet^{\mathcal{M}}\equiv  \conv \Big( \Big \{ u\in \Re^n:\exists (\phi,T) \in\mathcal{M}\text{ such that }u=u(\phi,T) \Big\}\Big)
\]
denote the convex hull of all feasible matching payoffs. Then
\[
\Big\{u\in \feasiblet^{\mathcal{M}}: \sum_{i\in C}u_i \ge \cf_C \text{ for all } C\in \esscoal\cup\sincoal \Big\} = \mathcal{K}.
\]
\end{lemma}
\begin{proof}
The fact that $\mathcal{K} \subseteq \big\{u\in \feasiblet^{\mathcal{M}}: \sum_{i\in C}u_i \ge \cf_C \text{ for all } C\in \esscoal\cup\sincoal \big\} $ follows from the definition of $\mathcal{K}$. To show $\big\{u\in \feasiblet^{\mathcal{M}}: \sum_{i\in C}u_i \ge \cf_C \text{ for all } C\in \esscoal\cup\sincoal \big\} \subseteq \mathcal{K}$, take any $\tilde{u}\in \feasiblet^{\mathcal{M}}$, since $\tilde{u}$ is a convex combination of feasible payoff vectors, it must be that 
\[
\sum_{i\in \players}\tilde{u}_i \le \hat{x}\equiv \max_{\phi\in A} \sum_{i\in \players} v_i(\phi).
\]
\cref{Lemma-K-Efficiency} then implies that $\sum_{i\in \players}\tilde{u}_i = \hat{x}$, so $\tilde{u}\in \mathcal{K}$ by \cref{Lemma-K-Characterization}.
\end{proof}

\paragraph{Proof of \cref{Proposition-KC-Secret}.} A model in which all essential coalitions can offer private wage terms corresponds to restricting attention to public PCEs in the model with secret transfers (\cref{Section-SecretTransfers}) that do not condition on any past wage terms. 

We first prove that every payoff vector in $\mathcal{K}$ can be supported by a public PCE that satisfies this measurability condition. For any $\tilde{u}\in \mathcal{K}$, there exists a core allocation $(\phi,T)$ such that $\tilde{u}= u(\phi,T)$. Consider the plan $\tilde{\sigma}$ in which $\tilde{\sigma}(h) =  (\phi,T)$ for every history $h\in \historiesTU$. This plan is public and does not condition on past wages, and produces discounted payoff profile $\tilde{u}$. Given that $(\phi,T)$ is a core allocation, $\tilde{\sigma}$ is also a PCE.

We now prove that for every $\delta\ge 0$, every public PCE in the model with private wages implements a discounted payoff profile in $\mathcal{K}$. Given that every such public PCE is also a public PCE with secret transfers, it suffices to show that this conclusion holds for public PCE with secret transfers. We apply \cref{Theorem-Secret} to establish this conclusion. All steps of the proof of \Cref{Theorem-Secret} carry over to the matching environment without adaptation except the step that invokes reflexivity of the effectivity correspondences (see \cref{footnote: reflexivity}). However, the assignment $\phi^{\circ}$ in which all players are unmatched satisfies $\max_{\phi'\in E_C(\phi^\circ)} \sum_{i\in C}v_i(\phi') = \min_{\phi\in A}\max_{\phi'\in E_C(\phi)} \sum_{i\in C}v_i(\phi') = \cf_C$ for every coalition $C$. Hence, the PCE can specify $\phi^{\circ}$ to minmax any coalition $C \in \seccoal \cup \sincoal$. Moreover, $\phi^{\circ} \in E_C(\phi^{\circ})$ for every coalition $C$. Thus, the argument for this step also goes through. Since $\esscoal\subseteq \seccoal$, \Cref{Theorem-Secret} implies that for every $\delta\ge 0$, every discounted payoff profile $\tilde{u}$ produced by a public PCE must satisfy $\sum_{i\in C}\tilde{u}_i \ge \cf_C \text{ for all } C\in \esscoal\cup\sincoal$, so $\tilde{u}\in \big\{u\in \feasiblet^{\mathcal{M}}: \sum_{i\in C}u_i \ge \cf_C \text{ for all } C\in \esscoal\cup\sincoal \big\}$. By \cref{Lemma-FCR-K-Equivalence}, $\tilde{u}$ must be an element of $\mathcal{K}$. \qed

\subsection{Proof of \cref{Proposition-Wages} on p. \pageref{Proposition-Wages}}\label{Section-SupplementaryProp3}

\paragraph{Preliminary Results.}

\begin{lemma} \label{Lemma-Core-Wage}
    All static stable matchings fill slots $\{(f,l):\rho(f,l)\ge  \max\{0, \eta(L)\}\}$ while leaving other slots vacant; more over, all workers receive the same payoff $r$ where $\max\{0,\eta(L+1)\} \le r \le  \max\{0,\eta(L)\}$.
\end{lemma}

\begin{proof}
We break down the proof into two parts.

{ \noindent \underline{Part 1}: All static stable matchings fill slots $\{(f,l):\rho(f,l)\ge  \max\{0, \eta(L)\}\}$ while leaving other slots vacant.}

Let $m$ be any static stable matching. We first show that $m$ must be utilitarian efficient. Suppose, for the sake of contradiction, that $m$ is not utilitarian efficient. Then there exists a reassignment of players that increases players' total payoff, which implies the existence of $f\in \firms$ and $W\subseteq \workers$ such that $v_f(W) +\sum_{w\in W} v_w(f) > u_f(m)+ \sum_{w\in W} u_w(m)$. But this implies that $m$ is profitably blocked by $(f,W)$, contradicting the stability of $m$.

Next, we show that since $m$ is utilitarian efficient, it fills all slots in $\{(f,l):\rho(f,l)\ge  \max\{0, \eta(L)\}\}$. Suppose, for the sake of contradiction, that there exists a slot $(\tilde{f}, \tilde{l}) \in \{(f,l):\rho(f,l)\ge  \max\{0, \eta(L)\}\}$ that is not filled. Let $\tilde{l}^* = \min \{ l : (\tilde{f}, l) \text{ is unfilled} \}$ denote the first unfilled position at firm $\tilde{f}$. Since firms have diminishing marginal products, we have $\rho(\tilde{f}, \tilde{l}^*) \geq \rho(\tilde{f}, \tilde{l})$, so $(\tilde{f}, \tilde{l}^*)$ is an open slot in $\{(f,l):\rho(f,l)\ge  \max\{0, \eta(L)\}\}$ that is immediately accessible by workers. Since there are $L$ workers in total, if not all slots in $\{(f,l):\rho(f,l)\ge  \max\{0, \eta(L)\}\}$ are filled, there exists ${w}' \in W$ who is either unemployed or filling a slot outside of $\{(f,l):\rho(f,l)\ge  \max\{0, \eta(L)\}\}$. In the first scenario, matching ${w}'$ to the unfilled slot $(\tilde{f}, \tilde{l}^*)$ would strictly increase the total surplus. In the second scenario, let $(\hat{f},\hat{l})$ be the slot filled by $w'$, and let $\hat{l}^* = \max \{ l : (\hat{f}, l) \text{ is filled} \}$ denote the last occupied slot at firm $\hat{f}$, and $\hat{w}^*$ denote the worker filling $(\hat{f},\hat{l}^*)$. It follows from decreasing marginal product that $(\hat{f},\hat{l}^*)$ is also outside of $\{(f,l):\rho(f,l)\ge  \max\{0, \eta(L)\}\}$, so matching $\hat{w}^*$ to the unfilled slot $(\tilde{f}, \tilde{l}^*)$ instead would strictly increase the total surplus, again contradicting the utilitarian efficiency of $m$. Thus, all stable matchings must fill the slots in $\{(f,l):\rho(f,l)\ge  \max\{0, \eta(L)\}\}$.

To show that all slots outside of $\{(f,l):\rho(f,l)\ge  \max\{0, \eta(L)\}\}$ are vacant, there are two cases to consider. If $\eta(L)>0$, we know from the arguments above that the $L$ slots in $\{(f,l):\rho(f,l)\ge   \eta(L)\}$ are filled, so all other slots must be vacant. If $\eta(L)<0$, then the set $\{(f,l):\rho(f,l)\ge  \max\{0, \eta(L)\}$ becomes $\{(f,l):\rho(f,l)\ge 0 \}$, and let us suppose, for the sake of a contradiction, that some slot $(\tilde{f}, \tilde{l})$ with $\rho(\tilde{f}, \tilde{l})< 0$ is filled. Let $\tilde{l}^* = \max \{ l : (\tilde{f}, l) \text{ is filled} \}$ denote the last filled slot at firm $\tilde{f}$, and let $\tilde{w}$ denote the worker matched to this position. Due to decreasing marginal returns, we have $\rho(\tilde{f}, \tilde{l}^*) <0$ as well, so simply unmatching $\tilde{w}$ from $(\tilde{f}, \tilde{l}^*)$ will increase the total surplus. This contradicts the efficiency of $m$, which implies that $m$ would not not stable. It follows that all stable matchings must leave slots outside $\{(f,l):\rho(f,l)\ge  \max\{0, \eta(L)\}\}$ vacant.
\medskip

{ \noindent \underline{Part 2}: All workers receive the same payoff $r$, where $\max\{0,\eta(L+1)\} \le r \le  \max\{0,\eta(L)\}$.}

First we establish that all workers receive the same payoff. Take any static stable matching $m$. From Part 1, all positions in $\{(f,l):\rho(f,l)\ge  \max\{0, \eta(L)\}\}$ are filled. We prove that all workers have the same payoff under two separate cases. 

First, suppose $\eta(L)<0$. It follows that $|\{(f,l):\rho(f,l)\ge  \max\{0, \eta(L)\}\}|<L$, so in the stable matching $m$ there exists a worker $\tilde{w}$ who is unmatched. This means $\tilde{w}$ receives $0$ payoff in the stable matching $m$. It then follows that any other employed worker must also receive $0$ payoff, since otherwise there is a profitable block where their employer replaces them with $\tilde{w}$.

Second, suppose $\eta(L)>0$. Since by assumption $\rho(f,L)<\max\{0,\eta(L)\}$, there exists ${f}_1$ and ${f}_2$ such that both $f_1$ and $f_2$ employ workers in $m$. Since workers are identical, each worker working for $f_1$ must receive the same payoff as any worker at $f_2$ in $m$. This implies that workers at $f_1$ and $f_2$ all have the same payoff. The same argument applies to workers employed by any other firm, so all workers receive the same payoff.

Let $r$ denote the payoff that workers receive, we next show that $\max\{0,\eta(L+1)\} \le r \le  \max\{0,\eta(L)\}$. It is obvious that $r\ge 0$ by workers' individual rationality. To complete the arguments, it suffices to demonstrate the validity of three statements: A. $r \ge \eta(L+1)$ if $\eta(L+1)>0$; B. $r \le \eta(L)$ if $\eta(L)>0$; and C. $r=0$ if $\eta(L)\le 0$.

Statement A: if $\eta(L+1)>0$, then decreasing marginal return implies $\eta(L)>0$, so from Part 1 we know all $L$ workers are assigned to $\{(f,l):\rho(f,l)\ge  \max\{0, \eta(L)\}\}= \{(f,l):\rho(f,l)\ge  \eta(L)\}$. Let $(\tilde{f},\tilde{l})$ denote the slot with value $\rho(\tilde{f},\tilde{l}) = \eta(L+1)$. By decreasing marginal return, any slot in $\{(\tilde{f},{l}): l< \tilde{l} \}$ at $\tilde{f}$ is in $\{(f,l):\rho(f,l)\ge  \eta(L)\}$ and already filled. It follows that $r\ge \eta(L+1)$, since otherwise $\tilde{f}$ can profitably block $m$ by poaching a worker from other firms, which generates additional surplus $\eta(L+1)$, while offering her a payoff $r'$ satisfying $\eta(L+1) > r' > r$.

Statement B: if $\eta(L)>0$, again from Part 1 we know that all $L$ workers are assigned to $\{(f,l):\rho(f,l)\ge  \eta(L)\}$. Let $(\tilde{f},\tilde{l})$ be the slot such that $\rho(\tilde{f},\tilde{l})= \eta(L)$. By decreasing marginal return we know that $(\tilde{f},\tilde{l})$ must be the last filled slot at firm $\tilde{f}$. It follows that workers' payoff is no more than $\eta(L)$ since otherwise $\tilde{f}$ can profitably block by firing the worker matched to the slot $(\tilde{f},\tilde{l})$.

Statement C: if $\eta(L) < 0$, then there are at most $(L-1)$ slots with a positive surplus, which by Part 1 implies that in any stable matching there exists a worker $\tilde{w}$ who is unmatched. In this case, workers' payoff must be $0$ since otherwise the matching is profitably blocked by a firm replacing one of its employees with worker $\tilde{w}$.

Combining statements A, B, and C lets us conclude that $\max\{0,\eta(L+1)\} \le r \le  \max\{0,\eta(L)\}$.
\end{proof}

\paragraph{Proof of \cref{Proposition-Wages}.} The first half of \cref{Proposition-Wages} follows from \cref{Proposition-KC-Public}, while the second half of \cref{Proposition-Wages} follows from combining \cref{Proposition-KC-Secret} and \cref{{Lemma-Core-Wage}}.
\qed

\subsection{Proof of \cref{Proposition-Comparative-Wages} on p. \pageref{Proposition-Comparative-Wages}}\label{Section-SupplementaryProp4}
Since by assumption both markets $M_1$ and $M_2$ satisfy $\eta_i(L+1)>0$, \cref{Lemma-Core-Wage} implies that in each market $M_i$, where $i=1$ or $2$, all static stable matchings fill the slots in $\{(f,l):\rho_i(f,l)\ge  \max\{0, \eta_i(L)\}\}=\{(f,l):\rho_i(f,l)\ge \eta_i(L)\}$. Moreover, the workers' payoff $r$ in market $M_i$ satisfies $ \eta_i(L+1) \le r\le \eta_i(L)$. 
Recall that the total surplus is $\Pi_i\equiv \sum_{\ell=1}^L \eta_i(\ell)$, while the set of potential workers' total surplus is $\Pi^{\mathcal W}_{i} = [L\eta_i(L+1),L\eta_i(L)]$, and the set of potential firms' total surplus is
\[
   \Pi_i^{\mathcal F} = \Pi_{i}  - \Pi^{\mathcal W}_{i} =\Big[\sum_{\ell=1}^L \eta_i(\ell)-L\eta_i(L), \, \sum_{\ell=1}^L \eta_i(\ell) - L\eta_i(L+1) \Big].
\]
To simplify notation let us denote $\underline{b}_i^{\workers} \equiv L \eta_i(L+1)$ and $\overline{b}_i^{\workers} \equiv L\eta_i(L)$, so $\Pi^{\mathcal W}_{i} = [\underline{b}^{\workers}_i, \overline{b}^{\workers}_i]$. Similarly, let $\underline{b}^{\firms}_i \equiv \sum_{\ell=1}^L \eta_i(\ell) - L\eta_i(L)$ and $\overline{b}^{\firms}_i = \sum_{\ell=1}^L \eta_i(\ell) - L\eta_i(L+1)$,
so $\Pi_i^{\mathcal F} = [\underline{b}^{\firms}_i, \overline{b}^{\firms}_i]$.

Let $s\equiv \eta_1(1) = \eta_2(1)$. For each $2 \le \ell \le L$, define $\Delta^i_{\ell} \equiv \eta_i(\ell-1)-\eta_i(\ell)$, so $\eta_i(\ell) = s- \sum_{k=2}^{\ell} \Delta^i_k$ for all $\ell \ge 2$.
It follows that 
\begin{align*}
     \sum_{\ell=1}^L \eta_i(\ell) =& s L - \sum_{\ell=2}^L (L+1-\ell)\Delta^i_{\ell}, \\
     L\eta_i(L) = sL - L\sum_{\ell=2}^L\Delta^i_{\ell}, \;\; &\text{ and }
     L\eta_i(L+1) = sL - L\sum_{\ell=2}^{L+1}\Delta^i_{\ell}.
\end{align*}
This allows us to express the bounds for firms' and workers' aggregate surplus in terms of $s$ and $\Delta^i_{\ell}$'s, yielding
\begin{alignat}{3} 
\underline{b}^{\workers}_i &=  sL - L\sum_{\ell=2}^{L+1}\Delta^i_{\ell}, &\text{ and }   \overline{b}^{\workers}_i &= sL - L\sum_{\ell=2}^L\Delta^i_{\ell},    \label{Equation-Wsurplus-Bounds}\\
    \underline{b}^{\firms}_i &= \sum_{\ell=2}^L (\ell-1)\Delta^i_{\ell}, &\text{ and }\overline{b}^{\firms}_i &=  \sum_{\ell=2}^{L+1}(\ell-1)\Delta^i_{\ell}.    \label{Equation-Fsurplus-Bounds}
\end{alignat}

Market $M_2$ exhibiting more steeply decreasing returns than $M_1$ is equivalent to $\Delta^2_{\ell} \ge \Delta^1_{\ell}$ for all $2\le \ell \le L+1$, which implies $\eta_2(\ell) \le \eta_1(\ell)$ for all $1 \le \ell \le L+1$,  so $\Pi_2 \le \Pi_1$.

In \eqref{Equation-Wsurplus-Bounds}, all the $\Delta^i_{\ell}$'s enter the bounds for worker surplus with negative coefficients, so $\underline{b}^{\workers}_2 \le \underline{b}^{\workers}_1$ and  $\overline{b}^{\workers}_2 \le \overline{b}^{\workers}_1$, where the inequalities are strict if $M_2$ has strictly more steeply decreasing returns than $M_1$. By contrast, in \eqref{Equation-Fsurplus-Bounds} the $\Delta^i_{\ell}$ terms enter the bounds with positive coefficients, so $\underline{b}^{\firms}_2 \ge   \underline{b}^{\firms}_1$ and $\overline{b}^{\firms}_2 \ge \overline{b}^{\firms}_1$, where, again, the inequalities are strict if $M_2$ has strictly more steeply decreasing returns than $M_1$. Together, the directions of change for these bounds imply $\Pi^{\mathcal{W}}_2 \preceq_{SSO} \Pi^{\mathcal{W}}_1$ and $\Pi^{\mathcal{F}}_2\succeq_{SSO} \Pi^{\mathcal{F}}_1$, with strict set orders if $M_2$ has strictly more steeply decreasing returns than $M_1$.\qed

\section{Perfect Coalition-Proof Equilibrium}\label{Section-PCPE}
In this subsection, we develop a solution concept that combines ingredients of Perfect Coalitional Equilibrium and coalition-proofness in the sense of \cite{bernheim1987coalition}. We focus on strategic form games (\Cref{Example-NormalFormGame}) and show that \cref{Theorem-NTU} continues to hold without change.\footnote{We consider strategic form games for two reasons. First, \cite{bernheim1987coalition} formulate their solution concept for this class of games. Second, our definition below exploits the product structure of the alternative set $A$ in a strategic form game. A similar construction would also apply to characteristic function and matching games (\cref{Example-CoalitionalGame}), wherein a player's payoff depends only on the identify of her coalition partners.} 

Consider any strategic form game $G$ with player set $\players$, action set $A \equiv \times_{i=1}^n A_i$, where each player $i$ has payoff function $v_i : A \rightarrow \Re$. For an action profile $a\in A$ and proper coalition $C\in \coalitions\backslash\{\players\}$, let $G|a_{-C}$ denote the induced game with player set $C$, action set $A_C = \times_{i\in C}A_i$, and payoff functions $\{v^C_i\}_{i\in C}$ where each $v^C_i:A_C\rightarrow \Re$ is defined by $v^C_i(a'_C) = v_i(a'_C,a_{-C})$ for all $a'_C\in A_C$. We first restate \citeauthor{bernheim1987coalition}'s definition of coalition-proof Nash equilibrium.

\begin{definition}\label{Definition-CPE}
\begin{enumerate}
    \item  In a single player game $G$, action $a_i^*$  is a \textbf{coalition-proof Nash Equilibrium} (CPNE) if and only if $a_i^*$ maximizes $v_i(a_i)$.
    \item Suppose the number of players is $n>1$ and assume that coalition-proof Nash equilibrium has been defined for games with fewer than $n$ players. Then,
\begin{enumerate}
    \item For any game $G$ with $n$ players, $a^* \in A$ is \textbf{self-enforcing} if, for all proper sub-coalitions $C \in \coalitions\backslash \{\players \}, a_C^*$ is a coalition-proof Nash equilibrium in the game $G| a^*_{-C}$.
    \item For any game $G$ with $n$ players, $a^* \in A$ is a coalition-proof Nash equilibrium if it is self-enforcing and if there does not exist another self-enforcing action profile $a \in A$ such that $v_i(a)>v_i(a^*)$ for all $i=1, \ldots, n$.
\end{enumerate} 
\end{enumerate}    
\end{definition}
We use this definition to formulate a coalition-proof analog of PCE. In the repeated game, fix a plan $\sigma$ and history $h$. Let $G(\sigma,h)$ denote the continuation game $G(\sigma,h)$, defined as the strategic form game in each player $i$ chooses an action from $A_i$ and has the payoff function $V_i(a) \equiv (1-\delta) v_i(a) + \delta U_i(h, a |\sigma)$.\footnote{Because we operate in a strategic form game, we simplify our notation by recording only the past action profiles and not the identity of blocking coalitions.} 

\begin{definition} \label{Definition-PCPE}
    A plan $\sigma$ is a \textbf{perfect coalition-proof equilibrium} (PCPE) if, at every history $h$, the default action profile prescribed at history $h$, $\sigma(h)$, is a CPNE in the continuation game $G(\sigma,h)$.
\end{definition}

Translating this definition to profitable blocking, a block by coalition $C$ invalidates a plan $\sigma$ from being a PCPE if it is profitable (as in \cref{Definition-PCE1}) and is itself immune to any further ``credible'' blocks by a sub-coalition of $C$. Hence, every PCE is a PCPE but the converse need not hold. The following result shows, however, that the set of PCPE payoffs coincides with the set of PCE payoffs as $\delta\rightarrow 1$. 

\begin{rtheorem}{\ref{Theorem-NTU}$^*$}\label{Theorem-PCPE}
    For every $\delta \geq 0$, every PCPE gives each player $i$ a payoff of at least $\cm_i$. Moreover, for every $v\in {\fcr}$, there is a $\underline{\delta}<1$ such that for every $\delta\in(\underline{\delta},1)$, there exists a PCPE with discounted payoff equal to $v$.
\end{rtheorem}

\Cref{Theorem-PCPE} establishes that for every discount factor $\delta$, no player $i$ can be driven below her coalitional minmax $\cm_i$ (defined on p. \pageref{Equation-coalitionminmax}). Recall that $C(i)$ is the (potentially singleton) set of players whose utilites are equivalent to player $i$'s, and $\cm_i$ corresponds to the minmax $ \min_{a\in \alternatives} \max_{a'\in E_{C(i)}(a)} v_i(a')$ that would come if coalition $C(i)$ behaved as a unitary agent. This lower bound for PCPE payoffs coincides with that which \Cref{Theorem-NTU} obtained for PCE payoffs. The second sentence of \Cref{Theorem-PCPE} asserts that every feasible payoff profile that gives each player strictly more than her coalitional minmax is supportable as a PCPE. 

\begin{proof}[Proof of \Cref{Theorem-PCPE}]
    We observe that the second sentence of \Cref{Theorem-PCPE} follows from \Cref{Theorem-NTU} and the fact that every PCE is a PCPE. Therefore, it suffices to establish the first sentence of \Cref{Theorem-PCPE}. 
    
     We establish this claim by proving its contrapositive: let $\sigma$ be a plan, and suppose there exists a player ${i}$ that satisfies $U_{i}(\emptyset|\sigma) < \cm_{i}$, we will prove that $\sigma$ cannot be a PCPE. Part 1 of the proof of \Cref{Theorem-NTU} establishes that coalition $C(i)$ has a profitable blocking plan: choosing its myopic best response to the default alternative at every history would result in a payoff no less than $\cm_{i}$. \Cref{Lemma:One-Shot-Deviation-NTU} shows that if the coalition $C(i)$ has a profitable blocking plan, then there exists some history $h$ at which $C(i)$ can profitably block. Let $a^*\equiv \sigma(h)$ denote the action profile prescribed at that history by the plan $\sigma$. The existence of a profitable block by coalition $C(i)$ implies that there exists $a'_{C(i)}\in A_{C(i)}$ such that for all $j\in C(i)$,
\begin{align}\label{Inequality-PCPE}
\begin{aligned}
&(1-\delta) v_j(a'_{C(i)},a^*_{-C(i)} ) + \delta U_j(h,a'_{C(i)},a^*_{-C(i)} |\sigma)\\ >&\, (1-\delta) v_j(a^*_{C(i)},a^*_{-C(i)} ) + \delta U_j(h,a^*_{C(i)},a^*_{-C(i)} |\sigma).
\end{aligned}
\end{align}
Observe that the RHS corresponds to $V_j(a^*)$ in the game $G(\sigma,h)$. We use the following lemma---which we prove afterwards---that this inequality cannot hold in a PCPE. 
     \begin{lemma} \label{Lemma-PCPEalignment}
    Suppose $G$ is a strategic form game in which all players share equivalent utilities.\footnote{In other words, for every pair of players $i,j$, there exist some $\alpha >0$ and $\beta\in \Re$ such that $v_i(a) = \alpha v_j(a) + \beta$ for all $a\in A$.} An action profile $a^*$ is a coalition-proof Nash equilibrium of $G$ if and only if $a^*\in \argmax_{a\in A}v_1(a)$.
\end{lemma}
\Cref{Lemma-PCPEalignment} asserts that in a strategic form game in which all players share equivalent utilities, every CPNE maximizes players' payoffs. \cref{Definition-PCPE}
implies that were $\sigma$ a PCPE, $a^*=\sigma(h)$ must be a CPNE in $G(\sigma, h)$, so then by \Cref{Definition-CPE}(b)(i), $a^*_{C(i)}$ must be a CPNE in the game $G(\sigma,h)|a^*_{-C(i)}$. Applying \Cref{Lemma-PCPEalignment} to the game $G(\sigma,h)|a^*_{-C(i)}$ yields that $a^*_{C(i)}$ maximizes $V_j(a_{C(i)}, a^*_{-C(i)})$ for all $j\in C(i)$. This property contradicts \eqref{Inequality-PCPE}, implying that $\sigma$ is not a PCPE.
\end{proof}

\begin{proof}[Proof of \Cref{Lemma-PCPEalignment}]
We first prove that every $a^* \in \argmax_{a \in A} v_1(a)$ is a coalition-proof Nash equilibrium by induction over the number of players $k$. 

Suppose $k=2$. Given any $a^* \in \argmax_{a \in A} v_1(a)$, since players have equivalent utilities, each player $i$ maximizes $v_i$ by playing $a^*_i$. Hence $a^*$ is self-enforcing in the sense of  \cref{Definition-CPE}(b)(i). Moreover, since $a^*$ is Pareto efficient among all action profiles, it cannot be strictly Pareto dominated by another self-enforcing profile. Therefore, $a^*$ is a coalition-proof Nash equilibrium.

Suppose that $G$ has $|\players|= n$ players. Suppose also that we have established that for any game $G(k)$ with $k<n$ players, action space $A^k$, and payoff functions $\{v^k_i\}_{i=1}^k$, every element of $\argmax_{a^k \in A^k} v^k_1(a^k)$ is a coalition-proof Nash equilibrium.  Let $a^*\in \argmax_{a \in A} v_1(a)$. Let $C \in \coalitions \backslash \{\players\}$ be any proper coalition and consider the induced $|C|$-player game $G|a^*_{-C}$. By definition, $a^*_C \in \argmax_{a_C \in A_C} v_1(a_C, a^*_{-C})$ so the inductive hypothesis implies that $a^*_C$ is a coalition-proof Nash equilibrium in the game $G|a^*_{-C}$. Hence $a^*$ is self-enforcing in $G$ in the sense of  \cref{Definition-CPE}(b)(i). Furthermore, $a^*$ cannot be strictly Pareto dominated by another self-enforcing profile, so $a^*$ is a coalition-proof Nash equilibrium in $G$.

Having proved that every element of $\argmax_{a \in A} v_1(a)$ is a coalition-proof Nash equilibrium, it then follows from \cref{Definition-CPE}(b)(ii) that no action profile outside of $\argmax_{a \in A} v_1(a)$ can be a coalition-proof Nash equilibrium.
\end{proof}

\section{An Application to Distributive Politics}\label{Section-Distributive}

This section studies a repeated distribution problem, in which the players repeatedly choose how to divide a dollar. Such division problems feature prominently in the political economy literature \citep[e.g.][]{baron1989bargaining} and relate to the \emph{simple games} \citep{von1945theory} studied in cooperative game theory. 

In this game, $A\equiv\{a \in \Re^{\players}_+: \sum_{i\in \players} a_i = 1\}$, where player $i$'s payoff from alternative $a$ is $a_i$. Divisions are chosen by a ``winning'' coalition: $\mathcal{W}$ is a set of coalitions such that for every coalition $C$ in $\mathcal{W}$, $E_C(a)=A$ for every division $a$, and for every coalition $C$ not in $\mathcal{W}$, $E_C(a)=\{a\}$. As standard, $\mathcal W$ is monotone and proper.\footnote{In other words, if $C$ is in $\mathcal W$, then $\mathcal{W}$ contains every superset of $C$ but not its complement.} A simple-majority rule protocol corresponds to $\mathcal W\equiv \{C\in \coalitions:|C|\geq \frac{n+1}{2}\}$. This formulation also allows for \emph{veto power}: if a player belongs to every winning coalition ($\cap_{C\in \mathcal{W}} C$), then effectively no block can happen without her approval. We denote the set of veto players by $D \equiv \cap_{C\in \mathcal{W}} C$. 

 \cite{BS2009} study simple majority rule, emphasizing how Dynamic Condorcet Winners exist although the stage game lacks a Condorcet Winner. We focus instead on settings with at least one veto and one non-veto player, and in which veto players are not dictators ($D\notin \mathcal W$). Absent history dependence, these settings lead to highly unequal splits: the veto players steal the entire dollar, emerging as \emph{de facto} dictators of the game. Formally, the set of core alternatives of the stage game is $\mathcal{K}\equiv \{a\in \Re^{\players}_+:\sum_{i\in D} a_i=1\}$. The logic is that any division that gives a positive share to a non-veto player would be profitably blocked by a winning coalition who would extract that share and divide it among themselves.

 Against this backdrop, we evaluate how history dependence can counter this tendency towards unequal splits. Consider a three-player example in which player $1$ alone has veto power; however, she needs the support of at least one other player to block. Player $1$ captures the entire dollar in the stage-game core. Nevertheless, relatively simple schemes in the repeated game can promote equal splits. Consider a core reversion plan that prescribes $\left(\frac{1}{3},\frac{1}{3},\frac{1}{3}\right)$ every period if that has been the division up to now and switches to the stage-game core otherwise. On the equilibrium path, even if player $1$ offers the entire dollar to either player $2$ or $3$, neither finds it profitable to block with her if $(1-\delta)(1)+\delta(0)\leq \frac{1}{3}$. Going further, core-reversion can support any division in the triangle formed by the vertices $\{(2\delta-1,1-\delta,1-\delta),(0,\delta,1-\delta),(0,1-\delta,\delta)\}$, which converges to the unit simplex as $\delta\rightarrow 1$. 

One could go beyond core-reversion to characterize all PCE payoffs. Because the game is convex and exhibits default-independent power, \Cref{Theorem-Stationary} implies that all PCE payoffs can be supported by stationary PCE. Using this result, we find that if players are sufficiently patient, then every payoff in which each non-veto player obtains up to $\delta$ can be supported in a PCE. We depict these outcomes in \Cref{Figure-APS}. 

\begin{figure}[t!]
          \begin{subfigure}{.5\linewidth}
            \centering
            \includegraphics[width=2.5in]{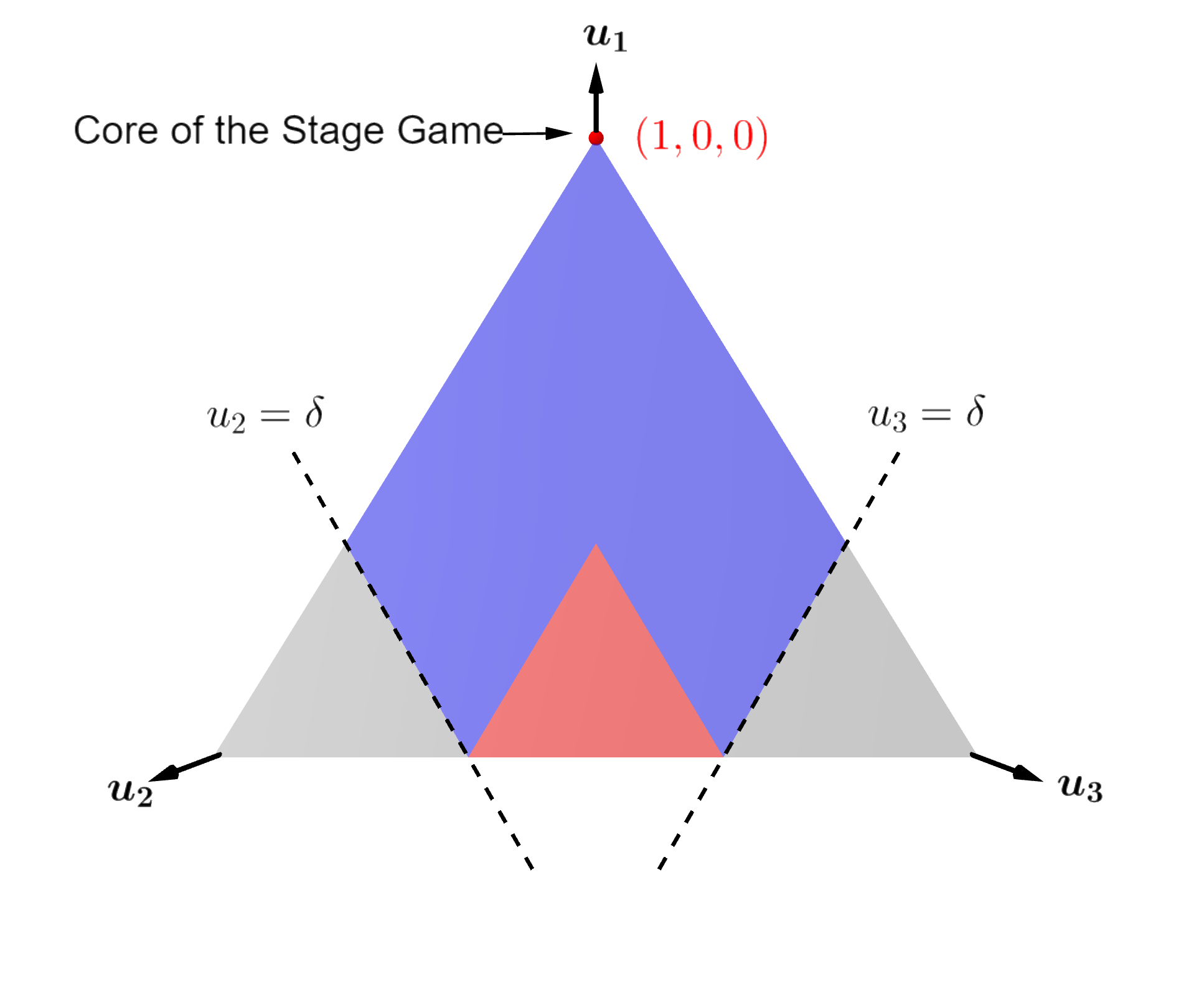}
            \vspace{-0.15in}
            \caption{Perfect monitoring for $\delta>1/2$}\label{Figure-APS}
          \end{subfigure}%
          \begin{subfigure}{.5\linewidth}
            \centering
            \includegraphics[width=2.5in]{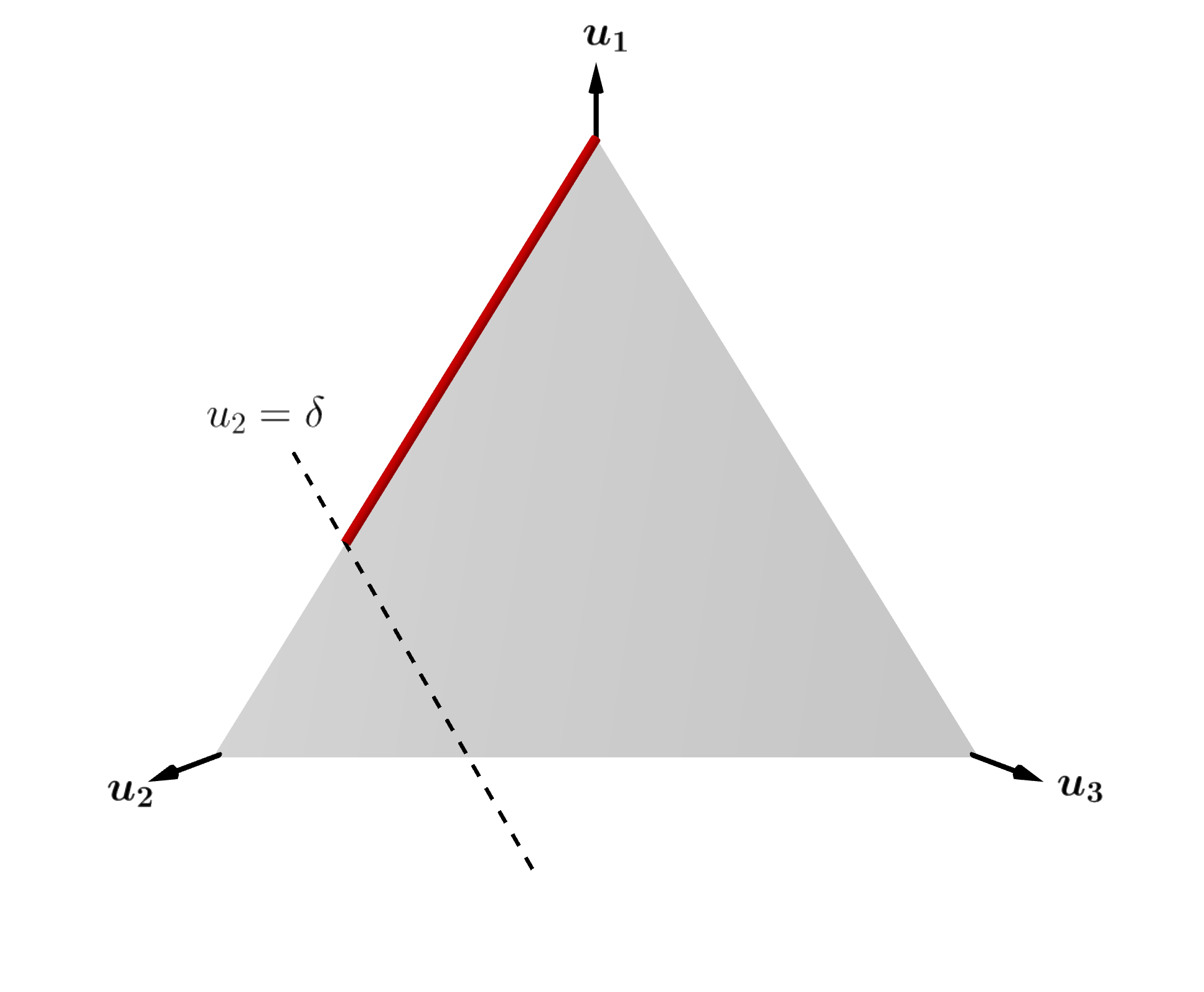}
            \vspace{-0.15in}
            \caption{Secret transfers between $1$ and $2$}\label{Figure-SecretTransfer}
          \end{subfigure}\vspace{0.1in}
          \caption{(A) depicts the set of supportable outcome. The red region depicts payoffs supported by core-reversion, and the blue region illustrates those from other PCE. (B) shows the set of supportable payoffs once coalition $\{1,2\}$ can make secret transfers; player $3$ then obtains $0$.}\label{Figure-SimpleGames}
\end{figure}

These schemes collapse if the veto player can make and receive secret side-payments. Suppose players $1$ and $2$ can transfer utility under the table. \Cref{Theorem-Secret} implies that player $3$ then obtains $0$ in all PCE payoffs, as illustrated in \Cref{Figure-SecretTransfer}. Even worse, player $1$ takes the entire dollar in every period if she can make secret side-payments with each player.
 
These intuitions generalize to $n$-player games in which there are at least one veto and one non-veto player, and veto players are not dictators. We call a coalition $C$ a \emph{minimal winning coalition} if $C$ is a winning coalition and every proper subset is not.

\begin{proposition} \label{Proposition:Simple-Games-Payoffs}
The following hold:
\begin{enumerate}[label=\emph{(\alph*)}]
\item Absent secret transfers, there exists $\underline\delta\geq 0$ such that if $\delta \geq \underline\delta$, the set of supportable payoffs are those that give at least $(1-\delta)$ to each winning coalition. 
\item A winning coalition $C$ obtains the entire dollar in every period in every PCE, regardless of $\delta$, if it can make secret transfers. 
\item The veto players obtain the entire dollar in every period in every PCE, regardless of $\delta$, if every minimal winning coalition can make secret transfers.
\end{enumerate}
\end{proposition}

\Cref{Proposition:Simple-Games-Payoffs}(a) highlights how egalitarian schemes can be supported by history dependence. We use \Cref{Theorem-Stationary} to obtain this fixed discount factor characterization; it turns out that $\underline\delta=0$ if there are at least two veto players so the characterization then is complete. \Cref{Proposition:Simple-Games-Payoffs}(b) and (c) elucidate how secret side-payments destabilize egalitarian schemes: the veto players regain \emph{de facto} dictatorial power if every minimal winning coalition can make secret transfers.

Below, we prove \cref{Proposition:Simple-Games-Payoffs}. We will use an alternative $a\in \alternatives$ to also represent its generated payoff profile $v(a)$. We establish two preliminary results. \Cref{lemma:single-veto-zero-vectors} establish the existence of ``punishment PCEs'' $\{\sigma^i\}_{i=1}^n$ that guarantee $U_i(\emptyset|\sigma^i)=0$ for each player $i$.  \Cref{lemma:OPC} proves that any PCE can be enforced by punishments where every member of a deviating coalition simultaneously obtains $0$.

\begin{lemma}\label{lemma:single-veto-zero-vectors}
	Under perfect monitoring, for every player $i\in \players$, there is a PCE $\sigma^i$ such that $U_i(\emptyset|\sigma^i)=0$ when $\delta > \frac{n -2}{n -1}$.
\end{lemma}

\begin{proof}
We consider two case, $|D|=1$ and $|D|\ge 2$. The case where $|D|=1$ requires the discount factor to be sufficiently high. The case where there are two or more veto players ($|D|\geq 2$) applies for every discount factor. 
\smallskip

\noindent  \underline{Case 1: $\abs{D}=1$.}
    Suppose without loss of generality that $D$ consists of player $1$. Let $\hat{a} \equiv (1, 0, \ldots, 0)$ denote the unique core alternative, and $\overline{a} \equiv (0, \frac{1}{n-1}, \ldots, \frac{1}{n-1})$ denote the alternative that equally divides the total payoff among all non-veto players.

    For $i\ne 1$, let $\sigma^i$ be the plan that specifies the core alternative $\hat{a}$ as default after every history, so each $\sigma^i$ is a PCE that satisfies $U_i(\emptyset|\sigma^i)=0$

	For $i=1$, let $\sigma^1$ be the plan that specifies $\overline{a}$ on path, and $\hat{a}$ at any history where an alternative distinct from $\overline{a}$ has been chosen in the past. Note that $U_1(\emptyset|\sigma^1)=0$. We will verify that $\sigma^1$ is a PCE. No coalition can profitably block once continuation play reverts back to the core alternative. On the path of play, consider a winning coalition $C \in \mathcal{W}$ blocking and choosing alternative $a'$. Since the game is non-dictatorial, if $C$ is a winning coalition, player $1$ cannot be its only member. Let $j\ne 1$ be a player in $C$. Since $a'_j \le 1$, we have 
	\[
	(1-\delta)a'_j + \delta 0 \le 1-\delta \le \frac{1}{n-1}
	\]
	so player $j$ prefers following the plan $\sigma^1$ over blocking and reverting to the core alternative. As a result, no coalition $C$ can profitably block the plan $\sigma^1$ at any history, so $\sigma^1$ is a PCE.
\smallskip

    \noindent  \underline{Case 2: $\abs{D}\ge 2$.} Without loss of generality, suppose $\{1,2\} \subseteq D$. Let $a^1 \equiv (1, 0, \ldots, 0)$ and $a^2 \equiv(0, 1, 0, \ldots, 0)$ be two alternatives that allocate all payoff to player $1$ and $2$, respectively. It follows that both $a^1$ and $a^2$ are core alternatives. Let $\sigma^1$ be the plan that specifies $a^2$ at all histories; for all $i\ne 1$, let $\sigma^i$ be the plan that specifies $a^1$ at all histories. Each $\sigma^i$ is a PCE, and $U_i(\emptyset|\sigma^i)=0$ for every $i\in \players$.
\end{proof}

\begin{lemma} \label{lemma:OPC}
	Suppose $\mathcal{U}$ is the set of PCE-supportable payoff profiles. For each player $i\in \players$, let $\underline{u}_i \equiv \min_{u\in \mathcal{U} }u_i$ be player $i$'s smallest possible payoff from PCEs. There is a stationary PCE with payoff profile $a$ if and only if for every coalition $C$ and alternative $a' \in E_C(a)$, there is a player $i\in C$ such that	
	\begin{equation} \label{inequality:simple game incentive constraints}
	(1-\delta)a'_i + \delta \underline{u}_i \le a_i
	\end{equation}
\end{lemma}
\begin{proof}
To see the ``only if'' direction, suppose \eqref{inequality:simple game incentive constraints} fails for some coalition $C$ and $a'\in E_C(a)$. In other words, suppose there exists a  coalition $C$ and alternative $a'$ such that $(1-\delta)a'_i + \delta \underline{u}_i > a_i \text{ for all } i\in C$. Towards a contradiction, suppose also that there exists a stationary PCE $\sigma$ that supports payoff $a$. Since $\sigma$ is a PCE, it follows that $U_i(h|\sigma ) \ge \underline{u}_i$ for every $i\in C$ and all $h\in \histories$. As a result, for every $i\in C$,
\begin{align*}
	(1-\delta)a'_i + \delta U_i(a',\{C\}|\sigma) \ge (1-\delta)a'_i + \delta \underline{u}_i > a_i.
\end{align*}
Moreover, since $\sigma$ is stationary, it always plays $a$ on path. The inequality above then implies that $(a',C)$ is a profitable block for coalition $C$ on path, contradicting $\sigma$ being a stationary PCE.

For the ``if'' direction, \eqref{inequality:simple game incentive constraints} implies that for every coalition $C$ and alternative $a' \in E_C(a)$, there exits a player $i{[a',C]}$ and a PCE $\sigma{[a',C]}$ such that 
\begin{equation} \label{inequality:aps enforcement on path}
(1-\delta)a'_{i[a',C]} + \delta U_{i[a',C]}\big(\emptyset\,|\,\sigma[a',C]\big) \le a_{i[a',C]}.
\end{equation}
Since the stage game exhibits default-independent power, by \cref{Theorem-Stationary}, we can without loss assume that each $\sigma{[a',C]}$ is a stationary PCE.

Consider a plan $\sigma^*$ that specifies $a$ on path, but switches to $\sigma[a',C]$ if coalition $C$ blocks to implement $a'$. Inequality \eqref{inequality:aps enforcement on path} implies that on path, no coalition can find profitably block. In addition, the fact that each $\sigma[a',C]$ is a PCE ensures that after any off-path history, no coalition can profitably block. Finally, $\sigma^*$ is also stationary since it is stationary on path, and each $\sigma[a',C]$ is also stationary. Therefore, $\sigma^*$ is a stationary PCE that supports payoff $a$.
\end{proof}

\paragraph{Proof of \cref{Proposition:Simple-Games-Payoffs}.} 

\noindent  \underline{Statement (a).} Set $\underline{\delta}=\frac{n-2}{n-1}$. By \cref{lemma:single-veto-zero-vectors}, there exist PCEs $\{\sigma^i: i\in \players \}$ satisfying $U_i(\emptyset|\sigma^i)=0$ for all $i\in \players$. It is straightforward to see that no players shared aligned payoffs in the stage game; in addition, 
no single player can form a winning coalition since the game is non-dictatorial. It follows that each player $i$'s individual minmax is $\im_i = 0$. Moreover, this minmax payoff is achieved by the PCE $\sigma^i$.

By \cref{lemma:OPC}, in order for a payoff profile $u$ to be supported by a stationary PCE, it is necessary and sufficient that for every winning coalition $C \in \mathcal{W}$, there exist no alternative $a'\in E_C(u)$ such that
\begin{equation} \label{inequality:simple game incentive constraints NTU}
(1-\delta)a'_i + \delta \cdot 0 = (1-\delta)a'_i > u_i \text{ for all }  i\in C.
\end{equation}
Note that the condition above is equivalent to $\sum_{i\in C} u_i \ge  1-\delta \text{ for every }  C\in\mathcal{W}$, since if $\sum_{i\in C} u_i < (1 - \delta) \cdot 1 $ for some coalition $C\in \mathcal{W}$, there would be a certain $a'\in E_C(u)$ representing a division of total payoff $1$ among players in $C$, such that (\ref{inequality:simple game incentive constraints NTU}) holds for every $i\in C$. It follows that a payoff profile $u$ is supportable by a stationary PCE if and only if $\sum_{i\in C} u_i \ge 1-\delta$ for every $C \in \mathcal{W}$. Finally, \cref{Theorem-Stationary} implies that this same set is also the set of PCE-supportable payoff profiles.\smallskip

\noindent  \underline{Statement (b).} Because a winning coalition can obtain the entire dollar by blocking, its minmax value is $1$. This statement then follows immediately from \Cref{Theorem-Secret}.\smallskip

\noindent  \underline{Statement (c).} Let $\hat{\mathcal{W}}$ denote the set of minimal winning coalitions. By definition, $\hat{\mathcal{W}}\subseteq \mathcal{W}$ so $\cap_{C\in \mathcal{W}} C \subseteq \cap_{C\in \hat{\mathcal{W}}} C $. Furthermore, $\cap_{C\in \hat{\mathcal{W}}} C \subseteq \cap_{C\in \mathcal{W}} C$, since otherwise there exists $i\in \cap_{C\in \hat{\mathcal{W}}} C$ and $\tilde{C}\in \mathcal{W}$ such that $i\notin \tilde{C}$, but this would lead to a contradiction since $\tilde{C}$ must contains a minimal winning coalition $\hat{C}$, and $i\in \hat{C}$. So $\cap_{C\in \hat{\mathcal{W}}} C = \cap_{C\in \mathcal{W}} C = D$. By \Cref{Theorem-Secret}, every $C\in \hat{\mathcal{W}}$ obtains total payoff $1$.  This implies the total payoff for players in $D$ is $1$.\qed
\end{document}